\documentclass[11pt,reqno]{amsart}
\usepackage{url,framed}
\usepackage[color=yellow!70]{todonotes}

\usepackage{latexsym}
\usepackage{dsfont}
\usepackage [latin1]{inputenc}

\usepackage{mathrsfs}

\usepackage{amsfonts}
\usepackage{amsmath}
\usepackage{amscd}
\usepackage{amssymb}
\usepackage{geometry}
\usepackage{color}
\usepackage{graphicx}%
\providecommand{\U}[1]{\protect\rule{.1in}{.1in}}
\newtheorem{theorem}{Theorem}

\newtheorem{definition}[theorem]{Definition}

\newtheorem{lemma}[theorem]{Lemma}

\newtheorem{proposition}[theorem]{Proposition}
\newtheorem{remark}[theorem]{Remark}

\makeatletter
\@addtoreset{figure}{section}
\def\thefigure{\thesection.\@arabic\c@figure}
\def\fps@figure{h, t}
\@addtoreset{table}{bsection}
\def\thetable{\thesection.\@arabic\c@table}
\def\fps@table{h, t}
\@addtoreset{equation}{section}

\makeatother

\begin{document}

\title{Solution properties of a 3D stochastic Euler fluid equation
%arising in geometric mechanics
}
\author{Dan Crisan, Franco Flandoli, Darryl D. Holm}

\address{DC, DDH: Department of Mathematics, Imperial College, London SW7 2AZ, UK.}
\address{FF: Department  of Applied  Mathematics, University  of Pisa, 
Via Bonanno 25B, 56126 Pisa, Italy.}

\date{3 October 2018}

\begin{abstract}
We prove local well-posedness in regular spaces and a Beale-Kato-Majda blow-up criterion for a recently derived stochastic model of the 3D Euler fluid equation for incompressible flow. This model describes incompressible fluid motions whose Lagrangian particle paths follow a stochastic process with cylindrical noise and also satisfy Newton's 2nd Law in every Lagrangian domain. 

%This model contains a geometric form of multiplicative Stratonovich noise which involves both the solution and its spatial gradient. 

\end{abstract}. 

\maketitle

%\tableofcontents

\newpage

\section{Introduction}

The present paper shows that two important analytical properties of deterministic Euler fluid dynamics in three dimensions possess close counterparts in the stochastic Euler fluid model introduced in \cite{Holm2015}. 
The first of these analytical properties is the local-in-time existence and uniqueness of deterministic Euler fluid flows. The second property is a criterion for blow-up in finite time due to Beale, Kato and Majda \cite{BKM1984}. For a historical review of these two fundamental analytical properties for deterministic Euler fluid dynamics, see, e.g., \cite{Gi2008}. 
We believe this fidelity of the stochastic model of \cite{Holm2015} investigated here 
with the analytical properties of the deterministic case bodes well for
the potential use of this model in, e.g., uncertainty quantification of either observed, or numerically simulated fluid flows.
The need and inspiration for such a model can be illustrated, for example, by examining 
data from satellite observations collected in the National Oceanic and Atmospheric
Administration (NOAA) ``Global Drifter Program'', a compilation of which is
shown in Figure \ref{fig:globaldrift}.

\begin{figure}[h]
\centering
%\subfigure[Surface drifter trajectories]
{\includegraphics[width=1.0\textwidth]{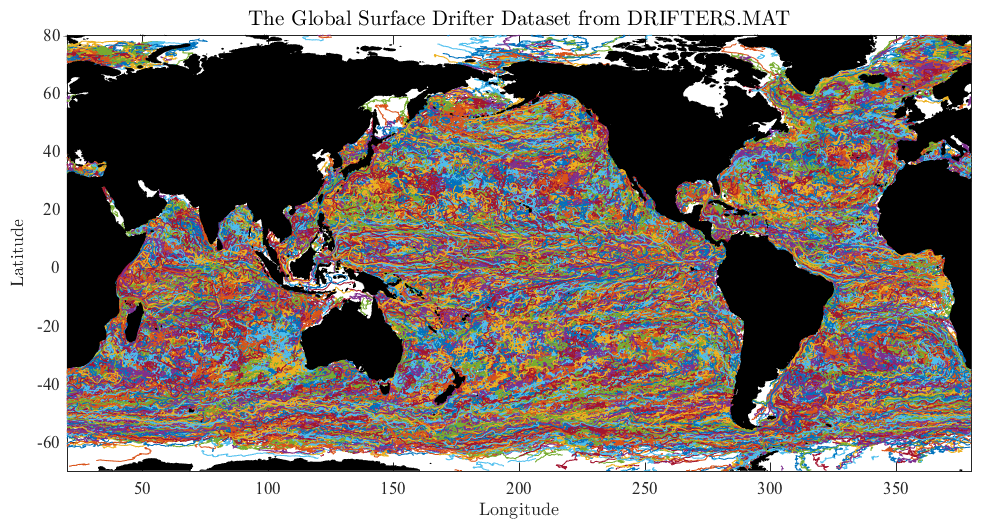} \label{fig:globaldrift}%
}\caption{This figure shows latitude and longitude of Lagrangian trajectories of drifters 
on the ocean surface driven by the wind and ocean currents, as compiled from satellite observations 
by the National Oceanic and Atmospheric Administration Global Drifter Program. Each colour
corresponds to a different drifter, see \cite{Lilly2017}. Upon looking carefully at 
the individual Lagrangian paths in this figure, one sees that each of them evolves as a mean drift flow,
composed with an erratic flow comprising rapid fluctuations around the mean. }%
\end{figure}

Figure \ref{fig:globaldrift} (courtesy of \cite{Lilly2017}) displays the global array of surface
drifter displacement trajectories from the National Oceanic and Atmospheric
Administration's ``Global Drifter Program''
(\url{www.aoml.noaa.gov/phod/dac}). In total, more than 10,000 drifters have
been deployed since 1979, representing nearly 30 million data points of
positions along the Lagrangian paths of the drifters at six-hour intervals.
This large spatiotemporal data set is a major source of information regarding
ocean circulation, which in turn is an important component of the global
climate system. For a recent discussion, see for example \cite{Sy-etal2016}.
This data set of spatiotemporal observations from satellites of the spatial
paths of objects drifting near the surface of the ocean provides inspiration
for further development of data-driven stochastic models of fluid dynamics of
the type discussed in the present paper. 

Inspired by this drifter data, the present paper investigates the existence, uniqueness and singularity properties of a recently derived stochastic model of the Euler fluid equations for
incompressible flow \cite{Holm2015} that is consistent with this data. For this purpose, we combine methods from
geometric mechanics, functional analysis and stochastic analysis. In the model
under investigation, one assumes that
the Lagrangian particle paths in the fluid motion $x_{t} = \eta_{t}(X)$ with
initial position $X\in\mathbb{R}^{3}$ each follow a Stratonovich stochastic
process given by
\begin{align}
\mathsf{d}\eta_{t}(X) = u(\eta_{t}(X),t)dt + \sum_{i} \xi_{i}(\eta_{t}(X))
\circ dB^{i}_{t} \,. \label{StochProc-intro}%
\end{align}
%%
%\begin{align}
%\mathsf{d}\eta_{t}(X) 
%= u(\eta_{t}(X),t)dt  
%+ \int_0^t K(\eta_{t-s}(X),\eta_{s}(X))\,ds\,dt
%+ \sum_{i} \xi_{i}(\eta_{t}(X))\,.
%\circ dB^{i}_{t} \,, \label{StochProc-intro}%
%\end{align}
%
%\todo[inline]{What is K, Dan? Why is this equation \eqref{StochProc-intro} $different from the one above it?}
%
{{}This approach immediately introduces the issue of spatial correlations. 

In particular,} an important feature of the data in Figure \ref{fig:globaldrift} is that the ocean currents show up as persistent spatial correlations, easily recognised visually as spatial regions in which the colours representing individual paths tend to concentrate. To capture this feature, we transform the Lagrangian trajectory description \eqref{StochProc-intro} into the spatial representation of the Eulerian transport velocity given by the Stratonovich stochastic vector field,
\begin{align}
\mathsf{dy}_{t}(x) = u(x,t)dt + \sum_{i} \xi_{i}(x) \circ dB^{i}_{t} =
\mathsf{d}\eta_{t}\, \eta_{t}^{-1}(x) \,. \label{StochVF-intro}%
\end{align}
In equations \eqref{StochProc-intro} and \eqref{StochVF-intro}, the $B_t^{i}$ with $i\in\mathbb{N}$ are scalar independent Brownian motions, and the $\xi_{i}(x)$ represent the spatial correlations which may be obtained as eigenvectors 
of the two-point velocity-velocity correlation matrix $C_{ij}(x,y)$, $i,j=1,2,\dots,N$, as an integral operator. Namely, 
\begin{align}
\sum_j \int C_{ij}(x,y) \xi_{j}(y)dy=\lambda \xi_{i}(x)\,. 
  \label{xi-eigenvectors-intro}%
\end{align}
These correlation eigenvectors exhibit a spectrum of spatial scales for the trajectories of
the drifters, indicating the variety of spatiotemporal scales in the evolution
of the ocean currents which transport the drifters. This feature of the data
is worthy of further study. In what follows, we will assume that the velocity correlation 
eigenvectors $\xi_i(x)$ with $i=1,\dots,N$ have been determined by reliable data assimilation procedures, so we may take them to be prescribed, divergence-free, three-dimensional vector functions. 
For explicit examples of the process of determining the $\xi_{i}(x)$ eigenvectors at coarse resolution 
from finely resolved numerical simulations, see \cite{CoCrHoShWe2018a,CoCrHoShWe2018b}. 
{{}For an extension of this method to include non-stationary correlation statistics, see \cite{GBHo2017}.\\

A rigorous analysis of the stochastic process $\eta_{t}$ in \eqref{StochProc-intro} is under way by the authors. Following from classical results (e.g., \cite{Kunita1,Kunita2}), we show in a forthcoming paper that $\eta_{t}$ is a temporally stochastic curve
on the manifold of smooth invertible maps with smooth inverses (i.e., diffeomorphisms).} Thus, although the time dependence of $\eta_{t}$ in \eqref{StochProc-intro} is not differentiable, its spatial
dependence is smooth. The stochastic process $\mathsf{d}\eta_{t}(X)$ in
\eqref{StochProc-intro} is also the pullback by the diffeomorphism  $\eta_{t}$  of the
stochastic vector field $\mathsf{dy}_{t}(x)$ in \eqref{StochVF-intro}. That
is, $\eta^{*}_{t}\mathsf{dy}_{t}(x) = \mathsf{d}\eta_{t}(X)$, see e.g. \cite{Holm2015} for details. 
Conversely, the stochastic vector field in \eqref{StochVF-intro} is the Eulerian
representation in fixed spatial coordinates $x$ of the stochastic process in
\eqref{StochProc-intro} for the Lagrangian fluid parcel paths, labelled by
their Lagrangian coordinates $X$.

The expression for the Lagrangian trajectories in equation \eqref{StochProc-intro} is clearly in accord with the observed behaviour of the Lagrangian trajectories displayed in Figure \ref{fig:globaldrift}. Moreover, the expression \eqref{StochVF-intro} for the corresponding Eulerian transport velocity has been derived recently in \cite{CoGoHo2017}  by using multi-time homogenisation methods for Lagrangian trajectories corresponding to solutions of the deterministic Euler equations, in the asymptotic limit of time-scale separation between the mean and fluctuating flow. In particular, the fluctuating dynamics in the second term in \eqref{StochVF-intro} has been shown in \cite{CoGoHo2017} to affect the mean flow. Thus, beyond being potentially useful as a means of uncertainty quantification, the decomposition in \eqref{StochVF-intro} represents a bona fide decomposition of the Eulerian fluid velocity into mean plus fluctuating components.

{{} The approach of incorporating uncertainties in incompressible fluid motion via stochastic Lagrangian fluid trajectories as  in equation \eqref{StochProc-intro} has several precedents, including for example, \cite{Brez1991,MR2004,Me2014}. However, the Eulerian fluid representation in \eqref{StochVF-intro} will lead us next to a stochastic partial differential equation (SPDE)  for the Eulerian drift velocity $u$ driven by \emph{cylindrical noise} represented by the Stratonovich term in \eqref{StochVF-intro} which differs from the Eulerian equations treated in these precedents.} For detailed discussions of SPDE with cylindrical noise, see \cite{DaPratoZab,Pa2007,PrevRoeckner,Sc1988}. 

\medskip\paragraph{\bf  Stochastic Euler fluid equations.}
As shown in \cite{Holm2015} via Hamilton's principle and rederived via Newton's Law in Appendix \ref{sec-derivation-recap} of the present paper, the stochastic Euler fluid equations we shall study in this paper may be represented in Kelvin circulation theorem form, as
\begin{align}
\mathsf{d}\oint_{c(t)} {v_{j}}(x,t) \,dx^{j} = \oint_{c(t)} \rho^{-1}F_{j}
\,dx^{j} \,, \label{Kel-stoch-intro}%
\end{align}
in which the closed loop $c(t)$ follows the Lagrangian stochastic process in \eqref{StochProc-intro}, which means it moves with stochastic Eulerian fluid velocity $\mathsf{dy}_{t}$ in \eqref{StochVF-intro}. In Kelvin's circulation theorem \eqref{Kel-stoch-intro}, the mass density is denoted
as $\rho$, and $F_{j} $ denotes the $j$-th component of the force exerted on
the flow. In the present work, the mass density $\rho$ will be assumed to be
constant. Notice that the covariant vector with components $v_j(x,t)$ in the integrand of \eqref{Kel-stoch-intro} is not the transport velocity in \eqref{StochVF-intro}. Instead, $v_j(x,t)$ is the $j$-th component of the momentum per unit mass. In what follows, the force per unit mass $\rho^{-1}F_{j} = -\rho ^{-1}\partial_{j}p$ will be taken to be proportional to the pressure gradient.
For this force, the Kelvin loop integral in \eqref{Kel-stoch-intro} for the
stochastic Euler fluid case will be preserved in time for any material loop whose motion is governed by the Stratonovich stochastic process \eqref{StochProc-intro}.
{{} That is, equation \eqref{Kel-stoch-intro} implies, for every rectifiable loop $c\subset \mathbb{R}^{3}$, the momentum per unit mass $v_t$ has the property that for all $t\in [0,T]$, 
\begin{equation}\label{Kel-stoch-intro2}
   \oint_{\eta_t(c)} v_t \cdot dx =    \oint_{c} v_0 \cdot dx
   , \quad    \ a.s.
\end{equation}
pathwise Kelvin theorem \eqref{Kel-stoch-intro2} is reminiscent of the Constantin-Iyer Kelvin theorem in \cite{CI08} which has the beautifully simple implication that smooth Navier-Stokes solutions $u_t$ are characterized by the following statistical Kelvin theorem which holds for all loops $\Gamma\subset \mathbb{R}^{3}$,
\begin{equation}\label{CIkelvin}
\int_\Gamma u_t \cdot dx =\mathbb{E}\left[ \int_{A_t(\Gamma)} u_0\cdot dx\right],
\end{equation}
where $A_t$ is the back-to-labels map for a stochastic flow of a certain forward It\^o equation and $\mathbb{E}$ denotes expectation for that flow. Unlike the pathwise Kelvin theorem \eqref{Kel-stoch-intro2} which holds for solutions of the stochastic Euler fluid equations, the Constantin-Iyer Kelvin theorem in \eqref{CIkelvin} is completely deterministic; since, the fluid velocity $u_t$ is a solution of the Navier-Stokes equations. For more discussion of Kelvin circulation theorems for stochastic Euler fluid equations see \cite{DrHo2018}.

}

In the case of the stochastic Euler fluid treated here in Euclidean
coordinates, applying the Stokes theorem to the Kelvin loop integral in
\eqref{Kel-stoch-intro} yields the equation for $\omega=\mathrm{curl}\,{v}$
proposed in \cite{Holm2015}, as
\begin{equation}
\mathsf{d}\omega+ (\mathsf{dy}_{t}\cdot\nabla) \omega- (\omega\cdot\nabla)
\mathsf{dy}_{t} = 0 \,,\qquad\omega|_{t=0}=\omega_{0} \,,
\label{eq Euler Strat-intro}%
\end{equation}
where 
%we have set $v=u$ in equation \eqref{StochVF-intro} for $\mathsf{dy}%
%_{t}$ and 
the loop integral on the right hand side of \eqref{Kel-stoch-intro} vanishes for pressure forces with constant mass density.

\medskip\paragraph{\bf Main results.}

This paper shows that two well-known analytical properties of the
deterministic 3D Euler fluid equations are preserved under the stochastic
modification in \eqref{eq Euler Strat-intro} we study here. First, the 3D
stochastic Euler fluid vorticity equation \eqref{eq Euler Strat-intro} is
locally well-posed in the sense that it possesses local in time existence and
uniqueness of solutions, for initial vorticity in the space $W^{2,2}%
(\mathbb{R}^{3})$ \cite{EbMa1970}. 
See Lichtenstein \cite{Li1925} as mentioned in \cite{FrVi2014} for a
historical precedent for local existence and uniqueness for the Euler fluid
equations. Second, the vorticity equations \eqref{eq Euler Strat-intro} also
possess a Beale-Kato-Majda (BKM) criterion for blow up which is identical to
the one proved for the deterministic Euler fluid equations in \cite{BKM1984}.

\begin{theorem}
[Existence and uniqueness]\label{main theorem-intro} Given initial vorticity
$\omega_{0}\in W^{2,2}\left(  \mathbb{T}^{3},\mathbb{R}^{3}\right)  $, there
exists a local solution in $W^{2,2}$ of the stochastic 3D Euler
equations \eqref{eq Euler Strat-intro}. Namely, if $\omega^{\left(  1\right)
},\omega^{\left(  2\right)  }:\Xi\times\left[  0,\tau\right]  \times
\mathbb{T}^{3}\rightarrow\mathbb{R}^{3}$ are two solutions defined up to the
same stopping time $\tau>0$, then $\omega^{\left(  1\right)  }=\omega^{\left(
2\right)  }$.
\end{theorem}

Our result corresponding to the celebrated Beale-Kato-Majda characterization
of blow-up \cite{BKM1984} is stated in the following.

\begin{theorem}
[Beale-Kato-Majda criterion for blow up]\label{BKM-intro}Given initial
$\omega_{0}\in W^{2,2}\left(  \mathbb{T}^{3},\mathbb{R}^{3}\right)  $, there
exists a stopping time $\tau_{\max}:\Xi\rightarrow\left[  0,\infty\right]  $
and a process $\omega:\Xi\times\lbrack0,\tau_{\max})\times\mathbb{T}%
^{3}\rightarrow\mathbb{R}^{3}$ with the following properties:

i) ($\omega$ is a solution) The process $t\mapsto\omega\left(  t\wedge
\tau_{\max},\cdot\right)  $ has trajectories that are almost surely in the
class $C\left(  [0,\tau_{\max});W^{2,2}\left(  \mathbb{T}^{3};\mathbb{R}%
^{3}\right)  \right)  $ and equation \eqref{eq Euler Strat-intro} holds as an
identity in $L^{2}\left(  \mathbb{T}^{3};\mathbb{R}^{3}\right)  $. In
addition, $\tau_{\max}$ is the largest stopping time with this property; and

ii) (Beale-Kato-Majda criterion \cite{BKM1984}) If $\tau_{\max}<\infty$, then%
\[
\int_{0}^{\tau_{\max}}\left\Vert \omega\left(  t\right)  \right\Vert _{\infty
}dt=+\infty
\]
and, in particular, $\lim\sup_{t\uparrow\tau_{\max}}\left\Vert \omega\left(
t\right)  \right\Vert _{\infty}=+\infty$.
\end{theorem}

%\newpage

\medskip\paragraph{\bf Plan of the paper}

$\,$

\begin{itemize}

\item Section \ref{sect assum + main results} discusses our assumptions and
summarizes the main results of the paper.

\begin{itemize}
\item Subsection \ref{sec-setting} formulates our objectives and sets the notation.

\item Subsection \ref{sect assumptions} discusses the cylindrical noise
properties of \eqref{StochProc-intro} and provides basic bounds on the Lie
derivatives needed in proving the main analytical results.

\item Subsection \ref{sect main result} provides additional definitions needed
in the context of explaining the main results of the paper.
\end{itemize}

\item Section \ref{sect proofs main results} provides proofs of the main results

\begin{itemize}
\item Subsections \ref{sect uniqueness new} and \ref{sect uniqueness W22}
prove the uniqueness properties needed for establishing Theorem
\ref{main theorem-intro}.

\item Subsection \ref{sect cut off} introduces a cut-off function which is
instrumental in the proof of the BKM theorem for the stochastic Euler
equations given in Subsection \ref{sect BKM proof}.

%\item Subsection \ref{sect cut off} proves global existence of solutions of a truncated model which introduces a cut-off and a viscous approximation which regularizes the solutions enough for them to be treated as
%classical. In this approximation, a tightness argument allows a key step in
%the proof of the BKM theorem for the stochastic Euler equations to be taken.
%This step turns out to hold independently of the value of the viscosity, which
%is then set to zero.

\end{itemize}

\item Section \ref{sect tech results} summarizes the proofs of several key
technical results which are summoned in establishing Theorem
\ref{main theorem-intro} and Theorem \ref{BKM-intro}.

\begin{itemize}
\item Subsection \ref{sect frac So reg} discusses fractional Sobolev
regularity in time.

\item Subsection \ref{sect bounds} provides the \textit{a priori} bounds
needed to prove estimate (\ref{bound 1}).

\item Subsection \ref{sect uniform in time} proves the bounds needed to
complete the proof that the estimate (\ref{bound 1}) is uniform in time.

\item Subsection \ref{sect basic bounds} establishes the key estimates for the
bounds involving Lie derivatives that are needed in the proofs.
\end{itemize}

\item Appendix \ref{sec-derivation-recap} provides a new derivation of the
stochastic Euler equations introduced in \cite{Holm2015} from the viewpoint
of Newton's 2nd Law and derives the corresponding Kelvin circulation theorem. 
The deterministic
(resp. stochastic) equations of motion are derived using the pullback of
Newton's 2nd Law by the deterministic (resp. stochastic) diffeomorphism
describing the Lagrange-to-Euler map. The Kelvin circulation theorems for both
cases are then derived from their corresponding Newtonian 2nd Laws. The
importance of the distinction between transport velocity and transported
momentum is emphasized in Appendix \ref{sec-derivation-recap} for both the 
deterministic and stochastic Newton's Laws and Kelvin's circulation theorems.
\end{itemize}
%\newpage

%%%%%%%%%%%%%%%%%%%%%%%%%%%%%%%%%%%%%%
\section{Assumptions and main results}

\label{sect assum + main results}

\subsection{Formulating objectives and setting notation\label{sec-setting}}

Our aim from now on will be to prove local-in-time existence and uniqueness of
regular solutions of the stochastic Euler vorticity equation
%\eqref{Euler-mot-eqn-calc-stoch}%
\begin{equation}
d\omega+\mathcal{L}_{v}\omega dt+\sum_{k=1}^{\infty}\mathcal{L}_{\xi_{k}%
}\omega\circ dB_{t}^{k}=0,\qquad\omega|_{t=0}=\omega_{0},
\label{eq Euler Strat}%
\end{equation}
which was proposed in \cite{Holm2015}. Here, $\mathcal{L}_{v}\omega$ (resp.
$\mathcal{L}_{\xi_{k}}\omega$) denotes the Lie derivative with respect to the
vector fields $v$ (resp. $\xi_{k}$) as in \eqref{Euler-mot-eqn-calc} applied
to the vorticity vector field. In particular,
\begin{equation}
\mathcal{L}_{\xi_{k}}\omega=(\xi_{k}\cdot\nabla)\omega-(\omega\cdot\nabla
)\xi_{k}=\left[  \,\xi_{k}\,,\,\omega\,\right]  . \label{eq Lie brkt}%
\end{equation}

A natural question is whether we should sum only over a finite number of terms
or, on the contrary, it is important to have an infinite sum, and not only for
generality. An important remark is that a finite number of eigenvectors arises
in the relevant case associated to a ``data driven" model based on what is
resolvable in either numerics of observations; and it would simplify some
technical issues (we do not have to assume (\ref{assump}) below). However, an
infinite sum could be of interest in regularization-by-noise
investigations:\ see an example in \cite{Delarue} (easier than 3D Euler
equations) where a singularity is prevented by an infinite dimensional noise.
However, it is also true that in some cases a finite dimensional noise is
sufficient also for regularization by noise, see examples in \cite{FlaGubPri1}%
, \cite{FlaGubPri2}, \cite{FlMaNe2014}.

As mentioned in Remark \ref{u-v Diff}, for the case of the \emph{Euler fluid
equations treated here in Cartesian $\mathbb{R}^{3}$ coordinates}, the two
velocities denoted $u$ and $v$ in the previous section may be taken to be
identical vectors for the case at hand in $\mathbb{R}^{3}$. Consequently, for the remainder of the present work, in a slight abuse of notation, we simply let $v$ denote the both fluid velocity and the momentum per unit mass. Then
$\omega=\operatorname{curl}v$ is the vorticity, and $\xi_{k}$ comprise $N$
divergence-free prescribed vector fields, subject to the assumptions stated
below. The  processes $B^{k}$ with $k\in\mathbb{N}$ are scalar independent Brownian
motions. The result we present next will extend the known analogous result for
deterministic Euler equations to the stochastic case.

To simplify some of the arguments, we will work on a torus $\mathbb{T}^{3}=\mathbb{R}^{3}/\mathbb{Z}^{3}$. However, the results should also hold in the
should also hold in the full space, $\mathbb{R}^{3}$.

The stochastic Euler vorticity equation \eqref{eq Euler Strat} above is stated
in Stratonovich form. The corresponding It\^{o} form is
\begin{equation}
d\omega+\mathcal{L}_{v}\omega dt+\sum_{k=1}^{\infty}\mathcal{L}_{\xi_{k}%
}\,\omega dB_{t}^{k}=\frac{1}{2}\sum_{k=1}^{\infty}\mathcal{L}_{\xi_{k}}%
^{2}\omega\ dt\,,\qquad\omega|_{t=0}=\omega_{0}\,, \label{eq Euler Ito}%
\end{equation}
where we write
\[
\mathcal{L}_{\xi_{k}}^{2}\omega=\mathcal{L}_{\xi_{k}}(\mathcal{L}_{\xi_{k}%
}\omega)=\left[  \xi_{k}\,,[\xi_{k}\,,\,\omega]\right]  \,,
\]
for the double Lie bracket of the divergence-free vector field $\xi_{k}$ with
the vorticity vector field $\omega$. Indeed, let us recall that Stratonovich
integral is equal to It\^{o} integral plus one-half of the corresponding cross-variation process:\footnote{The subscript $t$ on the square brackets distinguishes
between the cross-variation process and Lie bracket of vector fields. To avoid
confusion between these two uses of the square bracket, we will denote the Lie
bracket operation $[\,\xi_{k}\,,\,\cdot\,]$ by the symbol $\mathcal{L}%
_{\xi_{k}}\cdot $, as in equation \eqref{eq Lie brkt}.}
\[
\int_{0}^{t}\mathcal{L}_{\xi_{k}}\omega_{s}\circ dB_{s}^{k}=\int_{0}%
^{t}\mathcal{L}_{\xi_{k}}\omega_{s}dB_{s}^{k}+\frac{1}{2}\left[
\mathcal{L}_{\xi_{k}}\omega,B^{k}\right]  _{t}.
\]
By the linearity and the space-independence of $B^{k}$,\ $\left[  \mathcal{L}_{\xi
_{k}}\omega,B^{k}\right]  _{t}=\mathcal{L}_{\xi_{k}}\left[  \omega
,B^{k}\right]  _{t}$. As $\omega_{t}$ has the form $d\omega_{t}=a_{t}%
dt+\sum_{h}b_{t}^{h}\circ dB_{t}^{h}$, where $B^{h}$ are independent, the cross-variation process $\left[  \omega,B^{k}\right]  _{t}$ is given by
\[
\left[  \omega,B^{k}\right]  _{t}=\int_{0}^{t}b_{s}^{k}ds.
\]
In our case $b_{s}^{k}=-\mathcal{L}_{\xi_{k}}\omega_{s}$, hence%
\[
\left[  \omega,B^{k}\right]  _{t}=-\int_{0}^{t}\mathcal{L}_{\xi_{k}}\omega
_{s}ds.
\]
Finally%
\[
\mathcal{L}_{\xi_{k}}\left[  \omega,B^{k}\right]  _{t}=-\int_{0}%
^{t}\mathcal{L}_{\xi_{k}}^{2}\omega_{s}ds
\]
and therefore, in differential form,%
\[
\mathcal{L}_{\xi_{k}}\omega_{t}\circ dB_{t}^{k}=\mathcal{L}_{\xi_{k}}%
\omega_{t}\,dB_{t}^{k}-\frac{1}{2}\mathcal{L}_{\xi_{k}}^{2}\omega_{t}dt.
\]

Among different possible strategies to study equation (\ref{eq Euler Ito}),
some of them based on stochastic flows, we present here the extension to the
stochastic case of a classical PDE proof, see for instance \cite{Kato},
\cite{PL Lions}, \cite{MajdaBert}.

The proof is based on \textit{a priori} estimates in high order Sobolev
spaces. The deterministic classical result proves well-posedness in the space
$\omega\left(  t\right)  \in W^{3/2+\epsilon,2}\left(  \mathbb{T}%
^{3};\mathbb{R}^{3}\right)  $, for some $\epsilon>0$, when $\omega_{0}$
belongs to the same space. Here we simplify (due to a number of new very
non-trivial facts outlined below in section \ref{subsect difficulty}) and work
in the space $\omega\left(  t\right)  \in W^{2,2}\left(  \mathbb{T}%
^{3};\mathbb{R}^{3}\right)  $. Consequently, we may consider $\Delta
\omega\left(  t\right)  $ (to avoid fractional derivatives) and investigate
existence and uniqueness in the class of regularity $\Delta\omega\left(
t\right)  \in L^{2}\left(  \mathbb{T}^{3};\mathbb{R}^{3}\right)  $.

%\textbf{***(begin)} 
If $f,g\in L^{2}\left(  \mathbb{T}^{3};\mathbb{R}%
^{3}\right)  $, we write $\left\langle f,g\right\rangle =\int_{\mathbb{T}^{3}%
}f\left(  x\right)  \cdot g\left(  x\right)  dx$. We consider the basis of
$L^{2}\left(  \mathbb{T}^{3};\mathbb{C}\right)  $ of functions $\left\{
e^{2\pi i\xi\cdot x};\xi\in\mathbb{Z}^{3}\right\}  $ and, for every $f\in
L^{2}\left(  \mathbb{T}^{3};\mathbb{C}\right)  $, we introduce the Fourier
coefficients $\widehat{f}\left(  \xi\right)  =\int_{\mathbb{T}^{3}}e^{-2\pi
i\xi\cdot x}f\left(  x\right)  dx$, $\xi\in\mathbb{Z}^{3}$; Parseval identity
states that $\int_{\mathbb{T}^{3}}\left\vert f\left(  x\right)  \right\vert
^{2}dx=\sum_{\xi\in\mathbb{Z}^{3}}\left\vert \widehat{f}\left(  \xi\right)
\right\vert ^{2}$. If $v\in L^{2}\left(  \mathbb{T}^{3};\mathbb{R}^{3}\right)
$ is a vector field with components $v_{i}$, $i=1,2,3$, we write
$\widehat{v}\left(  \xi\right)  =\int_{\mathbb{T}^{3}}e^{2\pi i\xi\cdot
x}v\left(  x\right)  dx$ and we may easily check using the components that
have $\int_{\mathbb{T}^{3}}\left\vert v\left(  x\right)  \right\vert
^{2}dx=\sum_{\xi\in\mathbb{Z}^{3}}\left\vert \widehat{v}\left(  \xi\right)
\right\vert ^{2}$. Since functions which are partial derivatives of other
functions, on the torus, must have zero average, we shall always restrict
ourselves to functions $f\in L^{2}\left(  \mathbb{T}^{3};\mathbb{C}\right)  $
such that $\int_{\mathbb{T}^{3}}f\left(  x\right)  dx=0$. In this case
$\widehat{f}\left(  0\right)  =0$ and the term with $\xi=0$ does not appear in
the sums above.

We introduce, for every $s\geq0$, the fractional Sobolev space $W^{s,2}\left(
\mathbb{T}^{3};\mathbb{C}\right)  $ of all $f\in L^{2}\left(  \mathbb{T}%
^{3};\mathbb{C}\right)  $ such that
\[
\left\Vert f\right\Vert _{W^{s,2}}^{2}:=\sum_{\xi\in\mathbb{Z}^{3}%
\backslash\left\{  0\right\}  }\left\vert \xi\right\vert ^{2s}\left\vert
\widehat{f}\left(  \xi\right)  \right\vert ^{2}<\infty.
\]
As stated above, we are assuming zero average functions, hence we have
excluded $\xi=0$. We denote by $W_{\sigma}^{s,2}\left(  \mathbb{T}%
^{3},\mathbb{R}^{3}\right)  $ the space of all zero mean divergence free
(divergence in the sense of distribution) vector fields $v\in L^{2}\left(
\mathbb{T}^{3};\mathbb{R}^{3}\right)  $ such that all components $v_{i}$,
$i=1,2,3$, belong to $W^{s,2}\left(  \mathbb{T}^{3};\mathbb{C}\right)  $. For
a vector field $v\in W_{\sigma}^{s,2}\left(  \mathbb{T}^{3},\mathbb{R}%
^{3}\right)  $ the norm $\left\Vert v\right\Vert _{W^{s,2}}$ is defined by the
identity $\left\Vert v\right\Vert _{W_{\sigma}^{s,2}}^{2}=\sum_{i=1}%
^{3}\left\Vert v_{i}\right\Vert _{W^{s,2}}^{2}$, where $\left\Vert
v_{i}\right\Vert _{W^{s,2}}^{2}$ is defined above. We thus have again
$\left\Vert v\right\Vert _{W_{\sigma}^{s,2}}^{2}:=\sum_{\xi\in\mathbb{Z}%
^{3}\backslash\left\{  0\right\}  }\left\vert \xi\right\vert ^{2s}\left\vert
\widehat{v}\left(  \xi\right)  \right\vert ^{2}$. For $f\in W^{s,2}\left(
\mathbb{T}^{3};\mathbb{C}\right)  $, we denote by {$\left(  -\Delta\right)
^{s/2}f$ the function of }$L^{2}\left(  \mathbb{T}^{3};\mathbb{C}\right)  $
with Fourier coefficients $\left\vert \xi\right\vert ^{s}\widehat{f}\left(
\xi\right)  $. {Similarly, we write $-\Delta^{-1}f$ \ for the function having
Fourier coefficients $\left\vert \xi\right\vert ^{-2}\widehat{f}\left(
\xi\right)  $. We use the same notations for vector fields, meaning that the
operations are made componentwise.}

The Biot-Savart operator is the reconstruction of a zero mean divergence free vector
field $u$ from a divergence free vector field $\omega$ such that
$\operatorname{curl}u=\omega$. On the torus it is given by
$u=-\operatorname{curl}\Delta^{-1}\omega$. In Fourier components, it is given
by $\widehat{u}\left(  \xi\right)  =\left\vert \xi\right\vert ^{-2}\xi
\times\widehat{\omega}\left(  \xi\right)  $. We have the following well-known
result: for all $s\geq0$
\begin{equation}
\left\Vert u\right\Vert _{W_{\sigma}^{s+1,2}}\leq\left\Vert \omega\right\Vert
_{W_{\sigma}^{s,2}}.\label{Biot Savart}%
\end{equation}
Indeed, using the definition given above of $\left\Vert u\right\Vert
_{W_{\sigma}^{s+1,2}}^{2}$, the formula which relates $\widehat{u}\left(
\xi\right)  $ to $\widehat{\omega}\left(  \xi\right)  $ and the rule
$\left\vert a\times b\right\vert \leq\left\vert a\right\vert \left\vert
b\right\vert $, we get
\[
\left\Vert u\right\Vert _{W_{\sigma}^{s+1,2}}^{2}=\sum_{\xi\in\mathbb{Z}%
^{3}\backslash\left\{  0\right\}  }\left\vert \xi\right\vert ^{2s+2}\left\vert
\widehat{u}\left(  \xi\right)  \right\vert ^{2}=\sum_{\xi\in\mathbb{Z}%
^{3}\backslash\left\{  0\right\}  }\left\vert \xi\right\vert ^{2s+2}\left\vert
\xi\right\vert ^{-4}\left\vert \xi\times\widehat{\omega}\left(  \xi\right)
\right\vert ^{2}\leq\sum_{\xi\in\mathbb{Z}^{3}\backslash\left\{  0\right\}
}\left\vert \xi\right\vert ^{2s+2}\left\vert \xi\right\vert ^{-2}\left\vert
\widehat{\omega}\left(  \xi\right)  \right\vert ^{2}%
\]
and the latter is precisely equal to $\left\Vert \omega\right\Vert
_{W_{\sigma}^{s,2}}$, by the definition above. %\textbf{***(end)}

We shall denote the dual operator of the Lie derivative $\mathcal{L}_{\alpha}$ of a
vector field as $\mathcal{L}_{\alpha}^{\ast}$, defined by the identity%
\[
\left\langle \mathcal{L}_{\alpha}^{\ast} \beta,\gamma\right\rangle =\left\langle
\beta,\mathcal{L}_{\alpha}\gamma\right\rangle ,
\]
for all smooth vector fields $\alpha,\beta,\gamma$. When $\operatorname{div}\alpha=0$ the dual
Lie operator is given in vector components by%

\begin{equation}
\left(  \mathcal{L}_{\alpha}^{\ast}\gamma\right)  _{i} := - \sum_{j} \left(
\alpha^{j}\partial_{j} \gamma_{i} + \gamma_{j}\partial_{i}\alpha^{j}\right)  \,.
\label{def Lambda}%
\end{equation}

\subsection{Assumptions on $\left\{  \xi_{k}\right\}  $ and basic bounds on
Lie derivatives}

\label{sect assumptions}

We assume that the vector fields $\xi_{k}:\mathbb{T}^{3}\rightarrow
\mathbb{R}^{3}$ are of class $C^{4}$ and satisfy%
\begin{equation}
\left\Vert \sum_{k=1}^{\infty}\mathcal{L}_{\xi_{k}}^{2}f\right\Vert _{L^{2}%
}^{2}\leq C\left\Vert f\right\Vert _{W^{2,2}}^{2} \label{Hp 1}%
\end{equation}%
\begin{equation}
\sum_{k=1}^{\infty}\left\langle \mathcal{L}_{\xi_{k}}f,\mathcal{L}_{\xi_{k}%
}f\right\rangle \leq C\left\Vert f\right\Vert _{W^{2,2}}^{2} \label{Hp 2}%
\end{equation}
for all $f\in W^{2,2}\left(  \mathbb{T}^{3};\mathbb{R}^{3}\right)  $ and%
\begin{equation}
\sum_{k=1}^{\infty}\left\Vert \xi_{k}\right\Vert _{W^{3,2}}^{2}<\infty.
\label{Hp 3}%
\end{equation}
These properties will be used below, both to give a meaning to the stochastic
terms in the equation and to prove certain bounds. In addition, a recurrent
energy-type scheme in our proofs requires comparisons of quadratic variations
and Stratonovich corrections. Making these comparisons leads to sums of the
form $\left\langle \mathcal{L}_{\xi_{k}}^{2}f,f\right\rangle +\left\langle
\mathcal{L}_{\xi_{k}}f,\mathcal{L}_{\xi_{k}}f\right\rangle $. In dealing with
them, we have observed the validity of two striking bounds, which \textit{a
priori} may look surprising. They are:%
\begin{equation}
\left\langle \mathcal{L}_{\xi_{k}}^{2}f,f\right\rangle +\left\langle
\mathcal{L}_{\xi_{k}}f,\mathcal{L}_{\xi_{k}}f\right\rangle \leq C_{k}^{\left(
0\right)  }\left\Vert f\right\Vert _{L^{2}}^{2} \label{striking est 1}%
\end{equation}%
\begin{equation}
\left\langle \Delta\mathcal{L}_{\xi_{k}}^{2}f,\Delta f\right\rangle
+\left\langle \Delta\mathcal{L}_{\xi_{k}}f,\Delta\mathcal{L}_{\xi_{k}%
}f\right\rangle \leq C_{k}^{\left(  2\right)  }\left\Vert f\right\Vert
_{W^{2,2}}^{2}, \label{striking est 2}%
\end{equation}
for suitable constants $C_{k}^{\left(  0\right)  },C_{k}^{\left(  2\right)  }%
$. For these estimates to hold, the regularity of $f$ must be, respectively,
$W^{2,2}\left(  \mathbb{T}^{3};\mathbb{R}^{3}\right)  $ and $W^{4,2}\left(
\mathbb{T}^{3};\mathbb{R}^{3}\right)  $. The proofs of estimates
(\ref{striking est 1}) and (\ref{striking est 2}) are given in Section
\ref{sect basic bounds} below.\footnote{We thank Istvan Gy\"ongy and Nikolai Krylov for pointing out to us that estimates such as (\ref{striking est 1}) and (\ref{striking est 2}) hold in much more generality, see e.g. \cite{G1,G2,G3}. }

Concerning inequality (\ref{striking est 1}), it is clear that the second
order terms in $\left\langle \mathcal{L}_{\xi_{k}}^{2}f,f\right\rangle $ and
$\left\langle \mathcal{L}_{\xi_{k}}f,\mathcal{L}_{\xi_{k}}f\right\rangle $
will cancel. However, the cancellations among the first order terms are not so
obvious. Remarkably, though, these terms do cancel each other, so that only
the zero-order terms remain. Similar remarks apply to the other inequality.

In addition, we must assume
\begin{equation}
\sum_{k=1}^{\infty}C_{k}^{\left(  0\right)  }<\infty,\qquad\sum_{k=1}^{\infty
}C_{k}^{\left(  2\right)  }<\infty. \label{assump}%
\end{equation}
Because the constants $C_{k}^{\left(  i\right)  }$ are rather complicated, we
will not write them explicitly here. In the relevant case of a finite number
of $\xi_{k}$'s, there is obviously no need of this assumption. In the case of
infinitely many terms, see a sufficient condition in Remark
\ref{remark summability Ck} of Section \ref{sect basic bounds}.

\subsection{Statement of the main results}

\label{sect main result}

Let $\left\{  B^{k}\right\}  _{k\in\mathbb{N}}$ be a sequence of
independent Brownian motions on a filtered probability space $\left(
\Xi,\mathcal{F},\mathcal{F}_{t},P\right)  $. We do not use the most common
notation $\Omega$ for the probability space, since $\omega$ is the traditional
notation for the vorticity. Thus the elementary events will be denoted by
$\theta\in\Xi$. Let $\left\{  \xi_{k}\right\}  _{k\in\mathbb{N}}$ be a
sequence of vector fields, satisfying the assumptions of section
\ref{sect assumptions}. Consider equations (\ref{eq Euler Ito}) on
$[0,\infty)$.

\begin{definition}
[Local solution]\label{loc-soln} \label{def local solution} A local solution
in $W_{\sigma}^{2,2}$ of the stochastic 3D Euler equations (\ref{eq Euler Ito}%
) is given by a pair $\left(  \tau,\omega\right)  $ consisting of a stopping
time $\tau:\Xi\rightarrow\lbrack0,\infty)$ and a process $\omega:\Xi
\times\left[  0,\tau\right]  \times\mathbb{T}^{3}\rightarrow\mathbb{R}^{3}$
such that a.e. the trajectory is of class $C\left(  \left[  0,\tau\right]
;W_{\sigma}^{2,2}\left(  \mathbb{T}^{3};\mathbb{R}^{3}\right)  \right)  $,
$\omega\left(  t\wedge\tau,\cdot\right)  $, is adapted to $\left(
\mathcal{F}_{t}\right)  $, and equation (\ref{eq Euler Ito}) holds in the
usual integral sense; More precisely, for any bounded stopping time $\bar
{\tau}\leq\tau$
\begin{equation}
\omega_{\bar{\tau}}-\omega_{0}+\int_{0}^{\bar{\tau}}\mathcal{L}_{v}\omega
dt+\sum_{k=1}^{\infty}\int_{0}^{\bar{\tau}}\mathcal{L}_{\xi_{k}}\omega
dB_{t}^{k}=\frac{1}{2}\sum_{k=1}^{\infty}\int_{0}^{\bar{\tau}}\mathcal{L}%
_{\xi_{k}}^{2}\omega\ dt\, \label{integral eq Euler Ito}%
\end{equation}
holds as an identity in $L^{2}\left(  \mathbb{T}^{3};\mathbb{R}^{3}\right)  $.
\end{definition}

\begin{definition}
[Maximal solution]\label{def maximal solution} A maximal solution of
\eqref{eq Euler Ito} is given by a stopping time $\tau_{\max}:\Xi
\rightarrow\left[  0,\infty\right]  $ and a process $\omega:\Xi\times
\lbrack0,\tau_{\max})\times\mathbb{T}^{3}\rightarrow\mathbb{R}^{3}$ such that:
i)\ $P\left(  \tau_{\max}>0\right)  =1$, $\tau_{\max}=\lim_{n\rightarrow
\infty}\tau_{n}$ where $\tau_{n}$ is an increasing sequence of stopping times,
and ii)\ $\left(  \tau_{n},\omega\right)  $ is a local solution for every
$n\in\mathbb{N}$; In addition, $\tau_{\max}$ is the largest stopping time with
properties i) and ii). In other words, if $(\tau^{\prime},\omega^{\prime})$ is
another pair that satisfies i) and ii) and $\omega^{\prime}=\omega$ on
$[0,\tau^{\prime}\wedge\tau_{\max}),$ then, $\tau^{\prime}\leq\tau_{\max}$
$P$-almost surely.
\end{definition}

\begin{remark}
Due to assumptions (\ref{Hp 1}) and (\ref{Hp 2}) and the regularity of
$\omega$, the two terms related to the noise in equation (\ref{eq Euler Ito})
are well defined, as elements of $L^{2}\left(  \mathbb{T}^{3};\mathbb{R}%
^{3}\right)  $.
\end{remark}

\begin{remark}
\label{remark regul}Recall that, for every $\alpha\geq0$, $\omega\left(
t\right)  \in W^{\alpha,2}\left(  \mathbb{T}^{3};\mathbb{R}^{3}\right)  $
implies $v\left(  t\right)  \in W^{\alpha+1,2}\left(  \mathbb{T}%
^{3};\mathbb{R}^{3}\right)  $. Hence solutions in $W^{2,2}$ have paths such
that $v\in C\left(  \left[  0,\tau\right]  ;W^{3,2}\left(  \mathbb{T}%
^{3};\mathbb{R}^{3}\right)  \right)  $. Moreover, recall that $W^{\alpha
,2}\left(  \mathbb{T}^{3};\mathbb{R}^{3}\right)  \subset C\left(
\mathbb{T}^{3};\mathbb{R}^{3}\right)  $ for $\alpha>3/2$. Therefore
$\omega\cdot\nabla v\in C\left(  \left[  0,\tau\right]  \times\mathbb{T}%
^{3};\mathbb{R}^{3}\right)  $ and $v\cdot\nabla\omega$ is at least in
$C\left(  \left[  0,\tau\right]  ;W^{1,2}\left(  \mathbb{T}^{3};\mathbb{R}%
^{3}\right)  \right)  $, hence at least
\[
\left[  v,\omega\right]  \in C\left(  \left[  0,\tau\right]  ;W^{1,2}\left(
\mathbb{T}^{3};\mathbb{R}^{3}\right)  \right)  \qquad a.s.
\]
which explains why the term $\left[  v,\omega\right]  $ is in $L^{2}\left(
\mathbb{T}^{3};\mathbb{R}^{3}\right)  $. (Recall that Definition
\ref{loc-soln} instructs us to interpret equation (\ref{eq Euler Ito}) as an
identity in $L^{2}\left(  \mathbb{T}^{3};\mathbb{R}^{3}\right)  $.)
\end{remark}

\begin{remark}
\label{remark weak sol}If $\omega:\Xi\times\left[  0,\tau\right]
\times\mathbb{T}^{3}\rightarrow\mathbb{R}^{3}$ has the regularity properties
of Definition \ref{def local solution} and satisfies equation
(\ref{eq Euler Ito}) only in a weak sense, namely, for any bounded stopping
time $\bar{\tau}\leq\tau$
\[
\left\langle \omega\left(  \bar{\tau}\right)  ,\phi\right\rangle +\int%
_{0}^{\bar{\tau}}\left\langle \omega\left(  s\right)  ,\mathcal{L}_{v\left(
s\right)  }^{\ast}\phi\right\rangle ds+\sum_{k=1}^{\infty}\int_{0}^{\bar{\tau
}}\left\langle \omega\left(  s\right)  ,\mathcal{L}_{\xi_{k}}^{\ast}%
\phi\right\rangle B_{s}^{k}=\left\langle \omega_{0},\phi\right\rangle
+\frac{1}{2}\sum_{k=1}^{\infty}\int_{0}^{\bar{\tau}}\left\langle \omega\left(
s\right)  ,\mathcal{L}_{\xi_{k}}^{\ast}\mathcal{L}_{\xi_{k}}^{\ast}%
\phi\right\rangle \ ds
\]
for all $\phi\in C^{\infty}\left(  \mathbb{T}^{3};\mathbb{R}^{3}\right)  $,
then, by integration by parts, it satisfies equation (\ref{eq Euler Ito}) as
an identity in $L^{2}\left(  \mathbb{T}^{3};\mathbb{R}^{3}\right)  $.
\end{remark}

\begin{theorem}
\label{main theorem}Given $\omega_{0}\in W_{\sigma}^{2,2}\left(
\mathbb{T}^{3},\mathbb{R}^{3}\right)  $, there exists a maximal solution
$\left(  \tau_{\max},\omega\right)  $ of the stochastic 3D Euler equations
(\ref{eq Euler Ito}). Moreover if $\left(  \tau^{\prime},\omega^{\prime
}\right)  $ is another maximal solution of (\ref{eq Euler Ito}), then
necessarily $\tau_{\max}=\tau^{\prime}$ and $\omega=\omega^{\prime}$on
$[0,\tau_{\max})$.\ Moreover either $\tau_{\max}=\infty$ or $\lim
\sup_{t\uparrow\tau_{\max}}\left\Vert \omega\left(  t\right)  \right\Vert
_{W^{2,2}}=+\infty$.
\end{theorem}

In this paper, we will also prove a corresponding result to the celebrated
Beale-Kato-Majda criterion for blow-up of vorticity solutions of the
deterministic Euler fluid equations.

\begin{theorem}
\label{BKM}Given $\omega_{0}\in W_{\sigma}^{2,2}\left(  \mathbb{T}%
^{3},\mathbb{R}^{3}\right)  $, if $\tau_{\max}<\infty$, then%
\[
\int_{0}^{\tau_{\max}}\left\Vert \omega\left(  t\right)  \right\Vert _{\infty
}dt=+\infty\,.
\]
In particular, $\lim\sup_{t\uparrow\tau_{\max}}\left\Vert \omega\left(
t\right)  \right\Vert _{\infty}=+\infty$ almost surely.\end{theorem}

\begin{remark}
As in the deterministic case, Theorem \ref{BKM} can be used as a criterion for
testing whether a given numerical simulation has shown finite-time blow up.
Following \cite{Gi2008}, the classical Beale-Kato-Majda theorem implies that
algebraic singularities of the type $\|\omega\|_{\infty} \geq(t^{*} - t)^{-p}$
must have $p \geq1$. In our paper, we have shown that a corresponding BKM
result also applies for the stochastic Euler fluid equations; hence, the same
criterion applies here. In \cite{CoFeMa1996}, the $L^{\infty}$ condition in
the BKM theorem was reduced to $L^{p}$, for finite $p$, at the price of
imposing constraints on the direction of vorticity. We hope to obtain a
similar $L^{\infty}$  result for the stochastic 3D-Euler equation in future work.
\end{remark}

In sections \ref{sect uniqueness new} and \ref{uniq2} we prove uniqueness. The rest of
the paper will be devoted to proving {{}local} existence of the solution and Theorem \ref{BKM}.

\section{Proofs of the main results\label{sect proofs main results}}

\subsection{Local uniqueness of the solution of the stochastic 3D Euler
equation\label{sect uniqueness new}}

In the following proposition, we prove that any two local solutions of the
stochastic 3D Euler equation (\ref{eq Euler Ito}) that are defined up to the
\emph{same} stopping time must coincide. The proof hinges on the bound
(\ref{striking est 1}) and the assumption (\ref{assump}).

\begin{proposition}
\label{uniqueness proposition} Let $\tau$ be a stopping time and
$\omega^{\left(  1\right)  },\omega^{\left(  2\right)  }:[0,\tau
)\times\mathbb{T}^{3}\rightarrow\mathbb{R}^{3}$ be two solutions with paths of
class $C\left(  [0,\tau);W_{\sigma}^{2,2}\left(  \mathbb{T}^{3};\mathbb{R}%
^{3}\right)  \right)  $ that satisfy the stochastic 3D Euler equation
(\ref{eq Euler Ito}). Then $\omega^{\left(  1\right)  }=\omega^{\left(
2\right)  }$ on $[0,\tau).$
\end{proposition}

\begin{proof}
We have that\footnote{The following identity and all the subsequent ones hold
as identities in $L^{2}\left(  \mathbb{T}^{3};\mathbb{R}^{3}\right)  $ and
represent the differential form of their integral version in the same way as
equation (\ref{eq Euler Ito}) is the differential form of
(\ref{integral eq Euler Ito}).}
\[
d\omega^{\left(  i\right)  }+\mathcal{L}_{v^{\left(  i\right)  }}%
\omega^{\left(  i\right)  }\ dt+\sum_{k=1}^{\infty}\mathcal{L}_{\xi_{k}}%
\omega^{\left(  i\right)  }\ dB_{t}^{k}=\frac{1}{2}\sum_{k=1}^{\infty
}\mathcal{L}_{\xi_{k}}^{2}\omega^{\left(  i\right)  }\ dt,\ \ \ i=1,2,
\]
where $\omega^{\left(  i\right)  }=\operatorname{curl}v^{\left(  i\right)  }$.
The difference $\Omega=\omega^{\left(  1\right)  }-\omega^{\left(  2\right)
}$\ satisfies%
\[
d\Omega+\mathcal{L}_{v^{\left(  1\right)  }}\omega^{\left(  1\right)
}\ dt-\mathcal{L}_{v^{\left(  2\right)  }}\omega^{\left(  2\right)  }%
\ dt+\sum_{k=1}^{\infty}\mathcal{L}_{\xi_{k}}\Omega\ dB_{t}^{k}=\frac{1}%
{2}\sum_{k=1}^{\infty}\mathcal{L}_{\xi_{k}}^{2}\Omega\ dt
\]
and thus (set also $V=v^{\left(  1\right)  }-v^{\left(  2\right)  }$)%
\[
d\Omega+\mathcal{L}_{V}\omega^{\left(  1\right)  }\ dt+\mathcal{L}_{v^{\left(
2\right)  }}\Omega\ dt+\sum_{k=1}^{\infty}\mathcal{L}_{\xi_{k}}\Omega
\ dB_{t}^{k}=\frac{1}{2}\sum_{k=1}^{\infty}\mathcal{L}_{\xi_{k}}^{2}%
\Omega\ dt.
\]
It follows%
\begin{align*}
&  \frac{1}{2}d\left\Vert \Omega\right\Vert _{L^{2}}^{2}+\left\langle
\mathcal{L}_{V}\omega^{\left(  1\right)  },\Omega\right\rangle
\ dt+\left\langle \mathcal{L}_{v^{\left(  2\right)  }}\Omega,\Omega
\right\rangle \ dt+\sum_{k=1}^{\infty}\left\langle \mathcal{L}_{\xi_{k}}%
\Omega,\Omega\right\rangle \ dB_{t}^{k}\\
&  =\frac{1}{2}\sum_{k=1}^{\infty}\left\langle \mathcal{L}_{\xi_{k}}^{2}%
\Omega,\Omega\right\rangle \ dt+\frac{1}{2}\sum_{k=1}^{\infty}\left\langle
\mathcal{L}_{\xi_{k}}\Omega,\mathcal{L}_{\xi_{k}}\Omega\right\rangle \ dt.
\end{align*}
We rewrite%
\begin{align*}
&  \left\langle \mathcal{L}_{V}\omega^{\left(  1\right)  },\Omega\right\rangle
+\left\langle \mathcal{L}_{v^{\left(  2\right)  }}\Omega,\Omega\right\rangle
\\
&  =\left\langle V\cdot\nabla\omega^{\left(  1\right)  },\Omega\right\rangle
-\left\langle \omega^{\left(  1\right)  }\cdot\nabla V,\Omega\right\rangle
+\left\langle v^{\left(  2\right)  }\cdot\nabla\Omega,\Omega\right\rangle
-\left\langle \Omega\cdot\nabla v^{\left(  2\right)  },\Omega\right\rangle
\end{align*}
and use the following inequalities:%
\[
\left\vert \left\langle V\cdot\nabla\omega^{\left(  1\right)  },\Omega
\right\rangle \right\vert \leq\left\Vert \Omega\right\Vert _{L^{2}}\left\Vert
V\right\Vert _{L^{4}}\left\Vert \nabla\omega^{\left(  1\right)  }\right\Vert
_{L^{4}}\leq C\left\Vert \omega^{\left(  1\right)  }\right\Vert _{W^{2,2}%
}\left\Vert \Omega\right\Vert _{L^{2}}^{2}%
\]%
\[
\left\vert \left\langle \omega^{\left(  1\right)  }\cdot\nabla V,\Omega
\right\rangle \right\vert \leq\left\Vert \Omega\right\Vert _{L^{2}}\left\Vert
\nabla V\right\Vert _{L^{2}}\left\Vert \omega^{\left(  1\right)  }\right\Vert
_{L^{\infty}}\leq C\left\Vert \omega^{\left(  1\right)  }\right\Vert
_{W^{2,2}}\left\Vert \Omega\right\Vert _{L^{2}}^{2}%
\]%
\[
\left\langle v^{\left(  2\right)  }\cdot\nabla\Omega,\Omega\right\rangle =0
\]%
\[
\left\vert \left\langle \Omega\cdot\nabla v^{\left(  2\right)  }%
,\Omega\right\rangle \right\vert \leq\left\Vert \Omega\right\Vert _{L^{2}}%
^{2}\left\Vert \nabla v^{\left(  2\right)  }\right\Vert _{L^{\infty}}%
\leq\left\Vert \Omega\right\Vert _{L^{2}}^{2}\left\Vert \nabla v^{\left(
2\right)  }\right\Vert _{W^{2,2}}\leq C\left\Vert \omega^{\left(  2\right)
}\right\Vert _{W^{2,2}}\left\Vert \Omega\right\Vert _{L^{2}}^{2}.
\]

Here and below we repeatedly use the Sobolev embedding theorems
\begin{equation}
W^{2,2}\left(  \mathbb{T}^{3}\right)  \subset C_{b}\left(  \mathbb{T}%
^{3}\right)  ,\qquad W^{2,2}\left(  \mathbb{T}^{3}\right)  \subset
W^{1,4}\left(  \mathbb{T}^{3}\right)  \label{Sobolev embedding}%
\end{equation}
and the fact that Biot-Savart map $\omega\mapsto v$ maps $W^{\alpha,p}$ into
$W^{\alpha+1,p}$ for all $\alpha\geq0$ and $p\in\left(  1,\infty\right)  $.
$W_{\sigma}^{s,2}$ into $W_{\sigma}^{s+1,2}$ for all $s\geq0$, see
(\ref{Biot Savart}). For instance, the sequences of inequalities used above in
the case of the terms $\left\Vert V\right\Vert _{L^{4}}$ and $\left\Vert
\nabla v^{\left(  2\right)  }\right\Vert _{L^{\infty}}$ were%
\[
\left\Vert V\right\Vert _{L^{4}}\leq C\left\Vert V\right\Vert _{W^{1,2}}\leq
C^{\prime}\left\Vert \Omega\right\Vert _{L^{2}}%
\]%
\[
\left\Vert \nabla v^{\left(  2\right)  }\right\Vert _{L^{\infty}}\leq
C\left\Vert \nabla v^{\left(  2\right)  }\right\Vert _{W^{2,2}}\leq C^{\prime
}\left\Vert v^{\left(  2\right)  }\right\Vert _{W^{3,2}}\leq C^{\prime\prime
}\left\Vert \omega^{\left(  2\right)  }\right\Vert _{W^{2,2}}.
\]
We omit similar detailed explanations sometimes below, when they are of the
same kind. 

Using also (\ref{assump}), we get%
\[
d\left\Vert \Omega\right\Vert _{L^{2}}^{2}+2\sum_{k=1}^{\infty}\left\langle
\mathcal{L}_{\xi_{k}}\Omega,\Omega\right\rangle \ dB_{t}^{k}\leq C\left(
1+\left\Vert \omega^{\left(  1\right)  }\right\Vert _{W^{2,2}}+\left\Vert
\omega^{\left(  2\right)  }\right\Vert _{W^{2,2}}\right)  \left\Vert
\Omega\right\Vert _{L^{2}}^{2}\ dt.
\]
Then%
\begin{align*}
d\left(  e^{Y}\left\Vert \Omega\right\Vert _{L^{2}}^{2}\right)   &
=-e^{Y}\left\Vert \Omega\right\Vert _{L^{2}}^{2}C\left(  1+\left\Vert
\omega^{\left(  1\right)  }\right\Vert _{W^{2,2}}+\left\Vert \omega^{\left(
2\right)  }\right\Vert _{W^{2,2}}\right)  +e^{Y}d\left\Vert \Omega\right\Vert
_{L^{2}}^{2}\\
&  \leq-2e^{Y}\sum_{k=1}^{\infty}\left\langle \mathcal{L}_{\xi_{k}}%
\Omega,\Omega\right\rangle \ dB_{t}^{k},
\end{align*}
where $Y$ is defined as
\[
Y_{t}:=-\int_{0}^{t}C\left(  1+\left\Vert \omega_{s}^{\left(  1\right)
}\right\Vert _{W^{2,2}}+\left\Vert \omega_{s}^{\left(  2\right)  }\right\Vert
_{W^{2,2}}\right)  ds.
\]
The inequality (recall $\Omega_{0}=0$)
\[
e^{Y_{\bar{\tau}}}\left\Vert \Omega_{\bar{\tau}}\right\Vert _{L^{2}}^{2}%
\leq-2\sum_{k=1}^{\infty}\int_{0}^{\bar{\tau}}e^{Y_{s}}\left\langle
\mathcal{L}_{\xi_{k}}\Omega_{s},\Omega_{s}\right\rangle \ dB_{s}^{k}%
\]
holds for every bounded stopping time $\bar{\tau}\in\left[  0,\tau\right]  $.
Hence we have%
\begin{align*}
e^{Y_{t\wedge\tau}}\left\Vert \Omega_{t\wedge\tau}\right\Vert _{L^{2}}^{2}  &
\leq-2\sum_{k=1}^{\infty}\int_{0}^{t\wedge\tau}e^{Y_{s}}\left\langle
\mathcal{L}_{\xi_{k}}\Omega_{s},\Omega_{s}\right\rangle \ dB_{s}^{k}\\
&  =-2\sum_{k=1}^{\infty}\int_{0}^{t}1_{s\leq\tau}e^{Y_{s}}\left\langle
\mathcal{L}_{\xi_{k}}\Omega_{s},\Omega_{s}\right\rangle \ dB_{s}^{k}.
\end{align*}
In expectation, denoted $\mathbb{E}$, this implies%
\[
\mathbb{E}\left[  e^{Y_{t\wedge\tau}}\left\Vert \Omega_{t\wedge\tau
}\right\Vert _{L^{2}}^{2}\right]  \leq0
\]
namely $\mathbb{E}\left[  e^{Y_{t\wedge\tau}}\left\Vert \Omega_{t\wedge\tau
}\right\Vert _{L^{2}}^{2}\right]  =0$ and thus, for every $t$,%
\[
e^{Y_{t\wedge\tau}}\left\Vert \Omega_{t\wedge\tau}\right\Vert _{L^{2}}%
^{2}=0\qquad\text{a.s.}%
\]
Since $Y_{t\wedge\tau}<\infty$ a.s., we get $\left\Vert \Omega_{t\wedge\tau
}\right\Vert _{L^{2}}^{2}=0$ a.s. and thus
\[
\omega_{t\wedge\tau}^{\left(  1\right)  }=\omega_{t\wedge\tau}^{\left(
2\right)  }\qquad\text{a.s.}%
\]
Recalling the continuity of trajectories, this implies
\[
\omega^{\left(  1\right)  }=\omega^{\left(  2\right)  }\qquad\text{a.s.}%
\]
The proof of the proposition is complete.
\end{proof}

\subsection{Existence of a maximal solution\label{uniq2}}

Given $R>0$, consider the modified Euler equations%
\begin{equation}
d\omega_{R}+\kappa_{R}\left(  \omega_{R}\right)  \mathcal{L}_{v_{R}}\omega
_{R}\ dt+\sum_{k=1}^{\infty}\mathcal{L}_{\xi_{k}}\omega_{R}dB_{t}^{k}=\frac
{1}{2}\sum_{k=1}^{\infty}\mathcal{L}_{\xi_{k}}^{2}\omega_{R}\ dt,\qquad
\omega_{R}|_{t=0}=\omega_{0}, \label{eq Euler Ito cutoff}%
\end{equation}
where $\omega_{R}=\operatorname{curl}v_{R}$. In (\ref{eq Euler Ito cutoff}),
$\kappa_{R}\left(  \omega\right)  :=f_{R}(\left\Vert \nabla v\right\Vert
_{\infty})$, where $f_{R}$ is a smooth function, equal to 1 on $\left[
0,R\right]  $, equal to 0 on $[R+1,\infty)$ and decreasing in $\left[
R,R+1\right]  $.

\begin{lemma}
\label{lemma cut off}Given $R>0$ and $\omega_{0}\in W_{\sigma}^{2,2}\left(
\mathbb{T}^{3},\mathbb{R}^{3}\right)  $, let $\omega_{R}:\Xi\times
\lbrack0,\infty)\times\mathbb{T}^{3}\rightarrow\mathbb{R}^{3}$ be a global
solution in $W^{2,2}$ of equation (\ref{eq Euler Ito cutoff}). Let
\[
\tau_{R}=\inf\left\{  t\geq0:\left\Vert \omega\right\Vert _{W^{2,2}}\geq
\frac{R}{C}\right\}  ,
\]
where $C$ is a constant chosen so that
\[
\left\Vert \nabla v\right\Vert _{\infty}\leq C\left\Vert \omega\right\Vert
_{W^{2,2}}.
\]
Finally, let $\omega:\Xi\times\left[  0,\tau_{R}\right]  \times\mathbb{T}%
^{3}\rightarrow\mathbb{R}^{3}$ be the restriction of $\omega_{R}$. Then
$\omega$ is a local solution in $W_{\sigma}^{2,2}$ of the stochastic 3D Euler
equations (\ref{eq Euler Ito})
\end{lemma}

\begin{proof}
Obvious, because for $t\in\left[  0,\tau_{R}\right]  $ we have $\left\Vert
\nabla v\right\Vert _{\infty}\leq C\left\Vert \omega\right\Vert _{W^{2,2}}\leq
R$ and thus $\kappa_{R}\left(  \omega_{R}\right)  =1$, namely the equations
are the same.
\end{proof}

The following proposition is the cornerstone of the existence and uniqueness
of a maximal solution of the stochastic 3D Euler equation (\ref{eq Euler Ito})

\begin{proposition}
\label{Prop cut off}Given $R>0$ and $\omega_{0}\in W_{\sigma}^{2,2}\left(
\mathbb{T}^{3},\mathbb{R}^{3}\right)  $, there exists a global solution in
$W_{\sigma}^{2,2}$ of equation (\ref{eq Euler Ito cutoff}). Moreover, if
$\omega_{R}^{\left(  1\right)  },\omega_{R}^{\left(  2\right)  }:\Xi
\times\lbrack0,\infty)\times\mathbb{T}^{3}\rightarrow\mathbb{R}^{3}$ are two
global solutions in $W_{\sigma}^{2,2}$ of equation (\ref{eq Euler Ito cutoff}%
), then $\omega_{R}^{\left(  1\right)  }=\omega_{R}^{\left(  2\right)  }$.
\end{proposition}

We postpone the proof of Proposition \ref{Prop cut off} to the later sections.
For now let us show how it implies the existence of a maximal solution.

\begin{theorem}
\label{mtpart1}Given $\omega_{0}\in W_{\sigma}^{2,2}\left(  \mathbb{T}%
^{3},\mathbb{R}^{3}\right)  $, there exists a maximal solution $\left(
\tau_{\max},\omega\right)  $ of the stochastic 3D Euler equations
(\ref{eq Euler Ito}). Moreover, either $\tau_{\max}=\infty$ or $\lim
\sup_{t\uparrow\tau_{\max}}\left\Vert \tilde{\omega}\left(  t\right)
\right\Vert _{W^{2,2}}=+\infty.$
\end{theorem}

\begin{proof}
Choose $R=n$ in Lemma \ref{lemma cut off}, then $\left(  \tau_{n},\omega
_{n}\right)  $ is a local solution in $W_{\sigma}^{2,2}$ of the stochastic 3D
Euler equations (\ref{eq Euler Ito}). Moreover, define $\tau_{\max}%
:=\lim_{n\rightarrow\infty}\tau_{n}$ and define $\omega$ as $\omega
|_{[0,\tau_{n})}:=\omega_{Cn}|_{[0,\tau_{n})}$. By uniqueness $\omega
_{m}|_{[0,\tau_{n})}:=\omega_{n}|_{[0,\tau_{n})}$ for any $m\geq n$. So
$\omega$ is consistently defined.

The statement that either $\tau_{\max}=\infty$ or $\lim\sup_{t\uparrow
\tau_{\max}}\left\Vert \tilde{\omega}\left(  t\right)  \right\Vert _{W^{2,2}%
}=+\infty$~is obvious: if $\tau_{\max}<\infty,$ then by the continuity of
$\tilde{\omega}$ on $[0,\tau_{\max})$, there exists some random times
$\tilde{\tau}_{n}<\tau_{n}$ such that $\tau_{n}-\tilde{\tau}_{n}\leq\frac
{1}{n}$ and that $\left\Vert \tilde{\omega}\left(  \tilde{\tau}_{n}\right)
\right\Vert _{W^{2,2}}\geq\frac{n-1}{C}$. Then
\[
\lim\sup_{t\uparrow\tau_{\max}}\left\Vert \tilde{\omega}\left(  t\right)
\right\Vert _{W^{2,2}}\geq\lim\sup_{n\uparrow\infty}\left\Vert \tilde{\omega
}\left(  \tilde{\tau}_{n}\right)  \right\Vert _{W^{2,2}}=\infty.
\]

We prove by contradiction that $(\tau_{\max},\omega)$ is a maximal solution.
Assume that there exists a pair $\left(  \tau^{\prime},\omega^{\prime}\right)
$ such that $\omega^{\prime}=\omega$ on $[0,\tau^{\prime}\wedge\tau_{\max}),$
and $\tau^{\prime}>\tau_{\max}$ with positive probability. This can only
happen if $\tau_{\max}<\infty$, therefore by the continuity of $\omega
^{\prime}$ on $[0,\tau^{\prime})~$on the set $\left\{  \tau^{\prime}%
>\tau_{\max}\right\}  $
\[
\infty=\lim\sup_{n\uparrow\infty}\left\Vert \tilde{\omega}\left(  \Tilde{\tau
}_{n}\right)  \right\Vert _{W^{2,2}}=\lim\sup_{n\uparrow\infty}\left\Vert
\tilde{\omega}^{\prime}\left(  \tilde{\tau}_{n}\right)  \right\Vert _{W^{2,2}%
}=\left\Vert \tilde{\omega}^{\prime}\left(  \tau_{\max}\right)  \right\Vert
_{W^{2,2}}<\infty,
\]
which leads to a contradiction. Hence, necessarily, $\tau^{\prime}\leq\tau$
$P$-almost surely, therefore $\left(  \tau,\omega\right)  $ is a maximal solution.
\end{proof}

\subsection{Uniqueness of the maximal solution\label{sect uniqueness W22}}

Let us start by justifying the uniqueness of the solution truncated Euler
equation (\ref{eq Euler Ito cutoff}). The proof is similar with that of
Proposition \ref{uniqueness proposition} so we only sketch it here. Let
$\omega_{R}^{\left(  1\right)  },\omega_{R}^{\left(  2\right)  }:\Xi
\times\lbrack0,\infty)\times\mathbb{T}^{3}\rightarrow\mathbb{R}^{3}$ are two
global solutions in $W_{\sigma}^{2,2}$ of equation (\ref{eq Euler Ito cutoff}%
). We preserve the same notation as in the proof of Proposition
\ref{uniqueness proposition}, i.e. denote by $\Omega=\omega_{R}^{\left(
1\right)  }-\omega_{R}^{\left(  2\right)  }$\ and $V=v_{R}^{\left(  1\right)
}-v_{R}^{\left(  2\right)  }$. We also assume that the truncation function
$f_{R}$ is Lipschitz and we will denote by $K_{R}$ the
quantity
\[
K_{R}=f_{R}(||\nabla v_{R}^{\left(  1\right)  }||_{\infty})-f_{R}(|\nabla
v_{R}^{\left(  2\right)  }||_{\infty})
\]
and observe that,\footnote{Here and throughout the paper, we use the standard
notation $C$ for generic constants, whose value may differ from case to
case.}
\begin{align*}
\left\vert K_{R}\right\vert  &  \leq C\left\Vert \nabla v_{R}^{\left(
1\right)  }-\nabla v_{R}^{\left(  2\right)  }\right\Vert _{\infty}\\
&  =C\left\Vert \nabla V\right\Vert _{\infty}\leq C\left\Vert \Omega
\right\Vert _{2,2}^{2}.
\end{align*}

To simplify notation we will omit the dependence on
$R$ in the following.

We are looking to prove uniqueness using the $W^{2,2}$-topology, therefore we
need to estimate the sum $\left\Vert \Omega\right\Vert _{L^{2}}^{2}+\left\Vert
\Delta\Omega\right\Vert _{L^{2}}^{2}$.\ Then%
\[
d\Omega+\Phi\ dt+\sum_{k=1}^{\infty}\mathcal{L}_{\xi_{k}}\Omega\ dB_{t}%
^{k}=\frac{1}{2}\sum_{k=1}^{\infty}\mathcal{L}_{\xi_{k}}^{2}\Omega\ dt,
\]
where
\[
\Phi:=\kappa\left(  \omega^{\left(  1\right)  }\right)  \mathcal{L}%
_{v^{\left(  1\right)  }}\omega^{\left(  1\right)  }\ -\kappa\left(
\omega^{\left(  2\right)  }\right)  \mathcal{L}_{v^{\left(  2\right)  }}%
\omega^{\left(  2\right)  }%
\]
and thus%
\[
d\Omega+\left\langle \Phi,\Omega\right\rangle \ dt+\sum_{k=1}^{\infty
}\mathcal{L}_{\xi_{k}}\Omega\ dB_{t}^{k}=\frac{1}{2}\sum_{k=1}^{\infty
}\mathcal{L}_{\xi_{k}}^{2}\Omega\ dt.
\]
Then
\[
\frac{1}{2}d\left\Vert \Omega\right\Vert _{L^{2}}^{2}+\left\langle \Phi
,\Omega\right\rangle \ dt+\sum_{k=1}^{\infty}\left\langle \mathcal{L}_{\xi
_{k}}\Omega,\Omega\right\rangle \ dB_{t}^{k}=\frac{1}{2}\sum_{k=1}^{\infty
}\left\langle \mathcal{L}_{\xi_{k}}^{2}\Omega,\Omega\right\rangle
\ dt+\frac{1}{2}\sum_{k=1}^{\infty}\left\langle \mathcal{L}_{\xi_{k}}%
\Omega,\mathcal{L}_{\xi_{k}}\Omega\right\rangle \ dt.
\]
On the set $\tau^{\left(  2\right)  }\leq\tau^{\left(  1\right)  }$ observe
that $\Phi$ is $0$ if $\left\Vert \omega^{\left(  1\right)  }\right\Vert
_{W^{2,2}}\geq\frac{R}{C}$. It follows that, on this set there exists a
constant $c_{R}$ such that (recall that $0\leq\kappa\leq1$) \
\begin{align*}
\left\vert \left\langle \Phi,\Omega\right\rangle \right\vert  &  =\left\vert
K\left\langle \mathcal{L}_{v^{\left(  1\right)  }}\omega^{\left(  1\right)
},\Omega\right\rangle \ +\kappa\left(  \omega^{\left(  2\right)  }\right)
\left\langle \mathcal{L}_{V}\omega^{\left(  1\right)  },\Omega\right\rangle
\ +\kappa\left(  \omega^{\left(  2\right)  }\right)  \left\langle
\mathcal{L}_{v^{\left(  2\right)  }}\Omega,\Omega\right\rangle \right\vert
\ \\
&  \leq c_{R}\left\Vert \Omega\right\Vert _{W^{2,2}}^{2}+\left\vert
\left\langle \mathcal{L}_{V}\omega^{\left(  1\right)  },\Omega\right\rangle
\ \right\vert +\left\vert \left\langle \mathcal{L}_{v^{\left(  2\right)  }%
}\Omega,\Omega\right\rangle \right\vert
\end{align*}
and, similar to the proof of Proposition \ref{uniqueness proposition} we
deduced that
\begin{equation}
d\left\Vert \Omega\right\Vert _{L^{2}}^{2}+2\sum_{k=1}^{\infty}\left\langle
\mathcal{L}_{\xi_{k}}\Omega,\Omega\right\rangle \ dB_{t}^{k}\leq C\left(
1+\left\Vert \omega^{\left(  1\right)  }\right\Vert _{W^{2,2}}+\left\Vert
\omega^{\left(  2\right)  }\right\Vert _{W^{2,2}}\right)  \left\Vert
\Omega\right\Vert _{W^{2,2}}\ dt. \label{interm1}%
\end{equation}
Similarly, on the set $\tau^{\left(  2\right)  }\leq\tau^{\left(  1\right)  }$
observe that $\Phi$ vanishes, if $\left\Vert \omega^{\left(  2\right)
}\right\Vert _{W^{2,2}}\geq\frac{R}{C}$ and \ref{interm1} holds true, as seen by
observing that there exists a constant $c_{R}$ such that \
\begin{align*}
\left\vert \left\langle \Phi,\Omega\right\rangle \right\vert  &  =\left\vert
K\left\langle \mathcal{L}_{v^{\left(  2\right)  }}\omega^{\left(  2\right)
},\Omega\right\rangle \ +\kappa\left(  \omega^{\left(  1\right)  }\right)
\left\langle \mathcal{L}_{V}\omega^{\left(  2\right)  },\Omega\right\rangle
\ +\kappa\left(  \omega^{\left(  1\right)  }\right)  \left\langle
\mathcal{L}_{v^{\left(  1\right)  }}\Omega,\Omega\right\rangle \ \right\vert
\\
&  \leq c_{R}\left\Vert \Omega\right\Vert _{W^{2,2}}^{2}+\left\vert
\left\langle \mathcal{L}_{V}\omega^{\left(  2\right)  },\Omega\right\rangle
\ \right\vert +\left\vert \left\langle \mathcal{L}_{v^{\left(  1\right)  }%
}\Omega,\Omega\right\rangle \right\vert .
\end{align*}
Next we have
\[
d\Delta\Omega+\ \Delta\Phi dt+\sum_{k=1}^{\infty}\Delta\mathcal{L}_{\xi_{k}%
}\Omega\ dB_{t}^{k}=\frac{1}{2}\sum_{k=1}^{\infty}\Delta\mathcal{L}_{\xi_{k}%
}^{2}\Omega\ dt
\]
from which we deduce that%
\[
\frac{1}{2}d\left\Vert \Delta\Omega\right\Vert _{L^{2}}^{2}+\left\langle
\Delta\Phi,\Delta\Omega\right\rangle +\sum_{k=1}^{\infty}\left\langle
\Delta\mathcal{L}_{\xi_{k}}\Omega,\Delta\Omega\right\rangle \ dB_{t}^{k}%
=\frac{1}{2}\sum_{k=1}^{\infty}\left\langle \Delta\mathcal{L}_{\xi_{k}}%
^{2}\Omega,\Delta\Omega\right\rangle \ dt+\frac{1}{2}\sum_{k=1}^{\infty
}\left\langle \Delta\mathcal{L}_{\xi_{k}}\Omega,\Delta\mathcal{L}_{\xi_{k}%
}\Omega\right\rangle \ dt.
\]
From the Lemma \ref{lemma GN estimate} we have%
\begin{align*}
\left\vert \left\langle \Delta\mathcal{L}_{v^{\left(  2\right)  }}%
\Omega,\Delta\Omega\right\rangle \right\vert  &  \leq C\left\Vert \nabla
v^{\left(  2\right)  }\right\Vert _{L^{\infty}}\left\Vert \Omega\right\Vert
_{W^{2,2}}^{2}+C\left\Vert \Omega\right\Vert _{L^{\infty}}\left\Vert \nabla
v^{\left(  2\right)  }\right\Vert _{W^{2,2}}\left\Vert \Omega\right\Vert
_{W^{2,2}}\\
&  \leq C\left\Vert \omega^{\left(  2\right)  }\right\Vert _{W^{2,2}%
}\left\Vert \Omega\right\Vert _{W^{2,2}}^{2}.
\end{align*}
Moreover, by similar arguments,
\begin{align*}
\left\vert \left\langle \Delta\mathcal{L}_{V}\omega^{\left(  1\right)
},\Delta\Omega\right\rangle \right\vert  &  \leq C\left\Vert \nabla
V\right\Vert _{L^{\infty}}\left\Vert \omega^{\left(  1\right)  }\right\Vert
_{W^{2,2}}\left\Vert \Omega\right\Vert _{W^{2,2}}+C\left\Vert \omega^{\left(
1\right)  }\right\Vert _{L^{\infty}}\left\Vert \nabla V\right\Vert _{W^{2,2}%
}\left\Vert \Omega\right\Vert _{W^{2,2}}\\
&  \leq C\left\Vert \nabla V\right\Vert _{W^{2,2}}\left\Vert \omega^{\left(
1\right)  }\right\Vert _{W^{2,2}}\left\Vert \Omega\right\Vert _{W^{2,2}}\\
&  \leq C\left\Vert \omega^{\left(  1\right)  }\right\Vert _{W^{2,2}%
}\left\Vert \Omega\right\Vert _{W^{2,2}}^{2}.
\end{align*}
Similar estimates hold true for $\left\vert \left\langle \Delta\mathcal{L}%
_{V}\omega^{\left(  2\right)  },\Delta\Omega\right\rangle \ \right\vert $ and
$\left\vert \left\langle \Delta\mathcal{L}_{v^{\left(  1\right)  }}%
\Omega,\Delta\Omega\right\rangle \right\vert $. Next, as above, on the set
$\tau^{\left(  2\right)  }\leq\tau^{\left(  1\right)  }$ observe there exists
a constant $c_{R}$ such that
\begin{align*}
\left\vert K\left\langle \Delta\mathcal{L}_{v^{\left(  1\right)  }}%
\omega^{\left(  1\right)  },\Delta\Omega\right\rangle \right\vert  &  \leq
C\left\Vert \nabla v^{\left(  1\right)  }\right\Vert _{L^{\infty}}\left\Vert
\omega^{\left(  1\right)  }\right\Vert _{W^{2,2}}\left\Vert \Omega\right\Vert
_{W^{2,2}}+C\left\Vert \omega^{\left(  1\right)  }\right\Vert _{L^{\infty}%
}\left\Vert \nabla v^{\left(  1\right)  }\right\Vert _{W^{2,2}}\left\Vert
\Omega\right\Vert _{W^{2,2}}\\
&  \leq c_{R}\left\Vert \Omega\right\Vert _{W^{2,2}}^{2}.
\end{align*}
Similarly, on the set $\tau^{\left(  2\right)  }\leq\tau^{\left(  1\right)  }%
$,
\[
\left\vert K\left\langle \Delta\mathcal{L}_{v^{\left(  1\right)  }}%
\omega^{\left(  1\right)  },\Delta\Omega\right\rangle \right\vert \leq
c_{R}\left\Vert \Omega\right\Vert _{W^{2,2}}^{2}.
\]
Summarizing, we deduce that%
\begin{align*}
&  \frac{1}{2}d\left\Vert \Delta\Omega\right\Vert _{L^{2}}^{2}+\sum
_{k=1}^{\infty}\left\langle \Delta\mathcal{L}_{\xi_{k}}\Omega,\Delta
\Omega\right\rangle \ dB_{t}^{k}\leq C\left(  1+\left\Vert \omega^{\left(
1\right)  }\right\Vert _{W^{2,2}}+\left\Vert \omega^{\left(  2\right)
}\right\Vert _{W^{2,2}}\right)  \left\Vert \Omega\right\Vert _{W^{2,2}}^{2}\\
&  +\frac{1}{2}\sum_{k=1}^{\infty}\left\langle \Delta\mathcal{L}_{\xi_{k}}%
^{2}\Omega,\Delta\Omega\right\rangle \ dt+\frac{1}{2}\sum_{k=1}^{\infty
}\left\langle \Delta\mathcal{L}_{\xi_{k}}\Omega,\Delta\mathcal{L}_{\xi_{k}%
}\Omega\right\rangle \ dt.
\end{align*}
It is then sufficient to repeat the argument of the proof of Proposition
\ref{uniqueness proposition} to control $\left\Vert \Omega\right\Vert _{L^{2}%
}^{2}+\left\Vert \Delta\Omega\right\Vert _{L^{2}}^{2}$ and obtain the
uniqueness of the truncated Euler equation. The computation required here
requires more regularity in space than what we have for our solutions (we
have to compute, although only transiently, $\Delta\mathcal{L}_{\xi_{k}}%
^{2}\Omega$). In order to make the computation rigorous one has to regularize
solutions by mollifiers or Yosida approximations and do the computations on
the regularizations. In this process, commutators will appear and one has to check at the end
that they converge to zero. The details are tedious, but straightforward and
we do not write all of them here.

We are now ready to prove the general uniqueness result contained in Theorem
\ref{main theorem}. More precisely we prove the following

\begin{theorem}
\label{mtpart2}Let $\omega_{0}\in W_{\sigma}^{2,2}\left(  \mathbb{T}%
^{3},\mathbb{R}^{3}\right)  $ and $\left(  \tau_{\max},\omega\right)  $ be the
maximal solution of the stochastic 3D Euler equations (\ref{eq Euler Ito})
introduced in Theorem \ref{mtpart1}. Moreover, let $\left(  \tau,\tilde
{\omega}\right)  $ be another maximal solutions of the same equation with the
same initial condition $\omega_{0}\in W_{\sigma}^{2,2}\left(  \mathbb{T}%
^{3},\mathbb{R}^{3}\right)  $. Then necessarily $\tau=\tau_{\max}$ and
$\omega=\omega^{\prime}$on $[0,\tau_{\max})$.
\end{theorem}

\begin{proof}
From the local uniqueness result proved above we deduce that $\omega
=\tilde{\omega}\ $on $[0,\min\left(  \tau,\tau_{\max}\right)  ).\ $By an
argument similar to the one in Theorem \ref{mtpart1}, we cannot have
$\tau_{\max}<\tau$ on any non-trivial set. Hence $\tau<\tau_{\max}.\,$But then
from the maximality property of $(\tau,\tilde{\omega})$ it follows that
necessarily $\tau=\tau_{\max}\ $and therefore $\omega=\tilde{\omega}\ $on
$[0,\tau_{\max}).$
\end{proof}

\subsection{Proof of the Beale-Kato-Majda criterion}

\label{sect BKM proof}

In this section we prove Theorem \ref{BKM}. In the following, we will use the
fact that there exists two constants $C_{1}$and $C_{2}$ such that
\begin{equation}
C_{1}\left\vert \left\vert \omega\right\vert \right\vert _{2,2}\leq\left\vert
\left\vert v\right\vert \right\vert _{3,2}\leq C_{2}\left\vert \left\vert
\omega\right\vert \right\vert _{2,2}. \label{vomega}%
\end{equation}
The first inequality follows from that fact that $\omega=\operatorname{curl}%
v$, whilst the second inequality follows from (\ref{Biot Savart}).
\begin{lemma}
\label{lemma plus}There is a constants $C$ such that\footnote{We thank James-Michael Leahy for pointing out an error in an earlier version of the proof of this result.}%
\begin{equation}
\left\vert \left\vert \nabla v\right\vert \right\vert _{\infty}\leq C\left(
1+(\log(  \left\vert \left\vert \omega\right\vert \right\vert _{2,2}%
^{2}+e)) \left\vert \left\vert \omega\right\vert \right\vert
_{\infty} \right)  . \label{plus}%
\end{equation}
%\bigskip
\end{lemma}

\begin{proof}
Cf. \cite{BKM1984} the following inequality holds true
\begin{equation}
\left\vert \left\vert \nabla v\right\vert \right\vert _{\infty}\leq C\left(
1+\left(  1+\log^{+}\left\vert \left\vert v\right\vert \right\vert
_{3,2}\right)  \left\vert \left\vert \omega\right\vert \right\vert _{\infty
}\right)  +\left\vert \left\vert \omega\right\vert \right\vert _{2}.
\label{better}%
\end{equation}
The result is then obtained from \eqref{vomega}, the obvious inequality
$1+\log^{+}a\leq C\log\left(  a+e\right)  $ for $C$ sufficiently large (say $C\ge 2$) and the fact that $\left\vert
\left\vert \omega\right\vert \right\vert _{2}\leq C\left\vert \left\vert
\omega\right\vert \right\vert _{\infty}$ on a torus.
\end{proof}

\begin{theorem}
Let $\tau^{1}$ and $\tau^{2}$ be the folllowing stopping times
\begin{align*}
\tau^{1}  &  =\lim_{n\rightarrow\infty}\tau_{n}^{1}\text{~~}where~~\tau
_{n}^{1}=\inf_{t\geq0}\left\{  t\geq0|~~\left\vert \left\vert \omega
_{t}\right\vert \right\vert _{2,2}\geq n\right\} 
\,,\\
\tau^{2}  &  =\lim_{n\rightarrow\infty}\tau_{n}^{2}\text{~~}where~~\tau
_{n}^{2}=\inf_{t\geq0}\left\{  t\geq0|~~\int_{0}^{t}\left\vert \left\vert
\omega_{s}\right\vert \right\vert _{\infty}ds\geq n\right\}
\,.
\end{align*}
Then, $P$-almost surely $\tau^{1}=\tau^{2}$.
\end{theorem}

\begin{proof}
\textbf{Step 1. }$\tau^{1}\leq\tau^{2}$\textbf{.}\ 

From the imbedding $W^{2,2}\left(  \mathbb{T}^{3};\mathbb{R}^{3}\right)
\subset C\left(  \mathbb{T}^{3};\mathbb{R}^{3}\right)  ,$ there exists $C$
such that $\left\vert \left\vert \omega\right\vert \right\vert _{\infty}\leq
C\left\vert \left\vert \omega\right\vert \right\vert _{2,2}$. Then
\[
\int_{0}^{\tau_{n}^{1}}\left\vert \left\vert \omega_{s}\right\vert \right\vert
_{\infty}ds\leq\left(  \left[  C\right]  +1\right)  \sup_{s\leq\tau_{n}^{1}%
}\left\vert \left\vert \omega_{s}\right\vert \right\vert _{2,2}\leq\left(
\left[  C\right]  +1\right)  n.
\]
Hence $\tau_{n}^{1}\leq\tau_{\left(  \left[  C\right]  +1\right)  n}^{2}%
\leq\tau^{2}$ and therefore the claim holds true.

\textbf{Step 2. }$\tau^{2}\leq\tau^{1}$\textbf{.\ }$P$\textbf{-a.s.}\ 

We prove that\ for any $n,k>0$ we have%

\begin{equation}
\mathbb{E}\left[  \log\left(  \left(  \sup_{s\in\left[  0,\tau_{n}^{2}\wedge
k\right]  }\left\vert \left\vert \omega_{t}\right\vert \right\vert
_{2,2}\right)  ^{2}+e\right)  \right]  <\infty. \label{helpinghand}%
\end{equation}
In particular $\sup_{s\in\left[  0,\tau_{n}^{2}\wedge k\right]  }\left\vert
\left\vert \omega_{t}\right\vert \right\vert _{2,2}$ is a finite random
variable $P$-almost surely, that is
\[
\mathbb{P}\left(  \sup_{s\in\left[  0,\tau_{n}^{2}\wedge k\right]  }\left\vert
\left\vert \omega_{t}\right\vert \right\vert _{2,2}<\infty\right)  =1.
\]
Since
\[
\left\{  \sup_{s\in\left[  0,\tau_{n}^{2}\wedge k\right]  }\left\vert
\left\vert \omega_{t}\right\vert \right\vert _{2,2}<\infty\right\}
=\bigcup_{N}\left\{  \sup_{s\in\left[  0,\tau_{n}^{2}\wedge k\right]
}\left\vert \left\vert \omega_{t}\right\vert \right\vert _{2,2}<N\right\}
\subset\bigcup_{N}\left\{  \tau_{n}^{2}\wedge k<\tau_{N}^{1}\right\}
\subset\left\{  \tau_{n}^{2}\wedge k\leq\tau^{1}\right\}
\]
we deduce that $\tau_{n}^{2}\wedge k\leq\tau^{1}$ $P$-almost surely. Then
\[
\left\{  \tau^{2}\leq\tau^{1}\right\}  =\left\{  \lim_{n\mapsto\infty}\tau
_{n}^{2}\leq\tau^{1}\right\}  =\bigcap_{n}\left\{  \tau_{n}^{2}\leq\tau
^{1}\right\}  =\bigcap_{n}\bigcap_{k}\left\{  \tau_{n}^{2}\wedge k<\tau
^{1}\right\}
\]
and therefore the second claim holds true since all the sets in the above
intersection have full measure.

To prove (\ref{helpinghand}) we proceed as follows: For arbitrary $R>0$, and
$\nu\in\left(  0,1\right)  ,$ let $\omega_{R}^{\nu}$ the solution of equation%
\[
d\omega_{R}^{\nu}+\kappa_{R}\left(  \omega_{R}^{\nu}\right)  \mathcal{L}%
_{v_{R}^{\nu}}\omega_{R}^{\nu}\ dt+\sum_{k=1}^{\infty}\mathcal{L}_{\xi_{k}%
}\omega_{R}^{\nu}dB_{t}^{k}=\nu\Delta^{5}\omega_{R}^{\nu}dt+\frac{1}{2}%
\sum_{k=1}^{\infty}\mathcal{L}_{\xi_{k}}^{2}\omega_{R}^{\nu}\ dt
\]
with $\omega_{R}^{\nu}|_{t=0}=\omega_{0}$. We know from the analysis in the next section that if
$\omega_{0}\in W^{2,2}\left(  \mathbb{T}^{3};\mathbb{R}^{3}\right)  $, then
$\omega_{t}\in W^{4,2}\left(  \mathbb{T}^{3};\mathbb{R}^{3}\right)  $. To
simplify notation, in the following we will omit the dependence on $\nu$ and
$R$ of $\omega_{R}^{\nu}$ and denote it by $\omega$. We have that
\begin{align*}
&  \frac{1}{2}d\left\Vert \omega\right\Vert _{L^{2}}^{2}+\kappa_{R}\left(
\omega\right)  \left\langle \mathcal{L}_{v}\omega,\omega\right\rangle
\ dt+\sum_{k=1}^{\infty}\left\langle \mathcal{L}_{\xi_{k}}\omega
,\omega\right\rangle \ dB_{t}^{k}\\
&  \hspace{1cm}=\nu\left\langle \Delta^{5}\omega,\omega\right\rangle
dt+\frac{1}{2}\sum_{k=1}^{\infty}(\left\langle \mathcal{L}_{\xi_{k}}^{2}%
\omega,\omega\right\rangle \ dt+\left\langle \mathcal{L}_{\xi_{k}}%
\omega,\mathcal{L}_{\xi_{k}}\omega\right\rangle )\ dt.\\
&  \frac{1}{2}d\left\Vert \Delta\omega\right\Vert _{L^{2}}^{2}+\kappa
_{R}\left(  \omega\right)  \left\langle \Delta\mathcal{L}_{v}\omega
,\Delta\omega\right\rangle \ dt+\sum_{k=1}^{\infty}\left\langle \Delta
\mathcal{L}_{\xi_{k}}\omega,\Delta\omega\right\rangle \ dB_{t}^{k}\\
&  \hspace{1cm}=\nu\left\langle \Delta^{6}\omega,\Delta\omega\right\rangle
dt+\frac{1}{2}\sum_{k=1}^{\infty}(\left\langle \Delta\mathcal{L}_{\xi_{k}}%
^{2}\omega,\Delta\omega\right\rangle \ +\left\langle \Delta\mathcal{L}%
_{\xi_{k}}\omega,\Delta\mathcal{L}_{\xi_{k}}\omega\right\rangle )\ dt.
\end{align*}
Next we will use the following set of inequalities
\begin{align*}
\left\langle \Delta^{5}\omega,\omega\right\rangle  &  =-\left\vert \left\vert
\Delta^{5/2}\omega\right\vert \right\vert _{L^{2}}\leq0,~~\left\langle
\Delta^{6}\omega,\Delta\omega\right\rangle =-\left\vert \left\vert
\Delta^{7/2}\omega\right\vert \right\vert _{L^{2}}\leq0\\
\left\vert \left\langle \mathcal{L}_{v}\omega,\omega\right\rangle \right\vert
&  \leq\left\vert \left\vert \nabla v\right\vert \right\vert _{\infty
}\left\Vert \omega\right\Vert _{L^{2}}^{2},~~\left\vert \left\langle
\Delta\mathcal{L}_{v}\omega,\Delta\omega\right\rangle \right\vert \leq
C\left(  \left\Vert \omega\right\Vert _{\infty}+\left\vert \left\vert \nabla
v\right\vert \right\vert _{\infty}\right)  \left\Vert \omega\right\Vert
_{2,2}^{2}.
\end{align*}
The first two inequalities are obvious. The third one comes from the fact that in
$\left\langle \mathcal{L}_{v}\omega,\omega\right\rangle $ the term
$\left\langle v\cdot\nabla\omega,\omega\right\rangle $ vanishes and the term
$\left\vert \left\langle \omega\cdot\nabla v,\omega\right\rangle \right\vert $
is bounded by $\left\vert \left\vert \nabla v\right\vert \right\vert _{\infty
}\left\Vert \omega\right\Vert _{L^{2}}^{2}$. The last one, the most delicate
one, comes from Lemma \ref{lemma GN estimate}:%
\[
\left\vert \left\langle \Delta\mathcal{L}_{v}\omega,\Delta\omega\right\rangle
\right\vert \leq C\left\Vert \nabla v\right\Vert _{\infty}\left\Vert
\omega\right\Vert _{2,2}^{2}+C\left\Vert \omega\right\Vert _{\infty}\left\Vert
\nabla v\right\Vert _{2,2}\left\Vert \omega\right\Vert _{2,2}%
\]
and then we use $\left\Vert \nabla v\right\Vert _{2,2}\leq\left\Vert
v\right\Vert _{3,2}\leq C\left\Vert \omega\right\Vert _{2,2}$, see
(\ref{Biot Savart}).

Hence
\begin{align}
d\left\Vert \omega\right\Vert _{L^{2}}^{2}+2\sum_{k=1}^{\infty}\left\langle
\mathcal{L}_{\xi_{k}}\omega,\omega\right\rangle \ dB_{t}^{k}  &  \leq C\left(
1+\left\Vert \omega\right\Vert _{\infty}+\left\vert \left\vert \nabla
v\right\vert \right\vert _{\infty}\right)  \left\Vert \omega\right\Vert
_{L^{2}}^{2}\ dt\label{i1}\\
d\left\Vert \Delta\omega\right\Vert _{L^{2}}^{2}+2\sum_{k=1}^{\infty
}\left\langle \Delta\mathcal{L}_{\xi_{k}}\omega,\Delta\omega\right\rangle
\ dB_{t}^{k}  &  \leq C\left(  1+\left\Vert \omega\right\Vert _{\infty
}+\left\vert \left\vert \nabla v\right\vert \right\vert _{\infty}\right)
\left\Vert \omega\right\Vert _{2,2}^{2}\ dt, \label{j2}%
\end{align}
where we used the inequalities (\ref{striking est 1}), (\ref{striking est 2})
coupled with the assumption (\ref{assump}) to control
\[
\sum_{k=1}^{\infty}(\left\langle \mathcal{L}_{\xi_{k}}^{2}\omega
,\omega\right\rangle +\left\langle \mathcal{L}_{\xi_{k}}\omega
,\mathcal{L}_{\xi_{k}}\omega\right\rangle +\left\langle \Delta\mathcal{L}%
_{\xi_{k}}^{2}\omega,\Delta\omega\right\rangle \ +\left\langle \Delta
\mathcal{L}_{\xi_{k}}\omega,\Delta\mathcal{L}_{\xi_{k}}\omega\right\rangle ).
\]
Using It\^{o}'s formula and the fact that $\left\vert \left\vert \omega
_{t}\right\vert \right\vert _{2,2}^{2}\leq\left\Vert \omega_{t}\right\Vert
_{L^{2}}^{2}+\left\Vert \Delta\omega_{t}\right\Vert _{L^{2}}^{2}$ we deduce,
from (\ref{i1})+(\ref{j2}), that
\begin{align*}
d\log\left(  \left\vert \left\vert \omega_{t}\right\vert \right\vert
_{2,2}^{2}+e\right)   &  \leq\frac{1}{\left(  \left\vert \left\vert \omega
_{t}\right\vert \right\vert _{2,2}^{2}+e\right)  }d\left\vert \left\vert
\omega_{t}\right\vert \right\vert _{2,2}^{2}\\
&  -\frac{2}{\left(  \left\vert \left\vert \omega_{t}\right\vert \right\vert
_{2,2}^{2}+e\right)  ^{2}}\sum_{k=1}^{\infty}\left(  \left\vert \left\langle
\Delta\mathcal{L}_{\xi_{k}}\omega,\Delta\omega\right\rangle \right\vert
+\left\vert \left\langle \mathcal{L}_{\xi_{k}}\omega,\omega\right\rangle
\right\vert \ \right)  ^{2}\\
&  \leq\frac{C}{\left(  \left\vert \left\vert \omega_{t}\right\vert
\right\vert _{2,2}^{2}+e\right)  }\left(  1+\left\Vert \omega\right\Vert
_{\infty}+\left\vert \left\vert \nabla v\right\vert \right\vert _{\infty
}\right)  \left\Vert \omega\right\Vert _{2,2}^{2}\ dt+dM_{t},
\end{align*}
where $M\ $is the (local) martingale defined as
\[
M_{t}:=\sum_{k=1}^{\infty}\int_{0}^{t}\frac{2\left(  \left\langle
\mathcal{L}_{\xi_{k}}\omega,\omega\right\rangle +\left\langle \Delta
\mathcal{L}_{\xi_{k}}\omega,\Delta\omega\right\rangle \right)  }{\left(
\left\vert \left\vert \omega_{t}\right\vert \right\vert _{2,2}^{2}+e\right)
}\ dB_{s}^{k}%
\]

We use now (\ref{plus}) to deduce
\[
d\log\left(  \left\vert \left\vert \omega_{t}\right\vert \right\vert
_{2,2}^{2}+e\right)  \leq mC\left(  1+\left\vert \left\vert \omega\right\vert
\right\vert _{\infty}\right)  \log\left(  \left\vert \left\vert \omega
_{t}\right\vert \right\vert _{2,2}^{2}+e\right)  \ dt+dM_{t}.
\]
which, in
turn, implies that\begin{equation}
e^{-CY_{t}}\log\left(  \left\vert \left\vert \omega_{t}\right\vert \right\vert
_{2,2}^{2}+e\right)  \leq\log\left(  \left\Vert \omega_{0}\right\Vert _{L^{2}%
}^{2}+\left\Vert \Delta\omega\right\Vert _{L^{2}}^{2}+e\right)  +\int_{0}%
^{t}e^{-CY_{s}}dM_{s}, \label{qv0}%
\end{equation}
where%
\[
Y_{t}=\int_{0}^{t}\left(  1+\left\Vert \omega\right\Vert _{\infty}\right)
ds
\]
and we use the conventions $e^{-\infty}=0$ and $0\times\infty=0$.

Again, by using the fact $\left\vert \left\langle \mathcal{L}_{\xi_{k}}%
\omega,\omega\right\rangle \right\vert $ is controlled by $\left\vert
\left\vert \nabla\xi_{k}\right\vert \right\vert _{\infty}\left\Vert
\omega\right\Vert _{L^{2}}^{2}$ and that $\left\langle \Delta\mathcal{L}%
_{\xi_{k}}\omega,\Delta\omega\right\rangle |$ is controlled by $\left\Vert
\xi_{k}\right\Vert _{3,2}\left\Vert \omega\right\Vert _{2,2}^{2}$ folllowing
Lemma \ref{lemma GN estimate}, we deduced that
\begin{equation}
\left\vert \left\langle \mathcal{L}_{\xi_{k}}\omega,\omega\right\rangle
+\left\langle \Delta\mathcal{L}_{\xi_{k}}\omega,\Delta\omega\right\rangle
\right\vert \leq C\left\Vert \xi_{k}\right\Vert _{3,2}\left\Vert
\omega\right\Vert _{2,2}^{2} \label{qv1}.%
\end{equation}
From (\ref{qv1}) and assumption (\ref{Hp 3}) we deduce the following control on
the quadratic variation of stochastic integral in (\ref{qv0})
\begin{align*}
\left[  \int_{0}^{\cdot}e^{-CY_{s}}dM_{s}\right]  _{t}  &  =4\sum
_{k=1}^{\infty}\int_{0}^{t}e^{-2CY_{s}}\frac{\left(  \left\langle
\mathcal{L}_{\xi_{k}}\omega,\omega\right\rangle +\left\langle \Delta
\mathcal{L}_{\xi_{k}}\omega,\Delta\omega\right\rangle \right)  ^{2}}{\left(
\left\vert \left\vert \omega_{t}\right\vert \right\vert _{2,2}^{2}+e\right)
^{2}}\ ds\\
&  \leq4C\sum_{k=1}^{\infty}\left\Vert \xi_{k}\right\Vert _{3,2}^{2}\int%
_{0}^{t}\frac{\left\Vert \omega\right\Vert _{2,2}^{4}}{\left(  \left\vert
\left\vert \omega_{t}\right\vert \right\vert _{2,2}^{2}+e\right)  ^{2}}\ ds\\
&  \leq Ct.
\end{align*}
%\todo[inline]{\color{white} DH: Need citation for BDG inequality.}
Finally, using the Burkholder-Davis-Gundy inequality (see, e.g., Theorem 3.28, page 166 in \cite{KS}), we deduce that
\[
\mathbb{E}\left[  \sup_{s\in\left[  0,t\right]  }\left\vert \int_{0}%
^{s}e^{-CY_{r}}dM_{r}\right\vert \right]  \leq C\sqrt{t}.
\]
This means there exists a constant $C$ independent of $\nu$ and $R$ such
that, upon reverting to the notation $\omega_{R}^{\nu}$, we have
\[
\mathbb{E}\left[  \sup_{s\in\left[  0,t\right]  }e^{-\int_{0}^{s}C\left(
1+\left\Vert \omega_{R}^{\nu}\right\Vert _{\infty}\right)  dr}\log\left(
\left\vert \left\vert \omega_{R}^{\nu}\left(  s\right)  \right\vert
\right\vert _{2,2}^{2}+e\right)  \right]  \leq\log\left(  \left\Vert
\omega_{0}\right\Vert _{L^{2}}^{2}+\left\Vert \Delta\omega\right\Vert _{L^{2}%
}^{2}+e\right)  +C\sqrt{t}%
\]
By Fatou's lemma, it follows that the same limit holds for the limit of
$\omega_{R}^{\nu}$ as $\nu$ tends to 0 and $R$ tends to $\infty$, hence
\[
\mathbb{E}\left[  \sup_{s\in\left[  0,t\right]  }e^{-\int_{0}^{s}C\left(
1+\left\vert \left\vert \omega_{r}\right\vert \right\vert _{\infty}\right)
dr}\log\left(  \left\vert \left\vert \omega_{s}\right\vert \right\vert
_{2,2}^{2}+e\right)  \right]  <\infty.
\]
It follows that
\begin{align*}
e^{-Ck(1+n)}\mathbb{E}\left[  \sup_{s\in\left[  0,\tau_{n}^{2}\wedge k\right]
}\log\left(  \left\vert \left\vert \omega\left(  s\right)  \right\vert
\right\vert _{2,2}^{2}+e\right)  \right]   &  \leq\mathbb{E}\left[  \sup
_{s\in\left[  0,\tau_{n}^{2}\wedge k\right]  }e^{-CY_{s}}\log\left(
\left\vert \left\vert \omega\left(  s\right)  \right\vert \right\vert
_{2,2}^{2}+e\right)  \right] \\
&  \leq\mathbb{E}\left[  \sup_{s\in\left[  0,k\right]  }e^{-CY_{s}}\log\left(
\left\vert \left\vert \omega\left(  s\right)  \right\vert \right\vert
_{2,2}^{2}+e\right)  \right] \\
&  \leq\log\left(  \left\Vert \omega_{0}\right\Vert _{L^{2}}^{2}+\left\Vert
\Delta\omega\right\Vert _{L^{2}}^{2}+e\right)  +C\sqrt{t}<\infty.
\end{align*}
which gives us \eqref{helpinghand}. The proof is now complete.
\end{proof}

\begin{remark}
The original Beale-Kato-Majda result refers to a control of the explosion time
of $\left\vert \left\vert v\right\vert \right\vert _{3,2}$ in terms of
$\left\vert \left\vert \omega\right\vert \right\vert _{\infty}.$ Our result
refers to a control of the explosion time for $\left\vert \left\vert
\omega\right\vert \right\vert _{2,2}$ in terms of $\left\vert \left\vert
\omega\right\vert \right\vert _{\infty}$. However, due to (\ref{vomega}), we
can restate our result\ in terms of $\left\vert \left\vert v\right\vert
\right\vert _{3,2}$ as well.
\end{remark}

\subsection{Global existence of the truncated solution\label{sect cut off}}

Consider the following regularized equation with cut-off, with $\nu,R>0$,
\begin{align}
&  d\omega_{R}^{\nu}+\kappa_{R}\left(  \omega_{R}^{\nu}\right)  \mathcal{L}%
_{v_{R}^{\nu}}\omega_{R}^{\nu}\ dt+\sum_{k=1}^{\infty}\mathcal{L}_{\xi_{k}%
}\omega_{R}^{\nu}dB_{t}^{k}\label{regularized SPDE}\\
&  =\nu\Delta^{5}\omega_{R}^{\nu}dt+\frac{1}{2}\sum_{k=1}^{\infty}%
\mathcal{L}_{\xi_{k}}^{2}\omega_{R}^{\nu}\ dt\,,\qquad\omega_{R}^{\nu}%
|_{t=0}=\omega_{0}\,,\nonumber
\end{align}
where $\omega_{R}^{\nu}=\operatorname{curl}v_{R}^{\nu}$, $\operatorname{div}%
v_{R}^{\nu}=0$. On the solutions of this problem we want to perform
computations involving terms like $\Delta\mathcal{L}_{v}^{2}\omega\left(
t\right)  $, so we need $\omega\left(  t\right)  \in W^{4,2}\left(
\mathbb{T}^{3};\mathbb{R}^{3}\right)  $. This is why we introduce the strong
regularization $\nu\Delta^{5}\omega_{R}^{\nu}$; the precise power $5$ can be
understood from the technical computations of Step 1 below. While not optimal,
it a simple choice that allows us avoid more heavy arguments.

This regularized problem has the following property.

We understand equation (\ref{regularized SPDE}) either in the mild semigroup
sense (see below the proof) or in a weak sense over test functions, which are
equivalent due to the high regularity of solutions. However $\Delta^{5}%
\omega_{R}^{\nu}$ cannot be interpreted in a classical sense, since the solutions,
although very regular, will not be in $W^{10,2}\left(  \mathbb{T}%
^{3};\mathbb{R}^{3}\right)  $. The other terms of equation
(\ref{regularized SPDE}) can be interpreted in a classical sense.

\begin{lemma}
\label{lemma regularized equation}For every $\nu,R>0$ and $\omega_{0}\in
W_{\sigma}^{2,2}\left(  \mathbb{T}^{3},\mathbb{R}^{3}\right)  $, there exists
a pathwise unique global strong solution $\omega_{R}^{\nu}$, of class
$L^{2}\left(  \Xi;C\left(  \left[  0,T\right]  ;W_{\sigma}^{2,2}\left(
\mathbb{T}^{3};\mathbb{R}^{3}\right)  \right)  \right)  $ for every $T>0$. Its
paths have a.s. the additional regularity $C\left(  \left[  \delta,T\right]
;W^{4,2}\left(  \mathbb{T}^{3};\mathbb{R}^{3}\right)  \right)  $, for every
$T>\delta>0$.
\end{lemma}

\begin{proof}
\textbf{Step 1} (preparation). In the following we assume to have fixed $T>0$
and that all constants are generically denoted by $C>0$ any constant, with the
understanding that it may depend on $T$.

Let $D\left(  A\right)  =W_{\sigma}^{10,2}\left(  \mathbb{T}^{3}%
;\mathbb{R}^{3}\right)  $ and $A:D\left(  A\right)  \subset L_{\sigma}%
^{2}\left(  \mathbb{T}^{3},\mathbb{R}^{3}\right)  \rightarrow L_{\sigma}%
^{2}\left(  \mathbb{T}^{3},\mathbb{R}^{3}\right)  $ be the operator
$A\omega=\nu\Delta^{5}\omega$; $L_{\sigma}^{2}\left(  \mathbb{T}%
^{3},\mathbb{R}^{3}\right)  $ denotes here the closure of $D\left(  A\right)
$ in $L^{2}\left(  \mathbb{T}^{3},\mathbb{R}^{3}\right)  $ (the trace of the
periodic boundary condition at the level of $L^{2}$ spaces can be
characterized, see \cite{Temam}). The operator $A$ is self-adjoint and
negative definite. Let $e^{tA}$ be the semigroup in $L_{\sigma}^{2}\left(
\mathbb{T}^{3},\mathbb{R}^{3}\right)  $ generated by $A$. The fractional
powers $\left(  I-A\right)  ^{\alpha}$ are well defined, for every $\alpha>0$,
and are bi-continuous bijections between $W_{\sigma}^{\beta,2}\left(
\mathbb{T}^{3};\mathbb{R}^{3}\right)  $ and $W_{\sigma}^{\beta-10\alpha
,2}\left(  \mathbb{T}^{3};\mathbb{R}^{3}\right)  $, for every $\beta
\geq10\alpha$, in particular%
\[
\left\Vert f\right\Vert _{W^{10\alpha,2}}\le C_{\alpha}\left\Vert \left(
I-A\right)  ^{\alpha}f\right\Vert _{L^{2}}%
\]
for some constant $C_{\alpha}>0$, 
for all $f\in W_{\sigma}^{10\alpha,2}\left(
\mathbb{T}^{3};\mathbb{R}^{3}\right)  $.

In the sequel we write $\left\langle f,g\right\rangle =\int_{\mathbb{T}^{3}%
}f\left(  x\right)  \cdot g\left(  x\right)  dx$. We work on the torus, which
simplifies some definitions and properties; thus we write $\left(
1-\Delta\right)  ^{s/2}f$ for the function having Fourier transform $\left(
1+\left\vert \xi\right\vert ^{2}\right)  ^{s/2}\widehat{f}\left(  \xi\right)  $
($\widehat{f}\left(  \xi\right)  $ being the Fourier transform of
$f$);\ similarly we write $\Delta^{-1}f$ for the function having Fourier
transform $\left\vert \xi\right\vert ^{-1}\widehat{f}\left(  \xi\right)  $.

%\todo[inline]{\color{blue} DC: In Section 3 we defined the spaces $W^{s,p}$ %are as the spaces
%of functions $f$ such that $\left(  1-\Delta\right)  ^{s/2}f\in L^{p}$, %where $\left(
%1-\Delta\right)  ^{s/2}f$ us the function having Fourier transform $\left(
%1+\left\vert \xi\right\vert \right)  ^{s/2}\widehat{f}\left(  \xi\right)  %$
%($\widehat{f}\left(  \xi\right)  $ being the Fourier transform of
%$f$). Should we not then define the norm $\|f\|_{s,p}$ as
%\[
%\|f\|_{s,p}:=\|\left(1-\Delta\right)  ^{s/2}f\|_{L^p}
%\]
%In this case the above inequality is an identity with $C_\alpha=1$.
%}

The fractional powers commute with $e^{tA}$ and have the property (from the
general theory of analytic semigroups, see \cite{Pazy}) that for every
$\alpha>0$ and $T>0$
\[
\left\Vert \left(  I-A\right)  ^{\alpha}e^{tA}f\right\Vert _{L^{2}}\leq
\frac{C_{\alpha}}{t^{\alpha}}\left\Vert f\right\Vert _{L^{2}}%
\]
for all $t\in(0,T]$ and $f\in L_{\sigma}^{2}\left(  \mathbb{T}^{3}%
;\mathbb{R}^{3}\right)  $.

From these properties it follows that, for $p=2,4$%
\begin{equation}
\left\Vert \int_{0}^{t}e^{\left(  t-s\right)  A}f\left(  s\right)
ds\right\Vert _{W^{p,2}}^{2}\leq C\int_{0}^{t}\frac{1}{\left(  t-s\right)
^{p/10}}\left\Vert f\left(  s\right)  \right\Vert _{L^{2}}^{2}ds\leq
CT^{1-\frac{p}{10}}\sup_{t\in\left[  0,T\right]  }\left\Vert f\left(
s\right)  \right\Vert _{L^{2}}^{2}\label{regularity prop 1}%
\end{equation}
for all $f\in C\left(  \left[  0,T\right]  ;L_{\sigma}^{2}\left(
\mathbb{T}^{3};\mathbb{R}^{3}\right)  \right)  $ and $t\in\left[  0,T\right]
$, because
\[
\left\Vert \int_{0}^{t}e^{\left(  t-s\right)  A}f\left(  s\right)
ds\right\Vert _{W^{p,2}}\leq C\left\Vert \left(  I-A\right)  ^{p/10}\int%
_{0}^{t}e^{\left(  t-s\right)  A}f\left(  s\right)  ds\right\Vert _{L^{2}}\leq
C\int_{0}^{t}\frac{1}{\left(  t-s\right)  ^{p/10}}\left\Vert f\left(
s\right)  \right\Vert _{L^{2}}ds.
\]
\ In particular the map $f\mapsto\left(  \int_{0}^{t}e^{\left(  t-s\right)
A}f\left(  s\right)  ds\right)  _{t\in\left[  0,T\right]  }$ is linear
continuous from $C\left(  \left[  0,T\right]  ;L_{\sigma}^{2}\left(
\mathbb{T}^{3};\mathbb{R}^{3}\right)  \right)  $ to $C\left(  \left[
0,T\right]  ;W_{\sigma}^{2,2}\left(  \mathbb{T}^{3};\mathbb{R}^{3}\right)
\right)  $. Moreover, for $p=2,4$%
\begin{align}
\mathbb{E}\left[  \sup_{t\in\left[  0,T\right]  }\left\Vert \sum_{k=1}^{\infty}\int%
_{0}^{t}e^{\left(  t-s\right)  A}f_{k}\left(  s\right)  dB_{s}^{k}\right\Vert
_{W^{p,2}}^{2}\right]    & \leq C\mathbb{E}\left[  \sum_{k=1}^{\infty}\int_{0}^{t}%
\frac{1}{\left(  t-s\right)  ^{p/5}}\left\Vert f_{k}\left(  s\right)
\right\Vert _{L^{2}}^{2}ds\right]  \nonumber\\
& =CT^{1-p/5}\mathbb{E}\left[  \sup_{s\in\left[  0,T\right]  }\sum_{k=1}^{\infty
}\left\Vert f_{k}\left(  s\right)  \right\Vert _{L^{2}}^{2}\right]
\label{regularity prop 3}%
\end{align}%
\[
\]
because%
\begin{align*}
&  \mathbb{E}\left[  \sup_{t\in\left[  0,T\right]  }\left\Vert \sum_{k=1}^{\infty}%
\int_{0}^{t}e^{\left(  t-s\right)  A}f_{k}\left(  s\right)  dB_{s}%
^{k}\right\Vert _{W^{p,2}}^{2}\right]  =\mathbb{E}\left[  \sup_{t\in\left[  0,T\right]
}\left\Vert \sum_{k=1}^{\infty}\int_{0}^{t}\left(  I-A\right)  ^{p/10}%
e^{\left(  t-s\right)  A}f_{k}\left(  s\right)  dB_{s}^{k}\right\Vert _{L^{2}%
}^{2}\right]  \\
&  \leq C\mathbb{E}\left[  \sum_{k=1}^{\infty}\int_{0}^{t}\left\Vert \left(
I-A\right)  ^{p/10}e^{\left(  t-s\right)  A}f_{k}\left(  s\right)  \right\Vert
_{L^{2}}^{2}ds\right]  \leq C\mathbb{E}\left[  \sum_{k=1}^{\infty}\int_{0}^{t}\frac
{1}{\left(  t-s\right)  ^{p/5}}\left\Vert f_{k}\left(  s\right)  \right\Vert
_{L^{2}}^{2}ds\right]  .
\end{align*}
It is here that we use the power $5$ of $\Delta$, otherwise a smaller power
would suffice.

\textbf{Step 2} (preparation, cont.) The function $\omega\mapsto\kappa
_{R}\left(  \omega\right)  \mathcal{L}_{v}\omega$ from $W^{2,2}\left(
\mathbb{T}^{3};\mathbb{R}^{3}\right)  $ to $L^{2}\left(  \mathbb{T}%
^{3};\mathbb{R}^{3}\right)  $ is Lipschitz continuous and it has linear grows
(the constants in both properties depend on $R$). Let us check the Lipschitz
continuity; the linear growth is an easy consequence, applying Lipschitz
continuity with respect to a given element $\omega^{0}$.

%\todo[inline]{\color{blue} DC: Do we just need  r=R+2 ?
%Is $c_R=1$ ?  }

It is sufficient to check Lipschitz continuity in any ball $B\left(
0,r\right)  $, centered at the origin of radius $r$, in $W^{2,2}\left(
\mathbb{T}^{3};\mathbb{R}^{3}\right)  $. Indeed, when it is true, one can
argue as follows. Take $\omega^{\left(  i\right)  }$, $i=1,2$, in
$W^{2,2}\left(  \mathbb{T}^{3};\mathbb{R}^{3}\right)  $. If they belong to
$B\left(  0,R+2\right)  $, we have Lipschitz continuity. The case that both are outside $B\left(  0,R+2\right)  $ is trivial, because the cut-off function
vanishes. If one is inside $B\left(  0,R+2\right)  $ and the other outside,
consider the two cases:\ if the one inside is outside $B\left(  0,R+1\right)
$, it is trivial again, because the cut-off function vanishes for both
functions. If the one inside, say $\omega^{\left(  1\right)  }$, is in
$B\left(  0,R\right)  $, then $\mathcal{L}_{v^{\left(  1\right)  }}%
\omega^{\left(  1\right)  }\kappa_{R}\left(  \omega^{\left(  1\right)
}\right)  -\mathcal{L}_{v^{\left(  2\right)  }}\omega^{\left(  2\right)
}\kappa_{R}\left(  \omega^{\left(  2\right)  }\right)  =\mathcal{L}%
_{v^{\left(  1\right)  }}\omega^{\left(  1\right)  }\kappa_{R}\left(
\omega^{\left(  1\right)  }\right)  $;$\ $\ one has $\left\Vert \mathcal{L}%
_{v^{\left(  1\right)  }}\omega^{\left(  1\right)  }\kappa_{R}\left(
\omega^{\left(  1\right)  }\right)  \right\Vert _{L^{2}}\leq C_{R}$ (same
computations done below) and $\left\Vert \omega^{\left(  1\right)  }%
-\omega^{\left(  2\right)  }\right\Vert _{W^{2,2}}\geq c_{R}$, for two
constants $c_{R},C_{R}>0$, hence
\[
\left\Vert \mathcal{L}_{v^{\left(  1\right)  }}\omega^{\left(  1\right)
}\kappa_{R}\left(  \omega^{\left(  1\right)  }\right)  -\mathcal{L}%
_{v^{\left(  2\right)  }}\omega^{\left(  2\right)  }\kappa_{R}\left(
\omega^{\left(  2\right)  }\right)  \right\Vert _{L^{2}}\leq\frac{C_{R}}%
{c_{R}}\left\Vert \omega^{\left(  1\right)  }-\omega^{\left(  2\right)
}\right\Vert _{W^{2,2}}.
\]

Therefore, let us prove that the function $\omega\mapsto\kappa_{R}\left(
\omega\right)  \mathcal{L}_{v}\omega$ from $W^{2,2}\left(  \mathbb{T}%
^{3};\mathbb{R}^{3}\right)  $ to $L^{2}\left(  \mathbb{T}^{3};\mathbb{R}%
^{3}\right)  $ is Lipschitz continuous on $B\left(  0,r\right)  \subset
W^{2,2}\left(  \mathbb{T}^{3};\mathbb{R}^{3}\right)  $. Given $\omega^{\left(
i\right)  }\in B\left(  0,r\right)  $, $i=1,2$, let us use the decomposition%
\begin{align*}
&  \mathcal{L}_{v^{\left(  1\right)  }}\omega^{\left(  1\right)  }\kappa
_{R}\left(  \omega^{\left(  1\right)  }\right)  -\mathcal{L}_{v^{\left(
2\right)  }}\omega^{\left(  2\right)  }\kappa_{R}\left(  \omega^{\left(
2\right)  }\right) \\
&  =\mathcal{L}_{v^{\left(  1\right)  }}\left(  \omega^{\left(  1\right)
}-\omega^{\left(  2\right)  }\right)  \kappa_{R}\left(  \omega^{\left(
1\right)  }\right)  +\mathcal{L}_{\left(  v^{\left(  1\right)  }-v^{\left(
2\right)  }\right)  }\omega^{\left(  2\right)  }\kappa_{R}\left(
\omega^{\left(  2\right)  }\right)  +\mathcal{L}_{v^{\left(  1\right)  }%
}\omega^{\left(  2\right)  }\left(  \kappa_{R}\left(  \omega^{\left(
1\right)  }\right)  -\kappa_{R}\left(  \omega^{\left(  2\right)  }\right)
\right)  .
\end{align*}
Then%
\begin{align*}
&  \left\Vert \mathcal{L}_{v^{\left(  1\right)  }}\left(  \omega^{\left(
1\right)  }-\omega^{\left(  2\right)  }\right)  \kappa_{R}\left(
\omega^{\left(  1\right)  }\right)  \right\Vert _{L^{2}}^{2}\\
&  \leq\kappa_{R}\left(  \omega^{\left(  1\right)  }\right)  ^{2}\left\Vert
v^{\left(  1\right)  }\right\Vert _{\infty}^{2}\left\Vert \nabla\left(
\omega^{\left(  1\right)  }-\omega^{\left(  2\right)  }\right)  \right\Vert
_{L^{2}}^{2}+\kappa_{R}\left(  \omega^{\left(  1\right)  }\right)
^{2}\left\Vert \omega^{\left(  1\right)  }-\omega^{\left(  2\right)
}\right\Vert _{L^{4}}^{2}\left\Vert \nabla v^{\left(  1\right)  }\right\Vert
_{L^{4}}^{2}\\
&  \leq\kappa_{R}\left(  \omega^{\left(  1\right)  }\right)  ^{2}\left\Vert
\nabla v^{\left(  1\right)  }\right\Vert _{\infty}^{2}\left\Vert
\omega^{\left(  1\right)  }-\omega^{\left(  2\right)  }\right\Vert _{W^{2,2}%
}^{2}+\kappa_{R}\left(  \omega^{\left(  1\right)  }\right)  ^{2}\left\Vert
\omega^{\left(  1\right)  }-\omega^{\left(  2\right)  }\right\Vert _{W^{2,2}%
}^{2}\left\Vert \nabla v^{\left(  1\right)  }\right\Vert _{\infty}^{2}\\
&  \leq R^{2}\left\Vert \omega^{\left(  1\right)  }-\omega^{\left(  2\right)
}\right\Vert _{W^{2,2}}^{2};
\end{align*}
similarly%
\begin{align*}
&  \left\Vert \mathcal{L}_{\left(  v^{\left(  1\right)  }-v^{\left(  2\right)
}\right)  }\omega^{\left(  2\right)  }\kappa_{R}\left(  \omega^{\left(
2\right)  }\right)  \right\Vert _{L^{2}}^{2}\\
&  \leq\kappa_{R}\left(  \omega^{\left(  2\right)  }\right)  ^{2}\left\Vert
\nabla\left(  v^{\left(  1\right)  }-v^{\left(  2\right)  }\right)
\right\Vert _{\infty}^{2}\left\Vert \omega^{\left(  2\right)  }\right\Vert
_{W^{2,2}}^{2}\\
&  \leq Cr^{2}\left\Vert v^{\left(  1\right)  }-v^{\left(  2\right)
}\right\Vert _{W^{3,2}}^{2}\leq Cr^{2}\left\Vert \omega^{\left(  1\right)
}-\omega^{\left(  2\right)  }\right\Vert _{W^{2,2}}^{2}%
\end{align*}
by the Sobolev embedding theorem and (\ref{Biot Savart}). Finally%
\begin{align*}
&  \left\Vert \mathcal{L}_{v^{\left(  1\right)  }}\omega^{\left(  2\right)
}\left(  \kappa_{R}\left(  \omega^{\left(  1\right)  }\right)  -\kappa
_{R}\left(  \omega^{\left(  2\right)  }\right)  \right)  \right\Vert _{L^{2}%
}^{2}\\
&  \leq\left\vert \kappa_{R}\left(  \omega^{\left(  1\right)  }\right)
-\kappa_{R}\left(  \omega^{\left(  2\right)  }\right)  \right\vert
^{2}\left\Vert \nabla v^{\left(  1\right)  }\right\Vert _{\infty}%
^{2}\left\Vert \omega^{\left(  2\right)  }\right\Vert _{W^{2,2}}^{2}\\
&  \leq C\left\vert \kappa_{R}\left(  \omega^{\left(  1\right)  }\right)
-\kappa_{R}\left(  \omega^{\left(  2\right)  }\right)  \right\vert
^{2}\left\Vert \omega^{\left(  2\right)  }\right\Vert _{W^{2,2}}^{4}\\
&  \leq Cr^{4}\left\vert \kappa_{R}\left(  \omega^{\left(  1\right)  }\right)
-\kappa_{R}\left(  \omega^{\left(  2\right)  }\right)  \right\vert ^{2}%
\end{align*}
because $\left\Vert \nabla v^{\left(  1\right)  }\right\Vert _{\infty}^{2}\leq
C\left\Vert \omega^{\left(  2\right)  }\right\Vert _{W^{2,2}}^{2}$ as above,
and then we use the Lipschitz continuity of the function $\omega\mapsto
\kappa_{R}\left(  \omega\right)  $.

\textbf{Step 3} (local solution by fixed point). Given $\omega_{0}\in
L^{2}\left(  \Xi;W_{\sigma}^{2,2}\left(  \mathbb{T}^{3},\mathbb{R}^{3}\right)
\right)  $, $\mathcal{F}_{0}$-measurable, consider the mild equation%
\[
\omega\left(  t\right)  =\left(  \Gamma\omega\right)  \left(  t\right)
\]
where
\begin{align*}
\left(  \Gamma\omega\right)  \left(  t\right)   &  =e^{tA}\omega_{0}-\int%
_{0}^{t}e^{\left(  t-s\right)  A}\mathcal{L}_{v\left(  s\right)  }%
\omega\left(  s\right)  \kappa_{R}\left(  \omega\left(  s\right)  \right)
ds\\
&  +\int_{0}^{t}e^{\left(  t-s\right)  A}\frac{1}{2}\sum_{k=1}^{\infty
}\mathcal{L}_{\xi_{k}}^{2}\omega\left(  s\right)  ds-\sum_{k=1}^{\infty}%
\int_{0}^{t}e^{\left(  t-s\right)  A}\mathcal{L}_{\xi_{k}}\omega\left(
s\right)  dB_{s}^{k}\,
\end{align*}
with, as usual, $\operatorname{curl}v=\omega$, $\operatorname{div}v=0$. Set
$\mathcal{Y}_{T}:=L^{2}\left(  \Xi;C\left(  \left[  0,T\right]  ;W^{2,2}%
\left(  \mathbb{T}^{3};\mathbb{R}^{3}\right)  \right)  \right)  $. The map
$\Gamma$, applied to an element $\omega\in\mathcal{Y}_{T}$, gives us an
element $\Gamma\omega$ of the same space. Indeed:

i)\ $e^{tA}$ is bounded in $W^{2,2}\left(  \mathbb{T}^{3};\mathbb{R}%
^{3}\right)  $ (for instance because it commutes with $\left(  I-A\right)
^{1/5}$) hence $e^{tA}\omega_{0}$ is in $\mathcal{Y}_{T}$;

ii)\ $\mathcal{L}_{v}\omega\kappa_{R}\left(  \omega\right)  \in L^{2}\left(
\Xi;C\left(  \left[  0,T\right]  ;L_{\sigma}^{2}\left(  \mathbb{T}%
^{3};\mathbb{R}^{3}\right)  \right)  \right)  $ by Step 2, hence $\int_{0}%
^{t}e^{\left(  t-s\right)  A}\mathcal{L}_{v\left(  s\right)  }\omega\left(
s\right)  \kappa_{R}\left(  \omega\left(  s\right)  \right)  ds$ is an element
of $\mathcal{Y}_{T}$, by property (\ref{regularity prop 1}) of Step 1;

iii) $\sum_{k=1}^{\infty}\mathcal{L}_{\xi_{k}}^{2}\omega\in L^{2}\left(
\Xi;C\left(  \left[  0,T\right]  ;L^{2}\left(  \mathbb{T}^{3};\mathbb{R}%
^{3}\right)  \right)  \right)  $ from assumption (\ref{Hp 1}), hence $\int%
_{0}^{t}e^{\left(  t-s\right)  A}\frac{1}{2}\sum_{k=1}^{\infty}\mathcal{L}%
_{\xi_{k}}^{2}\omega\left(  s\right)  ds$ is in $\mathcal{Y}_{T}$ by property
(\ref{regularity prop 1});

iv) since, by assumption (\ref{Hp 2}),
\[
\sum_{k=1}^{\infty}\left\Vert \mathcal{L}_{\xi_{k}}\omega\left(  s\right)
\right\Vert _{L^{2}}^{2}\leq C\left\Vert \omega\left(  s\right)  \right\Vert
_{W^{2,2}}^{2}%
\]
we apply property (\ref{regularity prop 3}) and get that $\sum_{k=1}^{\infty
}\int_{0}^{t}e^{\left(  t-s\right)  A}\mathcal{L}_{\xi_{k}}\omega\left(
s\right)  dB_{s}^{k}$ is in $L^{2}\left(  \Xi;C\left(  \left[  0,T\right]
;W^{2,2}\left(  \mathbb{T}^{3};\mathbb{R}^{3}\right)  \right)  \right)  $.

The proof that $\Gamma$ is Lipschitz continuous in $\mathcal{Y}_{T}$ is based
on the same facts, in particular the Lipschitz continuity proved in Step 2.
Then, using the smallness of the constants for small $T$ in properties
(\ref{regularity prop 1}) and (\ref{regularity prop 3}) of Step 1, one gets
that $\Gamma$ is a contraction in $\mathcal{Y}_{T}$, for sufficiently small
$T>0$.

\textbf{Step 4} (\textit{a priori} estimate and global solution). The length
of the time interval of the local solution proved in Step 2 depends only on
the $L^{2}\left(  \Xi;W_{\sigma}^{2,2}\left(  \mathbb{T}^{3},\mathbb{R}%
^{3}\right)  \right)  $ norm of $\omega_{0}$. If we prove that, given $T>0$
and the initial condition $\omega_{0}$, there is a constant $C>0$ such that a
solution $\omega$ defined on $\left[  0,T\right]  $ has $\sup_{t\in\left[
0,T\right]  }\mathbb{E}\left[  \left\Vert \omega\left(  t\right)  \right\Vert
_{W^{2,2}}^{2}\right]  \leq C$, then we can repeatedly apply the local result
of Step 2 and cover any time interval. 

%\todo[inline]{\color{blue} DC: Do we not need the constant $C$ to be independent %of time ? Gronwall doesn't give us that though. }

Thus we need such \textit{a priori}
bound. Let $\omega$ be such a solution, namely satisfying $\omega=\Gamma
\omega$ on $\left[  0,T\right]  $. From the bounds (\ref{regularity prop 1})
and (\ref{regularity prop 3}) of Step 1, we have%
\begin{align*}
\mathbb{E}\left[  \left\Vert \omega\left(  t\right)  \right\Vert _{W^{2,2}}^{2}\right]
&  \leq C\mathbb{E}\left[  \left\Vert e^{tA}\omega_{0}\right\Vert _{W^{2,2}}%
^{2}\right]  +C\mathbb{E}\left[  \left\Vert \int_{0}^{t}e^{\left(  t-s\right)
A}\mathcal{L}_{v\left(  s\right)  }\omega\left(  s\right)  \kappa_{R}\left(
\omega\left(  s\right)  \right)  ds\right\Vert _{W^{2,2}}^{2}\right] \\
&  +C\mathbb{E}\left[  \left\Vert \int_{0}^{t}e^{\left(  t-s\right)  A}\frac{1}{2}%
\sum_{k=1}^{\infty}\mathcal{L}_{\xi_{k}}^{2}\omega\left(  s\right)
ds\right\Vert _{W^{2,2}}^{2}\right]  +C\mathbb{E}\left[  \left\Vert \sum_{k=1}^{\infty
}\int_{0}^{t}e^{\left(  t-s\right)  A}\mathcal{L}_{\xi_{k}}\omega\left(
s\right)  dB_{s}^{k}\right\Vert _{W^{2,2}}^{2}\right]
\end{align*}%
\[
\leq C\mathbb{E}\left[  \left\Vert \omega_{0}\right\Vert _{W^{2,2}}^{2}\right]
+C\mathbb{E}\left[  \int_{0}^{t}\frac{1}{\left(  t-s\right)  ^{2/5}}\left\Vert
\omega\left(  s\right)  \right\Vert _{W^{2,2}}^{2}ds\right]
\]
hence we may apply a generalized version of the Gr\"onwall lemma and conclude
that
\[
\sup_{t\in\left[  0,T\right]  }\mathbb{E}\left[  \left\Vert \omega\left(  t\right)
\right\Vert _{W^{2,2}}^{2}\right]  \leq C \,.
\]

\textbf{Step 5} (regularity). Let $\omega$ be the solution constructed in the
previous steps; it is the sum of the four terms given by the mild formulation
$\omega=\Gamma\omega$. {{}By the property $e^{tA}\omega_{0}\in D\left(  A\right)  $, namely
$Ae^{tA}\omega_{0}\in L_{\sigma}^{2}\left(  \mathbb{T}^{3},\mathbb{R}%
^{3}\right)  $, for all $t>0$ and $\omega_{0}\in L_{\sigma}^{2}\left(
\mathbb{T}^{3},\mathbb{R}^{3}\right)  $ (see [Pazy], property (5.7) in Theorem
5.2 of Chapter 2, due to the fact that $e^{tA}$ is an analytic semigroup), we
may take $\delta>0$ and have $Ae^{\delta A}\omega_{0}\in L_{\sigma}^{2}\left(
\mathbb{T}^{3},\mathbb{R}^{3}\right)  $; then, for $t\in\left[  \delta
,T\right]  $, we have $Ae^{tA}\omega_{0}=e^{\left(  t-\delta\right)
A}Ae^{\delta A}\omega_{0}=e^{\left(  t-\delta\right)  A}\omega_{\delta}$ where
$\omega_{\delta}:=Ae^{\delta A}\omega_{0}$ is an element of $L_{\sigma}%
^{2}\left(  \mathbb{T}^{3},\mathbb{R}^{3}\right)  $. Since $t\mapsto
e^{\left(  t-\delta\right)  A}\omega_{\delta}$ is continuous on $\left[
\delta,T\right]  $ (because the semigroup is strongly continuous), it follows
that $t\mapsto Ae^{tA}\omega_{0}$ is continuous on $\left[  \delta,T\right]
$, namely $t\mapsto e^{tA}\omega_{0}$ belongs to $C\left(  \left[
\delta,T\right]  ;D\left(  A\right)  \right)  $. In particular, it implies
$e^{tA}\omega_{0}\in C\left(  \left[  \delta,T\right]  ;W^{4,2}\left(
\mathbb{T}^{3},\mathbb{R}^{3}\right)  \right)  $, for every $T>\delta>0$.
} The two Lebesgue
integrals in $\Gamma\omega$ belong, pathwise a.s., to $C\left(  \left[
0,T\right]  ;W^{4,2}\left(  \mathbb{T}^{3};\mathbb{R}^{3}\right)  \right)  $,
for every $T>0$, because of property (\ref{regularity prop 1}), and the fact
that $\mathcal{L}_{v}\omega\kappa_{R}\left(  \omega\right)  $ and $\sum
_{k=1}^{\infty}\mathcal{L}_{\xi_{k}}^{2}\omega$ are, pathwise a.s., elements
of $C\left(  \left[  0,T\right]  ;L^{2}\left(  \mathbb{T}^{3};\mathbb{R}%
^{3}\right)  \right)  $, as showed in Step 2. Finally, the stochastic integral
in $\Gamma\omega$ belongs to $L^{2}\left(  \Xi;C\left(  \left[  0,T\right]
;W^{4,2}\left(  \mathbb{T}^{3};\mathbb{R}^{3}\right)  \right)  \right)  $ by
property (\ref{regularity prop 3}), and the fact that $\mathbb{E}\left[  \sup
_{s\in\left[  0,T\right]  }\sum_{k=1}^{\infty}\left\Vert \mathcal{L}_{\xi_{k}%
}\omega\left(  s\right)  \right\Vert _{L_{\sigma}^{2}}^{2}\right]  <\infty$,
as showed again in Step 2.
\end{proof}

\begin{definition}
On a complete separable metric space $\left(  X,d\right)  $, a family
$F=\left\{  \mu_{\nu}\right\}  _{\nu>0}$ of probability measures is called
\textit{tight} if for every $\epsilon>0$ there is a compact set $K_{\epsilon
}\subset X$ such that $\mu_{\nu}\left(  K_{\epsilon}\right)  \geq1-\epsilon$
for all $\nu>0$.
\end{definition}

\begin{remark}
The Prohorov theorem states that, for a tight family of probability measures,
one can extract a sequence $\left\{  \mu_{\nu_{n}}\right\}  _{n\in\mathbb{N}}$
which weakly converges to some probability measure,
\[
\mu:\ \lim_{n\rightarrow\infty}\int_{X}\varphi d\mu_{\nu_{n}}=\int_{X}\varphi
d\mu\,,
\]
for all bounded continuous functions $\varphi:X\rightarrow\mathbb{R}$. We
repeatedly use these facts below.
\end{remark}

In order to prove Proposition \ref{mtpart1}, we want to prove that the family
of solutions $\left\{  \omega_{R}^{\nu}\right\}  _{\nu>0}$ ($R$ is given)
provided by Lemma \ref{lemma regularized equation} is compact is a suitable
sense and that a converging subsequence extracted from this family converges
to a solution of equation (\ref{passed to limit}). Since $\left\{  \omega
_{R}^{\nu}\right\}  _{\nu>0}$ are random processes, the classical method we
follow is to prove compactness of their laws $\left\{  \mu_{\nu}\right\}
_{\nu>0}$. For this purpose, we have to prove that $\left\{  \mu_{\nu
}\right\}  _{\nu>0}$ is tight and we have to apply Prohorov theorem, as
recalled above. The metric space where we prove tightness of the laws will be
the space $E$ given by (\ref{spaces}) below.\footnote{The authors would like to thank Zdzislaw Brzezniak for pointing out to a gap in an earlier version of the tightness argument.} 

%For that purpose, we denote by
%$C_{w}\left(  \left[  0,T\right]  ;W_{\sigma}^{2,2}\left(  \mathbb{T}%
%^{3},\mathbb{R}^{3}\right)  \right)  $ the space $C\left(  \left[  0,T\right]
%;W_{\sigma}^{2,2}\left(  \mathbb{T}^{3},\mathbb{R}^{3}\right)  \right)  %$endowed with the weak-star topology in time.

\begin{lemma}\label{int}
Let $T>0$, $R>0$ and $\omega_{0}\in W_{\sigma}^{2,2}\left(  \mathbb{T}%
^{3},\mathbb{R}^{3}\right)  $ be given. Assume that the family of laws of
$\left\{  \omega_{R}^{\nu}\right\}  _{\nu>0}$ is tight in the space
\begin{equation}
E=%L^{2}\left(  0,T;W^{\beta,2}\left(  \mathbb{T}^{3},\mathbb{R}^{3}\right)\right)\cap
C\left(  \left[  0,T\right]  ;W_{\sigma}^{\beta,2}\left(
\mathbb{T}^{3},\mathbb{R}^{3}\right)  \right)  \label{spaces}%
\end{equation}
for some $\beta>\frac{3}{2}$ and satisfies, for some constant $C_R>0$,
\[
\mathbb{E}\left[\sup_{t\in [0,T]}\|\omega^\nu_R(t)\|^2_{W^{2,2}}\right]\le C_R
\]
for every $\nu>0$. Then the existence claim of Proposition
\ref{Prop cut off} holds true, and thus Theorem \ref{main theorem} is proved.
\end{lemma}

\begin{proof}
$\,$

\textbf{Step 1} (Gyongy-Krylov approach).\ We base our proof on classical
ingredients, but also on the following fact proved in \cite{GyoKry}, Lemma
1.1. Let $\left\{  Z_{n}\right\}  _{n\in\mathbb{N}}$ be a sequence of random
variables (r.v.) with values in a Polish space $\left(  E,d\right)  $ endowed
with the Borel $\sigma$-algebra $\mathcal{B}\left(  E\right)  $. Assume that
the family of laws of $\left\{  Z_{n}\right\}  _{n\in\mathbb{N}}$ is tight.
Moreover, assume that the limit in law of any pair $\left(  Z_{n_{j}^{\left(
1\right)  }},Z_{n_{j}^{\left(  2\right)  }}\right)  _{j\in\mathbb{N}}$ of
subsequences is a measure, on $E\times E$, supported on the diagonal of
$E\times E$. Then $\left\{  Z_{n}\right\}  _{n\in\mathbb{N}}$ converges in
probability to some r.v. $Z$.

We take as Polish space $E$ the space (\ref{spaces}) above, as random
variables $\left\{  Z_{n}\right\}  _{n\in\mathbb{N}}$ the sequence $\left\{
\omega_{R}^{1/n}\right\}  _{n\in\mathbb{N}}$, whose family of laws is tight by
assumption. We have to check that the limit in law of any pair $\left(
\omega_{R}^{1/n_{j}^{\left(  1\right)  }},\omega_{R}^{1/n_{j}^{\left(
2\right)  }}\right)  _{j\in\mathbb{N}}$ is supported on the diagonal of
$E\times E$. For this purpose, we shall use global uniqueness.\medskip

\textbf{Step 2} (Preparation by the Skorokhod theorem). Let us enlarge the
previous pair by the noise and consider the following triple: the sequence
$\left\{  \omega_{R}^{1/n_{j}^{\left(  1\right)  }},\omega_{R}^{1/n_{j}%
^{\left(  2\right)  }},\left\{  B_{\cdot}^{k}\right\}  _{k\in\mathbb{N}%
}\right\}  _{j\in\mathbb{N}}$ converges in law to a probability measure $\mu$,
on $E\times E\times C\left(  \left[  0,T\right]  \right)  ^{\mathbb{N}}$. We
have to prove that the marginal $\mu_{E\times E}$ of $\mu$ on $E\times E$ is
supported on the diagonal. By the Skorokhod representation theorem, there
exists a probability space $\left(  \widetilde{\Xi},\widetilde{\mathcal{F}%
},\widetilde{P}\right)  $ and $E\times E\times C\left(  \left[  0,T\right]
\right)  ^{\mathbb{N}}$-valued random variables $\left\{  \widetilde{\omega
}_{R}^{1,j},\widetilde{\omega}_{R}^{2,j},\left\{  \widetilde{B}_{\cdot}%
^{k,j}\right\}  _{k\in\mathbb{N}}\right\}  _{j\in\mathbb{N}}$ and $\left(
\widetilde{\omega}_{R}^{1},\widetilde{\omega}_{R}^{2},\left\{  \widetilde{B}%
_{\cdot}^{k}\right\}  _{k\in\mathbb{N}}\right)  $ with the same laws as
$\left\{  \omega_{R}^{1/n_{j}^{\left(  1\right)  }},\omega_{R}^{1/n_{j}%
^{\left(  2\right)  }},\left\{  B_{\cdot}^{k}\right\}  _{k\in\mathbb{N}%
}\right\}  _{j\in\mathbb{N}}$ and $\mu$, respectively,\footnote{In particular,
for each $j$, $\left\{  \widetilde{B}_{\cdot}^{k,j}\right\}  _{k\in\mathbb{N}%
}$ is a sequence of independent Brownian motions.} such that as $j\rightarrow
\infty$ one has $\widetilde{\omega}_{R}^{1,j}\rightarrow\widetilde{\omega}%
_{R}^{1}$ in $E$, $\widetilde{\omega}_{R}^{2,j}\rightarrow\widetilde{\omega
}_{R}^{2}$ in $E$, $\widetilde{B}_{\cdot}^{k,j}\rightarrow\widetilde{B}%
_{\cdot}^{k}$ in $C\left(  \left[  0,T\right]  \right)  $, a.s. In particular,
$\left\{  \widetilde{B}_{\cdot}^{k}\right\}  _{k\in\mathbb{N}}$ is a sequence
of independent Brownian motions.

Since the pairs $\left(  \omega_{R}^{1/n_{j}^{\left(  i\right)  }},\left\{
B_{\cdot}^{k}\right\}  _{k\in\mathbb{N}}\right)  $, $i=1,2$, solve equation
(\ref{regularized SPDE}) and $\left(  \widetilde{\omega}_{R}^{i,j},\left\{
\widetilde{B}_{\cdot}^{k,j}\right\}  _{k\in\mathbb{N}}\right)  $ have the same
laws (being marginals of vectors with the same laws), by a classical argument
(see for instance \cite{DaPratoZab}) the pairs $\left(  \widetilde{\omega}%
_{R}^{i,j},\left\{  \widetilde{B}_{\cdot}^{k,j}\right\}  _{k\in\mathbb{N}%
}\right)  $, $i=1,2$, also solve equation (\ref{regularized SPDE}), with
$\nu_{j}^{\left(  i\right)  }:=1/n_{j}^{\left(  i\right)  }$, $i=1,2$,
respectively. In other words,%
\begin{equation}
d\widetilde{\omega}_{R}^{i,j}+\kappa_{R}\left(  \widetilde{\omega}_{R}%
^{i,j}\right)  \mathcal{L}_{\widetilde{v}_{R}^{i,j}}\widetilde{\omega}%
_{R}^{i,j}\ dt+\sum_{k=1}^{\infty}\mathcal{L}_{\xi_{k}}\widetilde{\omega}%
_{R}^{i,j}d\widetilde{B}_{t}^{k,j}=\nu_{j}^{\left(  i\right)  }\Delta
^{5}\widetilde{\omega}_{R}^{i,j}dt+\frac{1}{2}\sum_{k=1}^{\infty}%
\mathcal{L}_{\xi_{k}}^{2}\widetilde{\omega}_{R}^{i,j}\ dt
\label{to be passed to limit}%
\end{equation}
with $\widetilde{\omega}_{R}^{i,j}|_{t=0}=\omega_{0}$, where
$\widetilde{\omega}_{R}^{i,j}=\operatorname{curl}\widetilde{v}_{R}^{i,j}$.
\medskip

\textbf{Step 3} (property of being supported on the diagonal). The passage to
the limit in equation \eqref{to be passed to limit} when there is strong
convergence ($\widetilde{P}$-a.s.) in $L^{2}\left(  0,T;L^{2}\left(
\mathbb{T}^{3},\mathbb{R}^{3}\right)  \right)  $ is relatively classical, see
\cite{FlaGat}. We sketch the main points in Step 4 below. One deduces%
\begin{equation}
d\widetilde{\omega}_{R}^{i}+\kappa_{R}\left(  \widetilde{\omega}_{R}%
^{i}\right)  \mathcal{L}_{\widetilde{v}_{R}^{i}}\widetilde{\omega}_{R}%
^{i}\ dt+\sum_{k=1}^{\infty}\mathcal{L}_{\xi_{k}}\widetilde{\omega}_{R}%
^{i}d\widetilde{B}_{t}^{k}=\frac{1}{2}\sum_{k=1}^{\infty}\mathcal{L}_{\xi_{k}%
}^{2}\widetilde{\omega}_{R}^{i}\ dt \label{passed to limit}%
\end{equation}
in the weak sense explained in Remark \ref{remark weak sol}. Since
$\widetilde{\omega}_{R}^{i}$ have paths in $C\left(  \left[  0,T\right]
;W^{2,2}\left(  \mathbb{T}^{3},\mathbb{R}^{3}\right)  \right)  $ (see Step 4 below), the
derivatives can be applied on $\widetilde{\omega}_{R}^{i}$ by integration by
parts and we get the equation in the strong sense. Now we apply the pathwise uniqueness of solutions for equation
(\ref{eq Euler Ito cutoff}) in $W^{2,2}$ as deduced in Section \ref{sect uniqueness W22} to deduce $\widetilde{\omega}%
_{R}^{1}=\widetilde{\omega}_{R}^{2}$. This means that the law of $\left(
\widetilde{\omega}_{R}^{1},\widetilde{\omega}_{R}^{2}\right)  $ is supported
on the diagonal of $E\times E$. Since this law is equal to $\mu_{E\times E}$,
we have that $\mu_{E\times E}$ is supported on the diagonal of $E\times E$.
\medskip

\textbf{Step 4} (convergence). In this step we give a few details about the
passage to the limit, as $j\rightarrow\infty$, from equation
(\ref{to be passed to limit}) to equation (\ref{passed to limit}). We do not
give the details about the linear terms, except for a comment about the term
$\nu_{j}^{\left(  i\right)  }\Delta^{5}\widetilde{\omega}_{R}^{i,j}$. Namely,
in weak form we write it as (with $\phi\in C^{\infty}\left(  \mathbb{T}%
^{3},\mathbb{R}^{3}\right)  $)%
\[
\nu_{j}^{\left(  i\right)  }\int_{0}^{t}\left\langle \widetilde{\omega}%
_{R}^{i,j}\left(  s\right)  ,\Delta^{5}\phi\right\rangle ds
\]
and use the pathwise convergence in $L^{2}\left(  0,T;L^{2}\left(
\mathbb{T}^{3},\mathbb{R}^{3}\right)  \right)  $ of $\widetilde{\omega}%
_{R}^{i,j}\left(  s\right)  $ plus the fact that $\nu_{j}^{\left(  i\right)
}\rightarrow0$.

The difficult term is the nonlinear one, also because of the cut-off term
$\kappa_{R}\left(  \widetilde{\omega}_{R}^{i}\left(  s\right)  \right)  $. We
want to prove that, given $\phi\in C^{\infty}\left(  \mathbb{T}^{3}%
,\mathbb{R}^{3}\right)  $,
\begin{equation}
\int_{0}^{t}\kappa_{R}\left(  \widetilde{\omega}_{R}^{i,j}\left(  s\right)
\right)  \left\langle \widetilde{\omega}_{R}^{i,j}\left(  s\right)
,\mathcal{L}_{\widetilde{v}_{R}^{i,j}}^{\ast}\phi\right\rangle
ds\overset{j\rightarrow\infty}{\rightarrow}\int_{0}^{t}\kappa_{R}\left(
\widetilde{\omega}_{R}^{i}\left(  s\right)  \right)  \left\langle
\widetilde{\omega}_{R}^{i}\left(  s\right)  ,\mathcal{L}_{\widetilde{v}%
_{R}^{i}}^{\ast}\phi\right\rangle ds \label{difficult limit}%
\end{equation}
with probability one. {{}From the Skorohod preparation in Step 2, we know that $\widetilde{\omega}%
_{R}^{i,j}\rightarrow\widetilde{\omega}_{R}^{i}$ as $j\rightarrow\infty$ in
the strong topology of $E$, $\widetilde{P}$-a.s., for $i=1,2$. In the sequel,
we fix the random parameter and the value of $i=1,2$. Since $W_{\sigma}%
^{\beta,2}\left(  \mathbb{T}^{3},\mathbb{R}^{3}\right)  $ is continuously
embedded into $C\left(  \mathbb{T}^{3},\mathbb{R}^{3}\right)  $ (recall that
$\beta>3/2$), it follows that $\widetilde{\omega}_{R}^{i,j}\rightarrow
\widetilde{\omega}_{R}^{i}$ in the uniform topology over $\left[  0,T\right]
\times\mathbb{T}^{3}$. By the continuity of Biot-Savart map from $W_{\sigma
}^{\beta,2}\left(  \mathbb{T}^{3},\mathbb{R}^{3}\right)  $ to $W_{\sigma
}^{\beta+1,2}\left(  \mathbb{T}^{3},\mathbb{R}^{3}\right)  $ and the formula
for $L_{\widetilde{v}_{R}^{i,j}}^{\ast}$ which contains first derivatives of
$\widetilde{v}_{R}^{i,j}$ we see that $L_{\widetilde{v}_{R}^{i,j}}^{\ast}%
\phi\rightarrow L_{\widetilde{v}_{R}^{i}}^{\ast}\phi$ in the strong topology
of $E$ again; and thus, again by Sobolev embedding, $L_{\widetilde{v}%
_{R}^{i,j}}^{\ast}\phi\rightarrow L_{\widetilde{v}_{R}^{i}}^{\ast}\phi$ in the
uniform topology over $\left[  0,T\right]  \times\mathbb{T}^{3}$. Hence
$\left\langle \widetilde{\omega}_{R}^{i,j}\left(  \cdot\right)
,L_{\widetilde{v}_{R}^{i,j}}^{\ast}\phi\right\rangle $ converges to
$\left\langle \widetilde{\omega}_{R}^{i}\left(  \cdot\right)
,L_{\widetilde{v}_{R}^{i}}^{\ast}\phi\right\rangle $ uniformly over $\left[
0,T\right]  $. Hence, if we prove that $k_{R}\left(  \widetilde{\omega}%
_{R}^{i,j}\left(  s\right)  \right)  \rightarrow k_{R}\left(
\widetilde{\omega}_{R}^{i}\left(  s\right)  \right)$} for a.e.
$s\in\left[  0,T\right]  $, and because these functions are bounded by 1, we
can take the limit in (\ref{difficult limit}). Therefore it remains to prove
that, $\widetilde{P}$-a.s., $\kappa_{R}\left(  \widetilde{\omega}_{R}%
^{i,j}\left(  s\right)  \right)  $ converges to $\kappa_{R}\left(
\widetilde{\omega}_{R}^{i}\left(  s\right)  \right)  $ for a.e. $s\in\left[
0,T\right]  $, or at least in probability w.r.t. time. This is true because
strong convergence in $L^{2}\left(  0,T\right)  $ in time implies convergence
in probability w.r.t. time; and we have strong convergence in $L^{2}\left(
0,T\right)  $, of $\kappa_{R}\left(  \widetilde{\omega}_{R}^{i,j}\left(
s\right)  \right)  $, because $\kappa_{R}$ is bounded continuous,
$\widetilde{\omega}_{R}^{i,j}$ converges strongly in $L^{2}\left(
0,T;W^{\beta,2}\left(  \mathbb{T}^{3},\mathbb{R}^{3}\right)  \right)  $, hence
$\nabla\widetilde{v}_{R}^{i,j}$ converges strongly in $L^{2}\left(
0,T;W^{\beta,2}\left(  \mathbb{T}^{3},\mathbb{R}^{3}\right)  \right)  $ hence
in $L^{2}\left(  0,T;C\left(  \mathbb{T}^{3},\mathbb{R}^{3}\right)  \right)  $
by Sobolev embedding theorem. Hence, $\kappa_{R}\left(  \widetilde{\omega}%
_{R}^{i,j}\right)  $ converges to $\kappa_{R}\left(  \widetilde{\omega}%
_{R}^{i}\right)  $ in probability w.r.t. time. Finally , from the integral identity  satisfied by the  limit process $\tilde \omega^i$, one can deduce that $\tilde \omega^i\in C\left(  \left[  0,T\right]  ;W^{2,2}\left(
\mathbb{T}^{3},\mathbb{R}^{3}\right)  \right) $ following the argument in \cite{KIM}. 
\end{proof}

Based on this lemma, we need to prove suitable bounds on $\left\{  \omega
_{R}^{\nu}\right\}  _{\nu>0}$.

\begin{theorem}
\label{thm tightness}Assume that, for some $N\geq0$ and $\alpha>{\frac14}$,%
\begin{equation}
\mathbb{E}\left[  \sup_{t\in\left[  0,T\right]  }\left\Vert \omega_{R}^{\nu
}\left(  t\right)  \right\Vert _{W^{2,2}}^{4}\right]  \leq C_{1}
\label{bound 1}%
\end{equation}
\begin{equation}
\mathbb{E}\int_{0}^{T}\int_{0}^{T}\frac{\left\Vert \omega_{R}^{\nu}\left(
t\right)  -\omega_{R}^{\nu}\left(  s\right)  \right\Vert _{W^{-N,2}}^{4}%
}{\left\vert t-s\right\vert ^{1+4\alpha}}dtds\leq C_{2} \label{bound 2}%
\end{equation}
for all $\nu\in\left(  0,1\right)  $. Then the assumptions of Lemma \ref{int} hold.

%family of laws of $\left\{\omega_{R}^{\nu}\right\}  _{\nu\in\left(  0,1\right)$ %is tight in the space
%$E$ given by (\ref{spaces}).
\end{theorem}

\begin{proof}
We shall use the following variant of Aubin-Lions lemma, that can be found in
Simon \cite{SIM}. Recall that, given an Hilbert space $W$,
a norm on $W^{\alpha,4}\left(  0,T;W\right)  $ is the fourth root of%
\[
\int_{0}^{T}\left\Vert f\left(  t\right)  \right\Vert _{W}^{4}dt+\int_{0}%
^{T}\int_{0}^{T}\frac{\left\Vert f\left(  t\right)  -f\left(  s\right)
\right\Vert _{W}^{4}}{\left\vert t-s\right\vert ^{1+4\alpha}}dtds.
\]
Assume that $V,H,W$ are separable Hilbert spaces with continuous dense
embedding $V\subset H\subset W$ such that there exists $\theta\in (0,1)$ and $M>0$ such that 
\[
\|v\|_H\le M\|v\|_V^{1-\theta}\|v\|_W^\theta
\]
for every $v\in V$. Assume that $V\subset H$ is a compact
embedding. Assume $\alpha>0$. Then%
\[
L^{\infty}\left(  0,T;V\right)  \cap W^{\alpha,4}\left(  0,T;W\right)
\]
is compactly embedded into $C\left(  [0,T];H\right)  $, see \cite{SIM}, Corollary 9. We apply it to the
spaces%
\[
H=W^{\beta,2}\left(  \mathbb{T}^{3},\mathbb{R}^{3}\right)  \text{, }%
V=W^{2,2}\left(  \mathbb{T}^{3},\mathbb{R}^{3}\right)  \text{, }%
W=W^{-N,2}\left(  \mathbb{T}^{3},\mathbb{R}^{3}\right)
\]
where $\beta\in\left(  \frac{3}{2},2\right)  $. The constraint $\beta<2$ is
imposed because we want to use the compactness of the embedding $W^{2,2}%
\left(  \mathbb{T}^{3},\mathbb{R}^{3}\right)  \subset W^{\beta,2}\left(
\mathbb{T}^{3},\mathbb{R}^{3}\right)  $. The constraint $\beta>\frac{3}{2}$ is
imposed because we want to use the embedding $W^{\beta,2}\left(
\mathbb{T}^{3},\mathbb{R}^{3}\right)  \subset C\left(  \mathbb{T}%
^{3},\mathbb{R}^{3}\right)  $.

Let $\left\{  Q_{\nu}\right\}  $ be the family of laws of $\left\{  \omega
_{R}^{\nu}\right\}  $, supported on
\[
E_{0}:=C\left(  \left[  0,T\right]  ;W^{2,2}\left(  \mathbb{T}^{3}%
,\mathbb{R}^{3}\right)  \right)  \cap W^{\alpha,4}\left(  0,T;W^{-N,2}\left(
\mathbb{T}^{3},\mathbb{R}^{3}\right)  \right)
\]
by the assumption of the theorem. We want to prove that $\left\{  Q_{\nu
}\right\}  $ is tight in $E$. The sets $K_{R_{1},R_{2},R_{3}}$ defined as%
\[
\left\{  f:\sup_{t\in\left[  0,T\right]  }\left\Vert f\left(  t\right)
\right\Vert _{W^{2,2}}^{2}\leq R_{1},\int_{0}^{T}\left\Vert f\left(  t\right)
\right\Vert _{W^{-N,2}}^{4}dt\leq R_{2},\int_{0}^{T}\int_{0}^{T}%
\frac{\left\Vert f\left(  t\right)  -f\left(  s\right)  \right\Vert
_{W^{-N,2}}^{4}}{\left\vert t-s\right\vert ^{1+4\alpha}}dtds\leq
R_{3}\right\}
\]
with $R_{1},R_{2},R_{3}>0$ are relatively compact in $E$. Let us prove that,
given $\epsilon>0$, there are $R_{1},R_{2},R_{3}>0$ such that
\[
Q_{\nu}\left(  K_{R_{1},R_{2},R_{3}}^{c}\right)  \leq\epsilon
\]
for all $\nu\in\left(  0,1\right)  $. We have%
\begin{align*}
Q_{\nu}\left(  \sup_{t\in\left[  0,T\right]  }\left\Vert f\left(  t\right)
\right\Vert _{W^{2,2}}^{2}>R_{1}\right)   &  =P\left(  \sup_{t\in\left[
0,T\right]  }\left\Vert \omega_{R}^{\nu}\left(  t\right)  \right\Vert
_{W^{2,2}}^{2}\right) \\
&  \leq\frac{\mathbb{E}\left[  \sup_{t\in\left[  0,T\right]  }\left\Vert
\omega_{R}^{\nu}\left(  t\right)  \right\Vert _{W^{2,2}}^{2}\right]  }{R_{1}%
}\leq\frac{C_{1}}{R_{1}}%
\end{align*}
and this is smaller than ${\epsilon}/{3}$ when $R_{1}$ is large enough.
Similarly we get%
\[
Q_{\nu}\left(  \int_{0}^{T}\int_{0}^{T}\frac{\left\Vert f\left(  t\right)
-f\left(  s\right)  \right\Vert _{W^{-N,2}}^{4}}{\left\vert t-s\right\vert
^{1+4\alpha}}dtds>R_{3}\right)  \leq\frac{\epsilon}{3}%
\]
when $R_{3}$ is large enough. Finally,
\begin{align*}
Q_{\nu}\left(  \int_{0}^{T}\left\Vert f\left(  t\right)  \right\Vert
_{W^{-N,2}}^{4}dt>R_{2}\right)   &  \leq Q_{\nu}\left(  T\sup_{t\in\left[
0,T\right]  }\left\Vert f\left(  t\right)  \right\Vert _{W^{-N,2}}^{4}%
dt>R_{2}\right) \\
&  \leq Q_{\nu}\left(  CT\sup_{t\in\left[  0,T\right]  }\left\Vert f\left(
t\right)  \right\Vert _{W^{2,2}}^{4}dt>R_{2}\right)
\end{align*}
for a constant $C>0$ such that $\left\Vert f\left(  t\right)  \right\Vert
_{W^{-N,2}}^{4}\leq C\left\Vert f\left(  t\right)  \right\Vert _{W^{2,2}}^{4}%
$. Hence also this quantity is smaller than $\frac{\epsilon}{3}$ when $R_{2}$
is large enough. We deduce $Q_{\nu}\left(  K_{R_{1},R_{2},R_{3}}^{c}\right)
\leq\epsilon$ and complete the proof.
\end{proof}

The difficult part of the estimates above is bound (\ref{bound 1}). Thus, let
us postpone it and first show bound (\ref{bound 2}).

\section{Technical results}

\label{sect tech results}

\subsection{Fractional Sobolev regularity in time}

\label{sect frac So reg}

In this section we show that bound (\ref{bound 2}), with $N=1$, follows from
(an easier version of) bound (\ref{bound 1}).

\begin{lemma}
Assume%
\[
\sup_{t\in\left[  0,T\right]  }\mathbb{E}\left[  \left\Vert \omega_{R}^{\nu
}\left(  t\right)  \right\Vert _{W^{2,2}}^{4}\right]  \leq C.
\]
Then the bound in (\ref{bound 2}), with $N=3$ and any $\alpha<\frac{1}{2}$,
holds true.
\end{lemma}

\begin{proof}
\textbf{Step 1} (Preparation). In the sequel, we take $t\geq s$ and denote by
$C>0$ any constant. From equation (\ref{regularized SPDE}) we have%
\begin{align*}
\omega_{R}^{\nu}\left(  t\right)  -\omega_{R}^{\nu}\left(  s\right)   &
=-\int_{s}^{t}\sum_{k=1}^{\infty}\mathcal{L}_{\xi_{k}}\omega_{R}^{\nu}\left(
r\right)  dB_{r}^{k}\\
&  +\int_{s}^{t}\left(  \nu\Delta^{5}\omega_{R}^{\nu}\left(  r\right)
+\frac{1}{2}\sum_{k=1}^{\infty}\mathcal{L}_{\xi_{k}}^{2}\omega_{R}^{\nu
}\left(  r\right)  -\kappa_{R}\left(  \omega_{R}^{\nu}\left(  r\right)
\right)  \mathcal{L}_{v_{R}^{\nu}\left(  r\right)  }\omega_{R}^{\nu}\left(
r\right)  \right)  \ dr
\end{align*}
hence
\begin{align*}
\mathbb{E}\left[  \left\Vert \omega_{R}^{\nu}\left(  t\right)  -\omega
_{R}^{\nu}\left(  s\right)  \right\Vert _{W^{-3,2}}^{4}\right]   &
\leq C\left(  t-s\right)^{3}  \int_{s}^{t}\mathbb{E}\left[  \kappa_{R}\left(
\omega_{R}^{\nu}\left(  r\right)  \right)^4  \left\Vert \mathcal{L}_{v_{R}^{\nu
}\left(  r\right)  }\omega_{R}^{\nu}\left(  r\right)  \right\Vert _{W^{-3,2}%
}^{4}\right]  dr\\
&  +C\left(  t-s\right) ^{3} \int_{s}^{t}\mathbb{E}\left[  \left\Vert \nu
\Delta^{5}\omega_{R}^{\nu}\left(  r\right)  \right\Vert _{W^{-3,2}}%
^{4}\right]  dr\\
&  +C\left(  t-s\right) ^{3} \int_{s}^{t}\mathbb{E}\left[  \left\Vert \frac{1}%
{2}\sum_{k=1}^{\infty}\mathcal{L}_{\xi_{k}}^{2}\omega_{R}^{\nu}\left(
r\right)  \right\Vert _{W^{-3,2}}^{4}\right]  dr\\
&  +C\mathbb{E}\left[  \left\Vert \int_{s}^{t}\sum_{k=1}^{\infty}%
\mathcal{L}_{\xi_{k}}\omega_{R}^{\nu}\left(  r\right)  dB_{r}^{k}\right\Vert
_{W^{-3,2}}^{4}\right]  .
\end{align*}
Recall that $\left\Vert f\right\Vert _{W^{-3,2}}\leq C\left\Vert f\right\Vert
_{L^{2}}$, which follows by duality from $\left\Vert f\right\Vert _{L^{2}}\leq
C\left\Vert f\right\Vert _{W^{3,2}}$. Hence, again denoting any of the
constants in the calculation below as $C>0$, we have%
\begin{align*}
\mathbb{E}\left[  \left\Vert \omega_{R}^{\nu}\left(  t\right)  -\omega
_{R}^{\nu}\left(  s\right)  \right\Vert _{W^{-3,2}}^{4}\right]   &  \leq
C\left(  t-s\right)^3  \int_{s}^{t}\mathbb{E}\left[  \kappa_{R}\left(
\omega_{R}^{\nu}\left(  r\right)  \right)^4  \left\Vert \mathcal{L}_{v_{R}^{\nu
}\left(  r\right)  }\omega_{R}^{\nu}\left(  r\right)  \right\Vert _{W^{-3,2}%
}^{4}\right]  dr\\
&  +C\left(  t-s\right)^3  \int_{s}^{t}\mathbb{E}\left[  \left\Vert \Delta
^{5}\omega_{R}^{\nu}\left(  r\right)  \right\Vert _{W^{-3,2}}^{4}\right]  dr\\
&  +C\left(  t-s\right)^3  \int_{s}^{t}\mathbb{E}\left[  \left\Vert \frac{1}%
{2}\sum_{k=1}^{\infty}\mathcal{L}_{\xi_{k}}^{2}\omega_{R}^{\nu}\left(
r\right)  \right\Vert _{L^{2}}^{4}\right]  dr\\
&  +C\mathbb{E}\left[  \left\Vert \int_{s}^{t}\sum_{k=1}^{\infty}%
\mathcal{L}_{\xi_{k}}\omega_{R}^{\nu}\left(  r\right)  dB_{r}^{k}\right\Vert
_{L^{2}}^{4}\right]  .
\end{align*}
The only term where $W^{-3,2}$ is necessary is the term $\left\Vert \Delta
^{5}\omega_{R}^{\nu}\left(  r\right)  \right\Vert _{W^{-3,2}}^{4}$; we keep it
also in the first term, but this is not essential. Now let us estimate each term.

\textbf{Step 2} (Estimates of the deterministic terms). We have%
\begin{align*}
\left\vert \left\langle v_{R}^{\nu}\cdot\nabla\omega_{R}^{\nu},\phi
\right\rangle \right\vert  &  =\left\vert \left\langle \omega_{R}^{\nu}%
,v_{R}^{\nu}\cdot\nabla\phi\right\rangle \right\vert \\
&  \leq\left\Vert \phi\right\Vert _{W^{1,2}}\left\Vert v_{R}^{\nu}\right\Vert
_{L^{2}}\left\Vert \omega_{R}^{\nu}\right\Vert _{L^{\infty}}\\
&  \leq C\left\Vert \phi\right\Vert _{W^{1,2}}\left\Vert \omega_{R}^{\nu
}\right\Vert _{L^{\infty}}\left\Vert \omega_{R}^{\nu}\right\Vert _{L^{2}}\\
&  \leq C\left\Vert \phi\right\Vert _{W^{1,2}}\left\Vert \omega_{R}^{\nu
}\right\Vert _{L^{\infty}}\left\Vert \omega_{R}^{\nu}\right\Vert _{W^{2,2}}\,,
\end{align*}
so that%
\[
\left\Vert v_{R}^{\nu}\cdot\nabla\omega_{R}^{\nu}\right\Vert _{W^{-3,2}}%
^{2}\leq C\left\Vert \omega_{R}^{\nu}\right\Vert _{L^{\infty}}^{2}\left\Vert
\omega_{R}^{\nu}\right\Vert _{W^{2,2}}^{2}.
\]
Moreover, also%
\[
\left\Vert \omega_{R}^{\nu}\cdot\nabla v_{R}^{\nu}\right\Vert _{L^{2}}%
\leq\left\Vert \omega_{R}^{\nu}\right\Vert _{\infty}\left\Vert \nabla
v_{R}^{\nu}\right\Vert _{L^{2}}\leq C\left\Vert \omega_{R}^{\nu}\right\Vert
_{L^{\infty}}^{2}\left\Vert \omega_{R}^{\nu}\right\Vert _{W^{2,2}}^{2}.
\]
Summarizing
\[
\left\Vert \mathcal{L}_{v_{R}^{\nu}\left(  r\right)  }\omega_{R}^{\nu}\left(
r\right)  \right\Vert _{W^{-3,2}}^{2}\leq C\left\Vert \omega_{R}^{\nu
}\right\Vert _{L^{\infty}}^{2}\left\Vert \omega_{R}^{\nu}\right\Vert
_{W^{2,2}}^{2}.
\]
Therefore%
\[
\kappa_{R}^4\left(  \omega_{R}^{\nu}\left(  r\right)  \right)  \left\Vert
\mathcal{L}_{v_{R}^{\nu}\left(  r\right)  }\omega_{R}^{\nu}\left(  r\right)
\right\Vert _{W^{-3,2}}^{4}\leq C\left\Vert \omega_{R}^{\nu}\right\Vert
_{W^{2,2}}^{4}.
\]

For the next term, we have%
\[
\left\Vert \Delta^{5}\omega_{R}^{\nu}\right\Vert _{H^{-3}}^{4}\leq C\left\Vert
\omega_{R}^{\nu}\right\Vert _{W^{2,2}}^{4}%
\]
hence%
\[
\int_{s}^{t}\mathbb{E}\left[  \left\Vert \Delta^{5}\omega_{R}^{\nu}\right\Vert
_{W^{-3,2}}^{4}\right]  dr\leq C\int_{s}^{t}\mathbb{E}\left[  \left\Vert
\omega_{R}^{\nu}\right\Vert _{W^{2,2}}^{4}\right]  dr\leq C
\]
because we have the property $\sup_{t\in\left[  0,T\right]  }\mathbb{E}\left[
\left\Vert \omega_{R}^{\nu}\left(  t\right)  \right\Vert _{W^{2,2}}%
^{4}\right]  \leq C$.

For the subsequent term we have%
\[
\left\Vert \sum_{k=1}^{\infty}\mathcal{L}_{\xi_{k}}^{2}\omega_{R}^{\nu}\left(
r\right)  \right\Vert _{L^{2}}^{2}\leq C\left\Vert \omega_{R}^{\nu}\left(
r\right)  \right\Vert _{W^{2,2}}^{2}%
\]
by assumption (\ref{Hp 1}) and therefore, %
\[
\int_{s}^{t}\mathbb{E}\left[  \left\Vert \frac{1}{2}\sum_{k=1}^{\infty
}\mathcal{L}_{\xi_{k}}^{2}\omega_{R}^{\nu}\left(  r\right)  \right\Vert
_{L^{2}}^{4}\right]  dr\leq C\int_{s}^{t}\mathbb{E}\left[  \left\Vert
\omega_{R}^{\nu}\left(  r\right)  \right\Vert _{W^{2,2}}^{4}\right]  dr\leq C
\]
as above.

\textbf{Step 3} (Estimate of the stochastic term). One has, by the Burkholder-Davis-Gundy inequality%
\begin{align*}
\mathbb{E}\left[  \left\Vert \int_{s}^{t}\sum_{k=1}^{\infty}\mathcal{L}%
_{\xi_{k}}\omega_{R}^{\nu}\left(  r\right)  dB_{r}^{k}\right\Vert _{L^{2}}%
^{4}\right]   &  \le C\mathbb{E}\left[  \left(\int_{s}^{t}\sum_{k=1}^{\infty}\left\Vert
\mathcal{L}_{\xi_{k}}\omega_{R}^{\nu}\left(  r\right)  \right\Vert _{L^{2}%
}^{2}\right)^2\right]  dr\\
&  \leq C(t-s)\int_{s}^{t}\mathbb{E}\left[  \left\Vert \omega_{R}^{\nu}\left(
r\right)  \right\Vert _{W^{2,2}}^{4}\right]  dr
\end{align*}
by assumption (\ref{Hp 2}),
\[
\leq C\mathbb{E}\left[  \left\Vert \omega_{0}^{\nu}\right\Vert _{W^{2,2}}%
^{4}\right]  \left(  t-s\right)^2
\]
by the assumption of this lemma.

\textbf{Step 4} (Conclusion). From the previous steps we have%
\[
\mathbb{E}\left[  \left\Vert \omega_{R}^{\nu}\left(  t\right)  -\omega
_{R}^{\nu}\left(  s\right)  \right\Vert _{W^{-3,2}}^{4}\right]  \leq C\left(
t-s\right) ^2 .
\]
Hence%
\[
\mathbb{E}\left[  \int_{0}^{T}\int_{0}^{T}\frac{\left\Vert \omega_{R}^{\nu
}\left(  t\right)  -\omega_{R}^{\nu}\left(  s\right)  \right\Vert _{W^{-3,2}%
}^{4}}{\left\vert t-s\right\vert ^{1+4\alpha}}dtds\right]  \leq%
  \int_{0}^{T}\int_{0}^{T}\frac{C}{\left\vert t-s\right\vert ^{4\alpha-1}%
}dtds  \leq C
\]
for all $\alpha<\frac{1}{2}$.
\end{proof}

\subsection{Some \textit{a priori} estimates}

\label{sect bounds}

In order to complete the proof of Theorem \ref{main theorem}, we still need to
prove estimate (\ref{bound 1}). To be more explicit, since now a long and
difficult computation starts, what we have to prove is that, given $R>0$,
called for every $\nu\in\left(  0,1\right)  $ by $\omega_{R}^{\nu}$ the
solution of equation%
\[
d\omega_{R}^{\nu}+\kappa_{R}\left(  \omega_{R}^{\nu}\right)  \mathcal{L}%
_{v_{R}^{\nu}}\omega_{R}^{\nu}\ dt+\sum_{k=1}^{\infty}\mathcal{L}_{\xi_{k}%
}\omega_{R}^{\nu}dB_{t}^{k}=\nu\Delta^{5}\omega_{R}^{\nu}dt+\frac{1}{2}%
\sum_{k=1}^{\infty}\mathcal{L}_{\xi_{k}}^{2}\omega_{R}^{\nu}\ dt
\]
with $\omega_{R}^{\nu}|_{t=0}=\omega_{0}$, there is a constant $C>0$ such
that
\[
\mathbb{E}\left[  \sup_{t\in\left[  0,T\right]  }\left\Vert \omega_{R}^{\nu
}\left(  t\right)  \right\Vert _{W^{2,2}}^{4}\right]  \leq C
\]
for every $\nu\in\left(  0,1\right)  $.

In order to simplify notations, we shall simply write
\begin{align*}
&  \omega\text{ for }\omega_{R}^{\nu}\\
&  v\text{ for }v_{R}^{\nu}\\
&  \kappa\text{ for }\kappa_{R}\left(  \omega_{R}^{\nu}\right)
\end{align*}
not forgetting that all bounds have to be uniform in $\nu\in\left(
0,1\right)  $.

\medskip\paragraph{\bf Difficulty compared to the deterministic case.}

\label{subsect difficulty}

In the deterministic case $\frac{d}{dt}\int_{\mathbb{T}^{3}}\left\vert
\Delta\omega\left(  t,x\right)  \right\vert ^{2}dx$ is equal to the sum of
several terms. Using Sobolev embedding theorems (\ref{Sobolev embedding}) one
can estimate all terms as
\[
\leq C\int_{\mathbb{T}^{3}}\left\vert \Delta\omega\left(  t,x\right)
\right\vert ^{2}dx\,,
\]
except for the term with higher order derivatives%
\[
\int_{\mathbb{T}^{3}}\left(  v\cdot\nabla\Delta\omega\right)  \cdot
\Delta\omega\ dx\,.
\]
However, this term vanishes, being equal to
\[
\frac{1}{2}\int_{\mathbb{T}^{3}}\left(  v\cdot\nabla\right)  \left\vert
\Delta\omega\right\vert ^{2}\ dx=-\frac{1}{2}\int_{\mathbb{T}^{3}}\left\vert
\Delta\omega\right\vert ^{2}\operatorname{div}v\ dx=0.
\]
In the stochastic case, though, we have many more terms, coming from two sources:\ 

i) the term $\frac{1}{2}\sum_{k}\mathcal{L}_{\xi_{k}}^{2}\omega\ dt$, which is
a second order differential operator in $\omega$, hence much more demanding
than the deterministic term $\mathcal{L}_{v}\omega$;

ii) the It\^{o} correction term in It\^{o} formula for $d\int_{\mathbb{T}^{3}%
}\left\vert \Delta\omega\left(  t,x\right)  \right\vert ^{2}dx$.

A quick inspection in these additional terms immediately reveals that the
highest order terms compensate (one from (i) and the other from (ii)) and
cancel each other. These terms are of ``order 6" in the sense that, globally
speaking, they involve 6 derivatives of $\omega$. The new outstanding problem
is that there remains a large amount of terms of ``order 5", hence not bounded
by $C\int_{\mathbb{T}^{3}}\left\vert \Delta\omega\left(  t,x\right)
\right\vert ^{2}dx$ (which is of ``order 4"). After a few computations one is
na\"{\i}vely convinced that these terms are too numerous to compensate and
cancel one another.

But this is not true. A careful algebraic manipulation of differential
operators, as well as their commutators and adjoint operators, finally shows
that all terms of ``order 5" do cancel each other. At the end we are able to
estimate remaining terms again by $C\int_{\mathbb{T}^{3}}\left\vert
\Delta\omega\left(  t,x\right)  \right\vert ^{2}dx$ (now in expectation) and
obtain the \textit{a priori} estimates we seek.

\medskip\paragraph{\bf Preparatory remarks.}

By again using the regularity result of Lemma \ref{lemma regularized equation}%
, we may write the identity%
\begin{align*}
\Delta\omega\left(  t\right)   &  =\Delta\omega_{0}+A\int_{0}^{t}\Delta\omega\left(
s\right)  ds-\int_{0}^{t}\kappa\left(  s\right)  \Delta\mathcal{L}_{v}%
\omega\left(  s\right)  ds\\
&  +\frac{1}{2}\int_{0}^{t}\Delta\sum_{k=1}^{\infty}\mathcal{L}_{\xi_{k}}%
^{2}\omega\left(  s\right)  ds\\
&  -\sum_{k=1}^{\infty}\int_{0}^{t}\Delta\mathcal{L}_{\xi_{k}}\omega\left(
s\right)  dB_{s}^{k}%
\end{align*}
and\ we may apply a suitable It\^{o} formula in the Hilbert space
$L^{2}\left(  \mathbb{T}^{3}\right)  $ (see \cite{KrylovRoz}) to obtain%
\begin{align}
\frac{1}{2}d\int_{\mathbb{T}^{3}}\left\vert \Delta\omega\left(  t,x\right)
\right\vert ^{2}dx+\nu\int_{\mathbb{T}^{3}}\left\vert \Delta^{2}%
\omega\right\vert ^{2}dxdt  &  =-\kappa\left(  t\right)  \left(
\int_{\mathbb{T}^{3}}\Delta\mathcal{L}_{v}\omega\cdot\Delta\omega dx\right)
dt\nonumber\\
&  +\frac{1}{2}\sum_{k=1}^{\infty}\left(  \int_{\mathbb{T}^{3}}\Delta
\mathcal{L}_{\xi_{k}}^{2}\omega\cdot\Delta\omega dx\right)  dt\nonumber\\
&  -\sum_{k=1}^{\infty}\left(  \int_{\mathbb{T}^{3}}\Delta\mathcal{L}_{\xi
_{k}}\omega\cdot\Delta\omega dx\right)  dB_{t}^{k}\nonumber\\
&  +\text{It\^{o} correction}. \label{imp}%
\end{align}
Being $\Delta\omega\left(  t\right)  $ of the form $d\left(  \Delta
\omega\left(  t,x\right)  \right)  =a_{t}\left(  x\right)  dt+\sum
_{k=1}^{\infty}b_{t}^{k}\left(  x\right)  dB_{t}^{k}$, with $b_{t}^{k}\left(
x\right)  =-\Delta\mathcal{L}_{\xi_{k}}\omega\left(  t\right)  $, one has
\[
d\frac{1}{2}\left\vert \Delta\omega\left(  t,x\right)  \right\vert ^{2}%
=\Delta\omega\left(  t,x\right)  \cdot d\left(  \Delta\omega\left(
t,x\right)  \right)  +\frac{1}{2}\sum_{k=1}^{\infty}\left\vert b_{t}%
^{k}\left(  x\right)  \right\vert ^{2}dt
\]
hence the It\^{o} correction above is given by (we have to integrate in $dx$
the previous identity)
\[
\text{It\^{o} correction}=\frac{1}{2}\sum_{k=1}^{\infty}\int_{\mathbb{T}^{3}%
}\left\vert \Delta\mathcal{L}_{\xi_{k}}\omega\left(  t\right)  \right\vert
^{2}dxdt.
\]
Let us list the main considerations about the identity \eqref{imp}.

\begin{enumerate}
\item The term $\nu\int_{\mathbb{T}^{3}}\left\vert \Delta^{2}\omega\right\vert
^{2}dx$ will not be used in the estimates, since they have to be independent
of $\nu$; we only use the fact that this term has the right sign.

\item The term
\begin{equation}
\kappa\left(  t\right)  \int_{\mathbb{T}^{3}}\Delta\mathcal{L}_{v}\omega
\cdot\Delta\omega dx \label{classical term}%
\end{equation}
can be estimated by $C\int_{\mathbb{T}^{3}}\left\vert \Delta\omega\left(
t,x\right)  \right\vert ^{2}dx$ exactly as in the deterministic theory. The
computations are given in subsection \ref{subsect estimate classical term} below.

\item The term $\sum_{k=1}^{\infty}\left(  \int_{\mathbb{T}^{3}}%
\Delta\mathcal{L}_{\xi_{k}}\omega\cdot\Delta\omega dx\right)  dB_{t}^{k}$ is a
local martingale. Rigorously, we shall introduce a sequence of stopping times
and then, taking expectation, this term will disappear. Then the stopping
times will be removed by a straightforward limit.

\item The main difficulty comes from the term%
\begin{equation}
\frac{1}{2}\sum_{k=1}^{\infty}\left(  \left\langle \Delta\mathcal{L}_{\xi_{k}%
}^{2}\omega,\Delta\omega\right\rangle +\left\langle \Delta\mathcal{L}_{\xi
_{k}}\omega,\Delta\mathcal{L}_{\xi_{k}}\omega\right\rangle \right)  ,
\label{difficult term}%
\end{equation}
since it includes, as mentioned above in section \ref{subsect difficulty},
various terms which are of \textquotedblleft order 6" and of \textquotedblleft
order 5", where\textquotedblleft order\textquotedblright\ means the global
number of spatial derivatives. These terms cannot be estimated by
$C\int_{\mathbb{T}^{3}}\left\vert \Delta\omega\left(  t,x\right)  \right\vert
^{2}dx$. As it turns out, the terms of \textquotedblleft order 6" cancel each
other: this is straightforward and expected. But a large number of intricate
terms of \textquotedblleft order 5" still remain, which, na\"{\i}vely, may
give the impression that the estimate cannot be closed. On the contrary,
though, they also cancel each other: this is the content of section
\ref{sect basic bounds}, summarised in assumption (\ref{assump}).
\end{enumerate}

\medskip\paragraph{\bf Estimate of the classical term (\ref{classical term}).}

\label{subsect estimate classical term} The following lemma deals with the
control of the classical term (\ref{classical term}).

\begin{lemma}
\label{lemma GN estimate}Given $u\in W_{\sigma}^{3,2},\omega\in W_{\sigma
}^{2,2}$ (not necessarily related by $\operatorname{curl}u=\omega$), one has%
\begin{align*}
\left\vert \int_{\mathbb{T}^{3}}\Delta\mathcal{L}_{u}\omega\cdot\Delta\omega
dx\right\vert  &  \leq C\left\Vert \nabla u\right\Vert _{L^{\infty}}\left\Vert
\omega\right\Vert _{W^{2,2}}^{2}+C\left\Vert \omega\right\Vert _{L^{\infty}%
}\left\Vert \nabla u\right\Vert _{W^{2,2}}\left\Vert \omega\right\Vert
_{W^{2,2}}\\
&  \leq C\left(  \left\Vert \nabla v\right\Vert _{L^{\infty}}+\left\Vert
\omega\right\Vert _{L^{\infty}}\right)  \left\Vert \omega\right\Vert
_{W^{2,2}}^{2}.
\end{align*}

\end{lemma}

\begin{proof}
Since the second inequality is derived from the first and the fact that
$\left\Vert \nabla u\right\Vert _{W^{2,2}}\leq C\left\Vert \omega\right\Vert
_{W^{2,2}}$ we concentrate on the first. We use tools and ideas from the
classical deterministic theory, see for instance \cite{BKM1984,Kato,PL
Lions,MajdaBert}. We have%
\begin{align*}
&  \int_{\mathbb{T}^{3}}\Delta\mathcal{L}_{v}\omega\cdot\Delta\omega dx\\
&  =\int_{\mathbb{T}^{3}}\Delta\left(  v\cdot\nabla\omega\right)  \cdot
\Delta\omega dx+\int_{\mathbb{T}^{3}}\Delta\left(  \omega\cdot\nabla v\right)
\cdot\Delta\omega dx\\
&  =\int_{\mathbb{T}^{3}}\left(  \Delta v\cdot\nabla\omega\right)  \cdot
\Delta\omega dx+\int_{\mathbb{T}^{3}}\left(  v\cdot\nabla\Delta\omega\right)
\cdot\Delta\omega dx+2\int_{\mathbb{T}^{3}}\sum_{\alpha}\left(  \partial
_{\alpha}v\cdot\nabla\partial_{\alpha}\omega\right)  \cdot\Delta\omega dx\\
&  +\int_{\mathbb{T}^{3}}\left(  \Delta\omega\cdot\nabla v\right)  \cdot
\Delta\omega dx+\int_{\mathbb{T}^{3}}\left(  \omega\cdot\nabla\Delta v\right)
\cdot\Delta\omega dx+2\int_{\mathbb{T}^{3}}\sum_{\alpha}\left(  \partial
_{\alpha}\omega\cdot\nabla\partial_{\alpha}v\right)  \cdot\Delta\omega dx.
\end{align*}
The term $\int_{\mathbb{T}^{3}}\left(  v\cdot\nabla\Delta\omega\right)
\cdot\Delta\omega dx$ is equal to zero, being equal to $\frac{1}{2}%
\int_{\mathbb{T}^{3}}v\cdot\nabla\left\vert \Delta\omega\right\vert ^{2}dx$
which is zero after integration by parts and using $\operatorname{div}v=0$.
The terms
\[
2\int_{\mathbb{T}^{3}}\sum_{\alpha}\left(  \partial_{\alpha}v\cdot
\nabla\partial_{\alpha}\omega\right)  \cdot\Delta\omega dx+\int_{\mathbb{T}%
^{3}}\left(  \Delta\omega\cdot\nabla v\right)  \cdot\Delta\omega dx
\]
are immediately estimated by $C\left\Vert \nabla v\right\Vert _{L^{\infty}%
}\left\Vert \omega\right\Vert _{W^{2,2}}^{2}$. The term $\int_{\mathbb{T}^{3}%
}\left(  \omega\cdot\nabla\Delta v\right)  \cdot\Delta\omega dx$ is easily
estimated by $C\left\Vert \omega\right\Vert _{L^{\infty}}\left\Vert
\omega\right\Vert _{W^{2,2}}^{2}$. It remains to undestand the other two
terms. We have%
\begin{align*}
\left\vert \int_{\mathbb{T}^{3}}\left(  \Delta v\cdot\nabla\omega\right)
\cdot\Delta\omega dx\right\vert  &  \leq C\left\Vert \Delta v\cdot\nabla
\omega\right\Vert _{L^{2}}\left\Vert \omega\right\Vert _{W^{2,2}}\\
\left\vert \int_{\mathbb{T}^{3}}\left(  \partial_{\alpha}\omega\cdot
\nabla\partial_{\alpha}v\right)  \cdot\Delta\omega dx\right\vert  &  \leq
C\left\Vert \partial_{\alpha}\omega\cdot\nabla\partial_{\alpha}v\right\Vert
_{L^{2}}\left\Vert \omega\right\Vert _{W^{2,2}}\,.
\end{align*}
Hence, we only need to prove that
\begin{equation}
\left\Vert \partial_{\alpha}\omega\partial_{\beta}\partial_{\gamma
}v\right\Vert _{L^{2}}\leq C\left(  \left\Vert \nabla v\right\Vert
_{L^{\infty}}+\left\Vert \omega\right\Vert _{L^{\infty}}\right)  \left\Vert
\omega\right\Vert _{W^{2,2}} \label{interpolation}%
\end{equation}
for every $\alpha,\beta,\gamma=1,2,3$. Recall the followig particular case of
Gagliardo-Nirenberg interpolation inequality:%
\[
\left\Vert \partial_{\alpha}f\right\Vert _{L^{4}}^{2}\leq C\left\Vert
f\right\Vert _{\infty}\left\Vert f\right\Vert _{W^{2,2}}\,,
\]
which implies%
\[
\left\Vert \partial_{\alpha}\omega\partial_{\beta}\partial_{\gamma
}v\right\Vert ^{2}\leq C\left\Vert \omega\right\Vert _{\infty}\left\Vert
\omega\right\Vert _{W^{2,2}}\left\Vert \nabla v\right\Vert _{\infty}\left\Vert
\nabla v\right\Vert _{W^{2,2}}.
\]
Moreover, due to the relation between $v$ and $\omega$, we also have
$\left\Vert \nabla v\right\Vert _{W^{2,2}}\leq C\left\Vert \omega\right\Vert
_{W^{2,2}}$. Hence,%
\begin{align*}
\left\Vert \partial_{\alpha}\omega\partial_{\beta}\partial_{\gamma
}v\right\Vert  &  \leq C\left\Vert \omega\right\Vert _{\infty}^{1/2}\left\Vert
\omega\right\Vert _{W^{2,2}}^{1/2}\left\Vert \nabla v\right\Vert _{\infty
}^{1/2}\left\Vert \omega\right\Vert _{W^{2,2}}^{1/2}\\
&  \leq C\left\Vert \omega\right\Vert _{\infty}\left\Vert \omega\right\Vert
_{W^{2,2}}+C\left\Vert \nabla v\right\Vert _{\infty}\left\Vert \omega
\right\Vert _{W^{2,2}}%
\end{align*}
and inequality (\ref{interpolation}) has been proved. The proof of the lemma
is complete.
\end{proof}

\subsection{Estimates uniform in time}

\label{sect uniform in time}

{We introduce the following notations }%
\begin{align*}
\alpha_{t}  &  =\int_{\mathbb{R}^{3}}\left\vert \omega\left(  t,x\right)
\right\vert ^{2}dx+\int_{\mathbb{R}^{3}}\left\vert \Delta\omega\left(
t,x\right)  \right\vert ^{2}dx\\
M_{t}  &  =\int_{0}^{t}\sum_{k=1}^{\infty}\left(  \int_{\mathbb{R}^{3}%
}\mathcal{L}_{\xi_{k}}\omega\cdot\omega dx+\int_{\mathbb{R}^{3}}%
\Delta\mathcal{L}_{\xi_{k}}\omega\cdot\Delta\omega dx\right)  dB_{s}^{k}\,.
\end{align*}
Consequently, following from the estimates of the previous section and
assumption (\ref{assump}), we have
\[
\alpha_{t}\leq\alpha_{0}+2M_{t}+C_{R}\int_{0}^{t}\alpha_{s}ds
\]
so
\[
\sup_{s\in\lbrack0,t]}\alpha_{s}\leq e^{C_{R}t}(\alpha_{0}+2\sup_{s\in
\lbrack0,t]}|M_{s}|)
\]%
\begin{equation}
\mathbb{E}[\sup_{s\in\lbrack0,t]}\alpha_{s}^{2}]\leq4e^{C_{R}t}(\alpha_{0}%
^{2}+\mathbb{E}[\sup_{s\in\lbrack0,t]}|M_{s}|^{2}])\,. \label{a1}%
\end{equation}
By the Burkholder-Davis-Gundy inequality (see, e.g., Theorem 3.28, page 166 in \cite{KS}), we have that
\begin{equation}
\mathbb{E}[\sup_{s\in\lbrack0,t]}|M_{s}|^{2}]\leq K_{2}\mathbb{E}[[M]_{t}|],
\label{a2}%
\end{equation}
where $[M]$ is the quadratic variation of the local martingale $M$ and
\[
\lbrack M]_{t}=\sum_{k=1}^{\infty}\int_{0}^{t}\left(  \int_{\mathbb{R}^{3}%
}\mathcal{L}_{\xi_{k}}\omega\cdot\omega dx+\int_{\mathbb{R}^{3}}%
\Delta\mathcal{L}_{\xi_{k}}\omega\cdot\Delta\omega dx\right)  ^{2}ds\,.
\]

\begin{lemma}
\label{lemma 30} Under the assumption (\ref{Hp 3}), there is a constant $C>0$
such that
\[
\lbrack M]_{t}\leq C\alpha_{t}^{2}.
\]

\end{lemma}

\begin{proof}
Since $\mathcal{L}_{\xi_{k}}\omega=\xi_{k}\cdot\nabla\omega-\omega\cdot
\nabla\xi_{k}$, we have%
\[
\Delta\mathcal{L}_{\xi_{k}}=\xi_{k}\cdot\nabla\Delta\omega+R_{k}\omega
\]
where $R_{k}\omega$ contains several terms, each one with at most second
derivatives of $\omega$. Since $\int_{\mathbb{R}^{3}}\xi_{k}\cdot\nabla
\Delta\omega\Delta\omega dx=0$, we deduce
\[
\left\vert \int_{\mathbb{R}^{3}}\mathcal{L}_{\xi_{k}}\omega\cdot\omega
dx+\int_{\mathbb{R}^{3}}\Delta\mathcal{L}_{\xi_{k}}\omega\cdot\Delta\omega
dx\right\vert \leq C_{k}\alpha_{s}%
\]
for some constant $C_{k}>0$. Hence%
\[
\lbrack M]_{t}=\sum_{k=1}^{\infty}\int_{0}^{t}\left(  \int_{\mathbb{R}^{3}%
}\mathcal{L}_{\xi_{k}}\omega\cdot\omega dx+\int_{\mathbb{R}^{3}}%
\Delta\mathcal{L}_{\xi_{k}}\omega\cdot\Delta\omega dx\right)  ^{2}ds\leq
\sum_{k=1}^{\infty}C_{k}^{2}\int_{0}^{t}\alpha_{s}^{2}ds.
\]
With a few more computations it is possible to show that%
\[
C_{k}\leq C\left\Vert \xi_{k}\right\Vert _{W^{3,2}}.
\]
Hence we use assumption (\ref{Hp 3}).
\end{proof}

From Lemma \ref{lemma 30}, we deduce%
\begin{equation}
\mathbb{E}[[M]_{t}|]\leq C\int_{0}^{t}\mathbb{E}[\sup_{r\in\lbrack0,s]}%
\alpha_{r}^{2}]dr\, \label{a3}%
\end{equation}
and thus finally from \eqref{a1}, \eqref{a2} and \eqref{a3} and Gr\"{o}nwall's
inequality we obtain
\[
\mathbb{E}[\sup_{s\in\lbrack0,t]}\alpha_{s}^{2}]\leq C \,,
\]
independently of $\epsilon>0$. This proves bound (\ref{bound 1}) and completes
the necessary \textit{a priori} bounds, modulo the estimates of the next section.

\subsection{Bounds on Lie derivatives}

\label{sect basic bounds}

Recall the notation $\mathcal{L}_{\xi_{k}},k=1,...,\infty$ for the (first
order) operators $\mathcal{L}_{\xi_{k}}\omega=\left[  \xi_{k},\omega\right]
,~~k=1,...,\infty$.

\begin{lemma}
Inequality (\ref{striking est 1}) holds for every vector field $f$ of class
$W^{2,2}$.
\end{lemma}

\begin{proof}
\textbf{Step 1}. We have
\begin{align*}
\mathcal{L}_{\xi_{k}}^{\ast}  &  =-\mathcal{L}_{\xi_{k}}+S_{2} ,\\
\mathcal{L}_{\xi_{k}}S_{2}  &  =S_{2}\mathcal{L}_{\xi_{k}}-S_{4},
\end{align*}
where $S_{2}$ and $S_{4}$ are certain zero order operators (see below for a
proof). We have
\begin{align*}
\langle\left[  \xi_{k},f\right]  ,\left[  \xi_{k},f\right]  \rangle
+\langle\left[  \xi_{k},\left[  \xi_{k},f\right]  \right]  ,f\rangle &
=\left\langle \mathcal{L}_{\xi_{k}}f,\mathcal{L}_{\xi_{k}}f\right\rangle
+\left\langle \mathcal{L}_{\xi_{k}}^{2}f,f\right\rangle \\
&  =\left\langle \mathcal{L}_{\xi_{k}}f,\mathcal{L}_{\xi_{k}}f\right\rangle
+\left\langle \mathcal{L}_{\xi_{k}}f,\mathcal{L}_{\xi_{k}}^{\ast
}f\right\rangle \\
&  =\left\langle \mathcal{L}_{\xi_{k}}f,\mathcal{L}_{\xi_{k}}f\right\rangle
-\left\langle \mathcal{L}_{\xi_{k}}f,\mathcal{L}_{\xi_{k}}f\right\rangle
+\left\langle \mathcal{L}_{\xi_{k}}f,S_{2}f\right\rangle \\
&  =\left\langle \mathcal{L}_{\xi_{k}}f,S_{2}f\right\rangle .
\end{align*}
However, since $\left\langle f,S_{2}f^{\prime}\right\rangle =\left\langle
S_{2}f,f^{\prime}\right\rangle $ for any $f,f^{\prime}$ two square integrable
vector fields (see below for a proof)
\begin{align*}
\left\langle \mathcal{L}_{\xi_{k}}f,S_{2}f\right\rangle =\left\langle
f,\mathcal{L}_{\xi_{k}}^{\ast}S_{2}f\right\rangle  &  =-\left\langle
f,\mathcal{L}_{\xi_{k}}S_{2}f\right\rangle +\left\langle f,S_{2}%
^{2}f\right\rangle \\
&  =-\left\langle f,S_{2}\mathcal{L}_{\xi_{k}}f\right\rangle +\left\langle
f,S_{4}f\right\rangle +\left\langle f,S_{2}^{2}f\right\rangle \\
&  =-\left\langle S_{2}f,\mathcal{L}_{\xi_{k}}f\right\rangle +\left\langle
f,S_{4}f\right\rangle +\left\langle f,S_{2}^{2}f\right\rangle
\end{align*}
Hence
\[
\langle\left[  \xi_{k},f\right]  ,\left[  \xi_{k},f\right]  \rangle
+\langle\left[  \xi_{k},\left[  \xi_{k},f\right]  \right]  ,f\rangle
=\left\langle \mathcal{L}_{\xi_{k}}f,S_{2}f\right\rangle ={\frac{1}{2}%
}\left\langle f,(S_{2}^{2}+S_{4})f\right\rangle
\]

\textbf{Step 2}. Now we prove that $\mathcal{L}_{\xi_{k}}^{\ast}%
=-\mathcal{L}_{\xi_{k}}+S_{2}$ and that\textbf{ }$\left\langle f,S_{2}%
f^{\prime}\right\rangle =\left\langle S_{2}f,f^{\prime}\right\rangle $ for any
two square integrable vector fields $f,f^{\prime}$. We also have by
integration by parts and using $\nabla\cdot\xi_{k}=0$ that
\begin{align*}
\left\langle \mathcal{L}_{\xi_{k}}f,f^{\prime}\right\rangle  &  =\sum_{i}%
\int_{\mathbb{R}^{3}}(\mathcal{L}_{\xi_{k}}f)_{i}(x)f_{i}^{\prime}(x)dx^{3}\\
&  =\sum_{i}\sum_{j}\int_{\mathbb{R}^{3}}(\xi_{k}^{j}\partial_{j}f_{i}%
-f_{j}\partial_{j}\xi_{k}^{i})(x)f_{i}^{\prime}(x)dx^{3}\\
&  =\sum_{i}\sum_{j}\left(  \int_{\mathbb{R}^{3}}(-\xi_{k}^{j}\partial
_{j}f_{i}^{\prime})(x)f_{i}(x)dx^{3}-\int_{\mathbb{R}^{3}}(\partial_{j}\xi
_{k}^{i})(x)f_{j}(x)f_{i}^{\prime}(x)dx^{3}\right) \\
&  =-\sum_{i}\sum_{j}\left(  \int_{\mathbb{R}^{3}}(\xi_{k}^{j}\partial
_{j}f_{i}^{\prime}-f_{j}^{\prime}\partial_{j}\xi_{k}^{i})(x)f_{i}%
(x)dx^{3}+\int_{\mathbb{R}^{3}}(\partial_{i}\xi_{k}^{j}+\partial_{j}\xi
_{k}^{i})(x)f_{j}(x)f_{i}^{\prime}(x)dx^{3}\right) \\
&  =-\left\langle f,\mathcal{L}_{\xi_{k}}f^{\prime}\right\rangle +\left\langle
f,S_{2}f^{\prime}\right\rangle ,
\end{align*}
where
\[
\left\langle f,S_{2}f^{\prime}\right\rangle =\left\langle S_{2}f,f^{\prime
}\right\rangle =-\sum_{i}\sum_{j}\int_{\mathbb{R}^{3}}(\partial_{i}\xi_{k}%
^{j}+\partial_{j}\xi_{k}^{i})(x)f_{j}(x)f_{i}^{\prime}(x)dx^{3}.
\]

\textbf{Step 3}. Finally we prove that\textbf{ }$\mathcal{L}_{\xi_{k}}%
S_{2}=S_{2}\mathcal{L}_{\xi_{k}}-S_{4}$. We have $\left(  S_{2}f\right)
_{i}\left(  x\right)  =\sum_{j}a_{ij}\left(  x\right)  f_{j}\left(  x\right)
,$where $a_{ij}\left(  x\right)  =a_{ji}\left(  x\right)  =-\left(
\partial_{i}\xi_{k}^{j}+\partial_{j}\xi_{k}^{i}\right)  \left(  x\right)
.$Then%
\begin{align*}
\left(  \mathcal{L}_{\xi_{k}}S_{2}f\right)  _{i}  &  =\sum_{j}(\xi_{k}%
^{j}\partial_{j}\left(  S_{2}f\right)  _{i}-\left(  S_{2}f\right)
_{j}\partial_{j}\xi_{k}^{i})(x)\\
&  =\sum_{j}\sum_{l}(\xi_{k}^{j}\partial_{j}\left(  a_{il}f_{l}\right)
-a_{jl}f_{l}\partial_{j}\xi_{k}^{i})(x)\\
&  =\sum_{j}\sum_{l}(\xi_{k}^{j}a_{il}\partial_{j}f_{l}+\xi_{k}^{j}%
f_{l}\partial_{j}a_{il}-a_{jl}f_{l}\partial_{j}\xi_{k}^{i})(x)\\
&  =\sum_{j}\sum_{l}(\xi_{k}^{j}a_{il}\partial_{j}f_{l})\left(  x\right)
+\sum_{l}b_{il}f_{l}\left(  x\right)  ,
\end{align*}
where $b_{il}=\sum_{j}(\xi_{k}^{j}\partial_{j}a_{il}-a_{jl}\partial_{j}\xi
_{k}^{i})$. Similarly, we find that%
\begin{align*}
\left(  S_{2}\mathcal{L}_{\xi_{k}}f\right)  _{i}  &  =\sum_{l}a_{il}\left(
x\right)  \left(  \mathcal{L}_{\xi_{k}}f\right)  _{l}\\
&  =\sum_{j}\sum_{l}a_{il}\left(  x\right)  (\xi_{k}^{j}\partial_{j}%
f_{l}-f_{j}\partial_{j}\xi_{k}^{l})(x)\\
&  =\sum_{j}\sum_{l}\left(  \xi_{k}^{j}a_{il}\partial_{j}f_{l}\right)
(x)-\sum_{j}c_{ij}f_{j}\left(  x\right)  ,
\end{align*}
where $c_{ij}=\sum_{l}a_{il}\left(  x\right)  \partial_{j}\xi_{k}^{l}.$ Hence
$\left(  S_{4}f\right)  _{i}\equiv\left(  S_{2}\mathcal{L}_{\xi_{k}}f\right)
_{i}-\left(  \mathcal{L}_{\xi_{k}}S_{2}f\right)  _{i}=-\sum_{l}(b_{il}%
+c_{il})f_{l}$.
\end{proof}

\begin{remark}
\label{remark summability Ck}From this computation one can easily deduce that
\begin{align*}
\left\vert a_{ij}\right\vert _{\infty}  &  \leq2\left\vert \left\vert
\nabla\xi_{k}\right\vert \right\vert _{\infty},\\
\left\vert b_{ij}\right\vert _{\infty}  &  \leq6\left(  \left\vert \left\vert
\xi_{k}\right\vert \right\vert _{\infty}\left\vert \left\vert \xi
_{k}\right\vert \right\vert _{2,\infty}+\left\vert \left\vert \nabla\xi
_{k}\right\vert \right\vert _{\infty}^{2}\right)  ,\\
\left\vert c_{ij}\right\vert _{\infty}  &  \leq6\left\vert \left\vert
\nabla\xi_{k}\right\vert \right\vert _{\infty}^{2},\\
C_{k}^{(0)}  &  =c\left(  \left\vert \left\vert \xi_{k}\right\vert \right\vert
_{\infty}\left\vert \left\vert \Delta\xi_{k}\right\vert \right\vert _{\infty
}+\left\vert \left\vert \nabla\xi_{k}\right\vert \right\vert _{\infty}%
^{2}\right)  ,
\end{align*}
where $c$ is an independent constant ($c=48$). Therefore, the first of
assumptions (\ref{assump}) is fulfilled, provided%
\[
\sum_{k=1}^{\infty}\left(  \left\vert \left\vert \xi_{k}\right\vert
\right\vert _{\infty}\left\vert \left\vert \Delta\xi_{k}\right\vert
\right\vert _{\infty}+\left\vert \left\vert \nabla\xi_{k}\right\vert
\right\vert _{\infty}^{2}\right)  <\infty\,.
\]
The condition for the second assumption in (\ref{assump}) is similar.
\end{remark}

\begin{remark}
A typical example arises when $\xi_{k}$ are multiples of a complete orthonormal
system $\left\{  e_{k}\right\}  $ of $L^{2}$, namely $\xi_{k}=\lambda_{k}%
e_{k}$. In the case of the torus, if $e_{k}$ are associated to sine and cosine
functions, they are equi-bounded. Moreover, if instead of indexing with
$k\in\mathbb{N}$, we use $k\in\mathbb{Z}^{3}$, typically $\left\vert \nabla
e_{k}\right\vert \leq C\left\vert k\right\vert $ and $\left\vert \Delta
e_{k}\right\vert \leq C\left\vert k\right\vert ^{2}$. In such a case, the
previous condition becomes%
\[
\sum_{k\in\mathbb{Z}^{3}}\lambda_{k}^{2}\left\vert k\right\vert ^{2}<\infty
\,,\]
which is a verifiable condition.
\end{remark}

\begin{lemma}
Inequality (\ref{striking est 2}) holds for every vector field $f$ of class
$W^{4,2}$.
\end{lemma}

\begin{proof}
Let us define $S_{1}$ to be the following operator $S_{1}f:=\Delta
\mathcal{L}_{\xi_{k}}f-\mathcal{L}_{\xi_{k}}\Delta f$. By a direct
computation, we find that%
\begin{align*}
\left(  S_{1}f\right)  _{i}  &  :=\left(  \Delta\mathcal{L}_{\xi_{k}%
}f-\mathcal{L}_{\xi_{k}}\Delta f\right)  _{i}\\
&  =\sum_{j,l}\partial_{l}^{2}\left(  \xi_{k}^{j}\partial_{j}f_{i}%
-f_{j}\partial_{j}\xi_{k}^{i}\right)  -\left(  \xi_{k}^{j}\partial_{j}%
\partial_{l}^{2}f_{i}-\partial_{l}^{2}f_{j}\partial_{j}\xi_{k}^{i}\right) \\
&  =\sum_{j,l}\partial_{l}^{2}\xi_{k}^{j}\partial_{j}f_{i}+2\partial_{l}%
\xi_{k}^{j}\partial_{l}\partial_{j}f_{i}-f_{j}\partial_{l}^{2}\partial_{j}%
\xi_{k}^{i}-2\partial_{l}f_{j}\partial_{l}\partial_{j}\xi_{k}^{i},\\
&  =Af_{i}+B_{i}f\,.
\end{align*}
Consequently, $S_{1}$ is a second order operator, whose dominant part may be
expressed as
\[
A:=\sum_{j,l}2\partial_{l}\xi_{k}^{j}\partial_{l}\partial_{j}\,,
\]
where $B_{i}$ is a first order operator. Similarly, the computation
\[
\left(  \mathcal{L}_{\xi_{k}}f\right)  _{i}=\sum_{j}\left(  \xi_{k}%
^{j}\partial_{j}f_{i}-f_{j}\partial_{j}\xi_{k}^{i}\right)  = Cf_{i}-D_{i}f
\]
shows that $Cf_{i}-D_{i}f$ is a first order differential operation, whose
dominant part may be expressed as the operator%
\[
C:=\sum_{j}\xi_{k}^{j}\partial_{j}%
\]
and $D_{i}$ is a zero order operator. Let $S_{3}:=S_{1}\mathcal{L}_{\xi_{k}%
}-\mathcal{L}_{\xi_{k}}S_{1}$. Then, one computes
\begin{align*}
\left(  S_{3}f\right)  _{i}  &  =\left(  \left(  S_{1}\mathcal{L}_{\xi_{k}%
}-\mathcal{L}_{\xi_{k}}S_{1}\right)  f\right)  _{i}\\
&  =A\left(  \mathcal{L}_{\xi_{k}}f\right)  _{i}+B_{i}\left(  \mathcal{L}%
_{\xi_{k}}f\right)  -C\left(  S_{1}f\right)  _{i}+D_{i}\left(  S_{1}f\right)
\\
&  =ACf_{i}-AD_{i}f+B_{i}\left(  \mathcal{L}_{\xi_{k}}f\right)  -C\left(
Af_{i}+B_{i}f\right)  +D_{i}\left(  S_{1}f\right) \\
&  =\left(  AC-CA\right)  f_{i}+E_{i}f\,.
\end{align*}
We now note that both $\left(  AC-CA\right)  $ and $E_{i}$ are second order
operators. Consequently,
\begin{align*}
\left\langle \Delta\mathcal{L}_{\xi_{k}}^{2}f,\Delta f\right\rangle  &
=\left\langle \left(  \mathcal{L}_{\xi_{k}}\Delta+S_{1}\right)  \mathcal{L}%
_{\xi_{k}}f,\Delta f\right\rangle \\
&  =\left\langle \Delta\mathcal{L}_{\xi_{k}}f,\mathcal{L}_{\xi_{k}}^{\ast
}\Delta f\right\rangle +\left\langle S_{1}\mathcal{L}_{\xi_{k}}f,\Delta
f\right\rangle \\
&  =-\left\langle \Delta\mathcal{L}_{\xi_{k}}f,\mathcal{L}_{\xi_{k}}\Delta
f\right\rangle +\left\langle \Delta\mathcal{L}_{\xi_{k}}f,S_{2}\Delta
f\right\rangle +\left\langle S_{1}\mathcal{L}_{\xi_{k}}f,\Delta f\right\rangle
\\
&  =-\left\langle \Delta\mathcal{L}_{\xi_{k}}f,\Delta\mathcal{L}_{\xi_{k}%
}f\right\rangle +\left\langle \Delta\mathcal{L}_{\xi_{k}}f,S_{1}f\right\rangle
+\left\langle S_{1}\mathcal{L}_{\xi_{k}}f,\Delta f\right\rangle +\left\langle
\Delta\mathcal{L}_{\xi_{k}}f,S_{2}\Delta f\right\rangle .
\end{align*}
Hence%
\[
\left\langle \Delta\mathcal{L}_{\xi_{k}}^{2}f,\Delta f\right\rangle
+\left\langle \Delta\mathcal{L}_{\xi_{k}}f,\Delta\mathcal{L}_{\xi_{k}%
}f\right\rangle =\left\langle \Delta\mathcal{L}_{\xi_{k}}f,S_{1}f\right\rangle
+\left\langle S_{1}\mathcal{L}_{\xi_{k}}f,\Delta f\right\rangle +\left\langle
\Delta\mathcal{L}_{\xi_{k}}f,S_{2}\Delta f\right\rangle \,.
\]
Observe that
\begin{align}
\left\langle \Delta\mathcal{L}_{\xi_{k}}f,S_{1}f\right\rangle +\left\langle
S_{1}\mathcal{L}_{\xi_{k}}f,\Delta f\right\rangle  &  =\left\langle
\mathcal{L}_{\xi_{k}}\Delta f,S_{1}f\right\rangle +\left\langle S_{1}%
f,S_{1}f\right\rangle +\left\langle S_{1}\mathcal{L}_{\xi_{k}}f,\Delta
f\right\rangle \nonumber\\
&  =\left\langle \Delta f,\mathcal{L}_{\xi_{k}}^{\ast}S_{1}f\right\rangle
+\left\langle S_{1}\mathcal{L}_{\xi_{k}}f,\Delta f\right\rangle +\left\langle
S_{1}f,S_{1}f\right\rangle \nonumber\\
&  =-\left\langle \Delta f,\mathcal{L}_{\xi_{k}}S_{1}f\right\rangle
+\left\langle \Delta f,S_{2}S_{1}f\right\rangle +\left\langle S_{1}%
\mathcal{L}_{\xi_{k}}f,\Delta f\right\rangle +\left\langle S_{1}%
f,S_{1}f\right\rangle \nonumber\\
&  =\left\langle \Delta f,S_{3}f\right\rangle +\left\langle \Delta
f,S_{2}S_{1}f\right\rangle +\left\langle S_{1}f,S_{1}f\right\rangle \,.
\label{p1}%
\end{align}
The last term satisfies
\begin{align*}
\left\langle \Delta\mathcal{L}_{\xi_{k}}f,S_{2}\Delta f\right\rangle  &
=\left\langle \left(  \mathcal{L}_{\xi_{k}}\Delta+S_{1}\right)  f,S_{2}\Delta
f\right\rangle \\
&  =-\left\langle \Delta f,\mathcal{L}_{\xi_{k}}S_{2}\Delta f,\right\rangle
+\left\langle \Delta f,S_{2}^{2}\Delta f\right\rangle +\left\langle
S_{1}f,S_{2}\Delta f\right\rangle \\
&  =-\left\langle \Delta f,S_{2}\mathcal{L}_{\xi_{k}}\Delta f\right\rangle
+\left\langle \Delta f,S_{4}\Delta f\right\rangle +\left\langle \Delta
f,S_{2}^{2}\Delta f\right\rangle +\left\langle S_{1}f,S_{2}\Delta
f\right\rangle \\
&  =-\left\langle S_{2}\Delta f,\mathcal{L}_{\xi_{k}}\Delta f\right\rangle
+\left\langle \Delta f,S_{4}\Delta f\right\rangle +\left\langle \Delta
f,S_{2}^{2}\Delta f\right\rangle +\left\langle S_{1}f,S_{2}\Delta
f\right\rangle ,
\end{align*}
so that%
\begin{equation}
\left\langle S_{2}\Delta f,\Delta\mathcal{L}_{\xi_{k}}f\right\rangle =\frac
{1}{2}\left(  \left\langle S_{4}\Delta f,\Delta f\right\rangle +\left\langle
S_{2}^{2}\Delta f,\Delta f\right\rangle +\left\langle S_{2}\Delta
f,S_{1}f\right\rangle \right)  . \label{p2}%
\end{equation}
Since the terms in both expressions (\ref{p1}) and (\ref{p2}) can be
controlled by $\left\Vert f\right\Vert _{W^{2,2}}^{2}$, it follows that
indeed, there exists $C_{k}^{\left(  2\right)  }=C_{k}^{\left(  2\right)
}\left(  \left\vert \left\vert \xi_{k}\right\vert \right\vert _{2,\infty
}\right)  $ such that
\[
\left\vert \left\langle \Delta\mathcal{L}_{\xi_{k}}^{2}f,\Delta f\right\rangle
+\left\langle \Delta\mathcal{L}_{\xi_{k}}f,\Delta\mathcal{L}_{\xi_{k}%
}f\right\rangle \right\vert \leq C_{k}^{\left(  2\right)  }\left\Vert
f\right\Vert _{W^{2,2}}^{2}\,.
\]

\end{proof}

\subsection*{Acknowledgements}

We thank J. D. Gibbon for raising the crucial question of whether the stochastic Euler vorticity equations treated here would have a BKM theorem. We thank V. Resseguier and E. M\'emin for generously exchanging ideas with us about their own related work. 
We also thank S. Alberverio, A. Arnaudon, A. L. Castro, A. B. Cruzeiro, P. Constantin, F. Gay-Balmaz, G. Iyer, S. Kuksin, T. S. Ratiu, F. X. Vialard and many others who participated in the EPFL 2015 Program on Stochastic Geometric Mechanics and offered their valuable suggestions during the onset of this work. 
Finally, we also thank the anonymous referees for their constructive advice for improvements of the style and content of the original manuscript. 
During this work, DDH was partially supported by the European Research Council Advanced Grant 267382 FCCA and EPSRC Standard Grant EP/N023781/1. DC was partially supported by the EPSRC Standard Grant EP/N023781/1.

\appendix

%%%%%%%%%%%%%%%%%%%%%%%%%%%%%%%%%%%%
\section{Derivation of the stochastic Euler equations}

\label{sec-derivation-recap}

\paragraph{\bf Summary.}

$\,$\newline The generalisation from the classic deterministic Reynolds Transport Theorem
(RTT) for momentum to its Stratonovich stochastic version derived in this 
Appendix preserves the
geometric Lie derivative structure of the RTT. Specifically, the Lie derivative structure of
the Stratonovich stochastic RTT for the vector momentum density 
%${m_{j}}(x,t){e^{j}}_{i} (x)\,d^{3}x$ 
derived here turns out to be the same as the expression
appearing in the stochastic Kelvin circulation theorem derived in \cite{Holm2015}
from a stochastic version of Hamilton's variational principle for ideal fluid flows.
Thus, by combining the Stratonovich stochastic RTT with Newton's 2nd Law for fluid
dynamics, one recovers a known family of stochastic fluid equations. The simplest of
these is the 3D stochastic Euler fluid model, which was introduced in \cite{Holm2015}. This model  
was re-derived via multi-time homogenisation in \cite{CoGoHo2017}, and is 
re-derived once again here from Newton's Law, as it is the main subject of our present investigation.

\subsection{Review of the deterministic case}

\medskip\paragraph{\bf Newton's 2nd Law, Reynolds transport theorem, pullbacks and Lie
derivatives.}

$\,$

The fundamental equations of fluid dynamics derive from Newton's 2nd Law,
\begin{align}
\frac{d\mathbf{M}(t)}{dt} = \frac{d}{dt}\int_{\Omega(t)} \mathbf{m}\,d^{3}x =
\int_{\Omega(t)} \mathbf{F}\,d^{3}x = \boldsymbol{\mathcal{F}}(t) \,,
\label{IN2ndLaw-sketch}%
\end{align}
which sets the rate of change in time $t$ of the total momentum $\mathbf{M}%
(t)$ of a moving volume of fluid $\Omega(t)$ equal to the total
volume-integrated force $\boldsymbol{\mathcal{F}}$ applied on it; thereby
producing an equation whose solution determines the time dependent flow
$\eta_{t}$ governing $\Omega(t)=\eta_{t}\Omega(0)$.

The fluid flows $\eta_{t}$ considered here will be smooth invertible
time-dependent maps with a smooth inverses. Such maps are called
\textit{diffeomorphisms}, and are often simply referred to as \textit{diffeos}. One may regard the map $\eta_{t}$ as a time-dependent curve on the space of diffeos. The corresponding Lagrangian particle path of a fluid parcel is given
by the smooth, invertible, time-dependent map,
\begin{align}
\eta_{t}X=\eta(X,t) \in\mathbb{R}^{3} , \quad
\hbox{for initial reference position}\quad x(X,0)=\eta_{0}X=X .
\label{Lag-paths}%
\end{align}
This subsection deals with the deterministic derivation of the Eulerian ideal
fluid equations. So the map $\eta_{t}$ is deterministic here. The next
subsection will deal with parallel arguments for the stochastic version of
$\eta_{t}$ in equation \eqref{StochProc-intro}, and we will keep the same notation
for the diffeomorphisms in both subsections.

In standard notation from, e.g., Marsden and Hughes \cite{MaHu1994}, we may
write the $i$-th component of the total fluid momentum $M_{i}(t)$ in a
time-dependent domain of $\mathbb{R}^{3}$ denoted $\Omega(t)=\eta_{t}%
\Omega(0)$, as
\begin{align}
M_{i}(t)  &  = \int_{\Omega(t)} \mathbf{m}(x,t)\cdot\mathbf{e}_{i}(x) \,d^{3}x
= \int_{\Omega(t)} {m_{j}}(x,t) {e^{j}}_{i}(x) \,d^{3}x\label{Mom-1}\\
&  = \int_{\Omega(0)} \eta^{*}_{t}\Big(\mathbf{m}(x,t)\cdot\mathbf{e}_{i}(x)
\,d^{3}x \Big) = \int_{\Omega(0)} \eta^{*}_{t}\Big({m_{j}}(x,t) {e^{j}}_{i}(x)
\,d^{3}x \Big) . \label{Mom-2}%
\end{align}
Here the $\mathbf{e}_{i}(x)$ are coordinate basis vectors and the operation
$\eta^{*}_{t}$ denotes \textit{pullback} by the smooth time-dependent map
$\eta_{t}$. That is, the \textit{pullback} operation in the formulas above for
the total momentum ``pulls back'' the map $\eta_{t}$ through the functions in
the integrand. For example, in the fluid momentum density $\mathbf{m}(x,t)$ at
spatial position $x\in\mathbb{R}^{3}$ at time $t$, we have $\eta^{*}%
_{t}\mathbf{m}(x,t) = \mathbf{m}(\eta(X,t),t)$.

\medskip\paragraph{\bf Lie derivative.}

The time derivative of the pullback of $\eta_{t}$ for a scalar function
$\theta(x,t)$ is given by the chain rule as,
\begin{equation}
\frac{d}{dt}\eta_{t}^{*}\theta(x,t) = \frac{d}{dt}\theta(\eta_{t}X,t) =
\partial_{t}\theta(\eta(X,t),t) + \frac{\partial\theta}{\partial\eta^{j}}%
\frac{d\eta^{j}(X,t)}{dt} = \eta_{t}^{*}\left(  \partial_{t}\theta(x,t) +
\frac{\partial\theta}{\partial x^{j}}u^{j}(x,t) \right)  . \label{theta-dot}%
\end{equation}
The Eulerian velocity vector field $u(x,t)= u^{j}(x,t)\partial_{x^{j}}$ in
\eqref{theta-dot} generates the flow $\eta_{t}$ and is tangent to it at the
identity, i.e., at $t=0$. The time-dependent components of this velocity
vector field may be written in terms of the flow $\eta_{t}$ and its pullback
$\eta_{t}^{*}$ in several equivalent notations as, for example,
\begin{equation}
\frac{d\eta^{j}(X,t)}{dt} = u^{j}(\eta(X,t),t) = \eta_{t}^{*}u^{j}(x,t) =
u^{j}(\eta_{t}^{*}x,t) \,,\quad\hbox{or simply}\quad u = \dot{\eta}_{t}%
\,\eta_{t}^{-1} \,. \label{Eul-vel}%
\end{equation}
The calculation \eqref{theta-dot} also defines the \textit{Lie derivative
formula} for the scalar function $\theta$, namely \cite{MaHu1994}
\begin{equation}
\frac{d}{dt}\eta_{t}^{*}\theta(x,t) = \eta_{t}^{*}\big( \partial_{t}%
\theta(x,t) + \mathcal{L} _{u} \theta(x,t) \big) ,
\label{Lie-deriv-formula-PB}%
\end{equation}
where $\mathcal{L} _{u} $ denotes Lie derivative along the time-dependent
vector field $u=u^{j}(x,t)\partial_{x^{j}}$ with vector components
$u^{j}(x,t)$. In this example of a scalar function $\theta$, evaluating
formula \eqref{Lie-deriv-formula-PB} at time $t=0$ gives the standard
definition of Lie derivative of a scalar function $\theta(x)$ by a
time-independent vector field $u=u^{j}(x)\partial_{x^{j}}$, namely,
\begin{equation}
\frac{d}{dt}\Big|_{t=0}\eta_{t}^{*}\theta(x) = \mathcal{L} _{u}(x) \theta(x) =
u^{j}(x)\frac{\partial\theta(x)}{\partial x^{j}} \,.
\label{Lie-deriv-formula-scalar}%
\end{equation}

\begin{remark}
\textrm{To recap, in equations \eqref{Mom-1} and \eqref{Mom-2} for the total
momentum, the Eulerian spatial coordinate $x \in\mathbb{R}^{3}$ is fixed in
space, and the Lagrangian body coordinate $X \in\Omega(t)$ is fixed in the
moving body. The Lagrangian particle paths $\eta_{t}^{*}x=\eta(X,t)=\eta_{t}X
\in\mathbb{R}^{3}$ with $x(X,0)=\eta_{0}X=X$ may be regarded as time-dependent
maps from a reference configuration where points in the fluid are located at
$X$ to their current position $\eta_{t}^{*}x=\eta(X,t)$. Introducing the
pullback operation enables one to transform the integration in \eqref{Mom-1}
over the current fluid domain $\Omega(t)$ with moving boundaries into an
integration over the fixed reference domain $\Omega(0)$ in \eqref{Mom-2}. This
transformation allows the time derivative to be brought inside the integral
sign to act on the pullback of the integrand by the flow map $\eta_{t}$.
Taking the time derivatives inside the integrand then produces Lie derivatives
with respect to the vector field representing the flow velocity. }
\end{remark}

The coordinate basis vectors $\mathbf{e}_{i}(x)$ in \eqref{Mom-1} for the
moving domain and the corresponding basis vectors in the fixed reference
configuration $\mathbf{E}_{i}(X)$ are spatial gradients of the Eulerian and
Lagrangian coordinate lines in their respective domains.
%This means if the Eulerian coordinates are Euclidean, the Lagrangian coordinates in general will not be Euclidean.
The coordinate basis vectors $\mathbf{e}_{i}$ in the moving frame and
$\mathbf{E}_{i}$ in the fixed reference frame are related to each other by
contraction with the Jacobian matrix of the map $\eta_{t}$; namely,
\cite{MaHu1994}
\begin{equation}
\eta_{t}^{*}{e^{j}}_{i}(x) = \frac{\partial\eta^{j}(X,t)}{\partial X^{A}%
}{E^{A}}_{i}(X) = \left(  \eta_{t}^{*}\frac{\partial x^{j}}{\partial X^{A}%
}\right)  {E^{A}}_{i}(X) =: (\eta_{t}^{*}{J^{j}}_{A}){E^{A}}_{i}(X)\,.
\label{eE-relation}%
\end{equation}
As a consequence of the definition of Eulerian velocity in equation
\eqref{Eul-vel}, the defining relation \eqref{eE-relation} for $\eta_{t}%
^{*}{e^{j}}_{i}(x)$ implies the following evolution equation for the Eulerian
coordinate basis vectors,
\begin{align}%
\begin{split}
\frac{d}{dt}\big(\eta_{t}^{*}{e^{j}}_{i} (x)\big)  &  = \frac{\partial
}{\partial X^{A}}\frac{d\eta^{j}(X,t)}{dt}{E^{A}}_{i}(X)\\
&  = \eta_{t}^{*}\left(  \frac{\partial u^{j} }{\partial x^{k}} \frac{\partial
x^{k}}{\partial X^{A}}\right)  {E^{A}}_{i}(X)\\
&  = \eta_{t}^{*}\left(  \frac{\partial u^{j} }{\partial x^{k}} {J^{k}}%
_{A}\right)  {E^{A}}_{i}(X)\\
&  = \eta_{t}^{*}\left(  \frac{\partial u^{j} }{\partial x^{k}} {e^{k}}_{i}
\right)  \,.
\end{split}
\label{1-form-dyn}%
\end{align}

Likewise, the mass of each volume element will be conserved under the flow,
$\eta_{t}$. In terms of the pullback, this means
\begin{equation}
\eta_{t}^{*}\big(\rho(x,t)d^{3}x\big) = \rho(\eta(X,t),t) \det(J)\,d^{3}X =
\rho_{0}(X)\,d^{3}X \,, \label{mass-relation}%
\end{equation}
where the function $\rho(x,t)$ represents the current mass distribution in
Eulerian coordinates in the moving domain, and the function $\rho_{0}(X)$
represents the mass distribution in Lagrangian coordinates in the reference
domain at the initial time, $t=0$. Consequently, the time derivative of the
mass conservation relation \eqref{mass-relation} yields the \textit{continuity
equation} for the Eulerian mass density,
\begin{equation}
\frac{d}{dt}\Big(\eta_{t}^{*}\big(\rho(x,t)d^{3}x\big)\Big) = \eta_{t}%
^{*}\Big(\big(\partial_{t} \rho+ u^{j}\partial_{x^{j}}\rho+ \rho\,
\partial_{x^{j}}u^{j} \big) d^{3}x\Big) = \eta_{t}^{*}\big( (\partial_{t} +
\mathcal{L} _{u})(\rho\,d^{3}x) \big) = 0 \,, \label{continuity-eqn}%
\end{equation}
and again, as expected, the Lie derivative $\mathcal{L} _{u}$ appears. In this
example of a density, evaluating the formula \eqref{continuity-eqn} at time
$t=0$ gives the standard definition of Lie derivative of a density,
$\rho(x)\,d^{3}x$, by a time-independent vector field $u=u^{j}(x)\partial
_{x^{j}}$, namely,
\begin{equation}
\frac{d}{dt}\Big|_{t=0}\eta_{t}^{*}\big(\rho(x)\,d^{3}x\big) = \mathcal{L}
_{u(x)} \big(\rho(x)\,d^{3}x\big) = \mathrm{div}\big(\rho(x)\mathbf{u}(x)
\big)\,d^{3}x \,. \label{Lie-deriv-formula-density}%
\end{equation}

Next, we insert the mass conservation relation \eqref{mass-relation} into
equation \eqref{Mom-2} and introduce the covector $v_{j}(x,t):= m_{j}%
(x,t)/\rho(x,t)$ in order to distinguish between the momentum per unit mass
$v(x,t)$ and the velocity vector field $u(x,t)$ defined in \eqref{Eul-vel} for
the flow $\eta_{t}$, which transports the Lagrangian particles. In terms of
$v$, we may write the total momentum in \eqref{Mom-2} as
\begin{align}
M_{i}(t)  &  = \int_{\Omega(0)} \eta^{*}_{t} \left(  v_{j}(x,t){e^{j}}%
_{i}(x)\, \rho(x,t)\,d^{3}x \right)  \,. \label{Mom-3}%
\end{align}
Introducing the two transformation relations \eqref{eE-relation} for ${e^{j}%
}_{i}(x)$ and \eqref{mass-relation} for $\rho(x,t)\,d^{3}x$, yields
\begin{align}
M_{i}(t)  &  := \int_{\Omega(0)} \eta^{*}_{t}\left(  {v_{j}}(x,t)
\frac{\partial x^{j}}{\partial X^{A}} \right)  {E^{A}}_{i}(X) \rho
_{0}(X)\,d^{3}X\label{Mom-3a}\\
&  = \int_{\Omega(0)} \left(  {v_{j}}(x,t) \frac{\partial x^{j}}{\partial
X^{A}} \right)  \delta\big(x-\eta(X,t)\big) {E^{A}}_{i}(X) \rho_{0}(X)\,d^{3}X
\,, \label{Mom-3b}%
\end{align}
where in the last line we have inserted a delta function $\delta(x-\eta(X,t))$
for convenience in representing the pullback of a factor in the integrand to
the Lagrangian path.

\medskip\paragraph{\bf Newton's 2nd Law for fluids.}

We aim to explicitly write Newton's 2nd Law for fluids, which takes the form
\begin{align}
\frac{dM_{i}(t)}{dt} = \int_{\Omega(t)} \rho^{-1}F_{j} \,{e^{j}}_{i}(x)
\,\rho\,d^{3}x \,, \label{IN2ndLaw}%
\end{align}
for an assumed force density $F_{i} \,{e^{j}}_{i}(x) \,d^{3}x$ in a coordinate
system with basis vectors ${e^{j}}_{i}(x)$. To accomplish this, we of course
must compute the time derivative of the total momentum $M_{i}(t)$ in
\eqref{Mom-3a}. The result for the time derivative $dM_{i}(t)/dt$ is the
following,
\begin{align}
\frac{dM_{i}(t)}{dt}  &  = \int_{\Omega(0)} \eta^{*}_{t}\left(  \Big( \partial
_{t}{v_{j}}(x,t) + \frac{dx^{k}}{dt}v_{j,k} + v_{k} \frac{\partial}{\partial
x^{j}}\frac{dx^{k}}{dt} \Big) \frac{\partial x^{j}}{\partial X^{A}}\right)
{E^{A}}_{i}(X) \rho_{0}(X)\,d^{3}X \,. \label{Mom-4}%
\end{align}
Upon defining $u^{k}:=\frac{dx^{k}}{dt}$ (in a slight abuse of notation) and
using equations \eqref{eE-relation} and \eqref{mass-relation} this calculation
now yields
\begin{align}
\frac{dM_{i}(t)}{dt}  &  = \int_{\Omega(0)} \eta^{*}_{t}\bigg(\rho
(x,t)\Big( \partial_{t}{v_{j}}(x,t) + u^{k}\partial_{x^{k}}v_{j} + v_{k}
\partial_{x^{j}}u^{k} \Big) \,{e^{j}}_{i}(x) \,d^{3}x\bigg)\\
&  = \int_{\Omega(t)} \rho(x,t)\Big( \partial_{t}{v_{j}}(x,t) + u^{k}%
\partial_{x^{k}}v_{j} + v_{k} \partial_{x^{j}}u^{k} \Big) \,{e^{j}}_{i}(x)
\,d^{3}x \,. \label{Mom-5}%
\end{align}
Perhaps not unexpectedly, one may also deduce the Lie-derivative relation
\begin{align}
\frac{dM_{i}(t)}{dt}  &  = \int_{\Omega(t)} \rho(x,t)\Big( \partial_{t}{v_{j}%
}(x,t) + u^{k}\partial_{x^{k}}v_{j} + v_{k} \partial_{x^{j}}u^{k}
\Big) \,{e^{j}}_{i}(x) \,d^{3}x\label{Mom-circ}\\
&  = \int_{\Omega(t)} (\partial_{t} + \mathcal{L} _{u})\Big({v_{j}%
}(x,t)\,{e^{j}}_{i}(x) \,\rho(x,t)\,d^{3}x\Big) \,, \label{Mom-6}%
\end{align}
where, in the last step, we have applied the Lie derivative of the continuity
equation in \eqref{continuity-eqn}. We note that care must be taken in passing
to Euclidean spatial coordinates, in that one must first expand the spatial
derivatives of ${e^{j}}_{i}(x)$, before setting ${e^{j}}_{i}(x)=\partial
_{i}x^{j} = \delta^{j}_{i}$. One may keeping track of these basis vectors by
introducing a 1-form basis. Upon using the continuity equation
\eqref{continuity-eqn}, one may then write Newton's 2nd Law for fluids in
equation \eqref{IN2ndLaw} as a local 1-form expression,
\begin{align}
\big(\partial_{t}{v_{i}}(x,t) + u^{k}\partial_{k}v_{i} + v_{k} \partial
_{i}u^{k} \big)\,dx^{i} = (\partial_{t} + \mathcal{L} _{u})\big( {v_{i}}(x,t)
\,dx^{i} \big) = \rho^{-1}F_{i} \,dx^{i} . \label{IN2ndLaw-1form}%
\end{align}

\begin{remark}
[Distinguishing between $u$ and $v$]\label{u-v Diff} \rm
In formula
\eqref{IN2ndLaw-1form}, two quantities with the dimensions of velocity appear,
denoted as $u$ and $v$. The fluid velocity $u$ is a contravariant vector field
(with spatial component index up) which \emph{transports} fluid properties,
such as the mass density in the continuity equation \eqref{continuity-eqn}. In
contrast, the velocity $v$ is the \emph{transported} momentum per unit mass,
corresponding to a velocity 1-form $v_{i} dx^{i}$ (the circulation integrand
in Kelvin's theorem) and it is covariant (spatial component index down).

In general, these two velocities are different, they have different physical
meanings (velocity versus specific momentum) and they transform differently
under the diffeos. Mathematically, they are dual to each other, in the sense
that insertion (i.e., substitution) of the vector field $u$ into the 1-form
$v$ yields a real number, $u^{k}v_{k}$, where we sum repeated indices over
their range. Only in the case when the kinetic energy is given by the $L^{2}$
metric and the coordinate system is Cartesian with a Euclidean metric can the
components of the two velocities $u$ and $v$ be set equal to each other, as vectors.

And, as luck would have it, this special case occurs for the Euler fluid
equations in $\mathbb{R}^{3}$. Consequently, when we deal with the stochastic
Euler fluid equations in $\mathbb{R}^{3}$ in the later sections of the paper,
our notation will simplify, because we will not need to distinguish between
the two types of velocity $u$ and $v$. That is, in the later sections of the
paper, when stochastic Euler fluid equations are considered in $\mathbb{R}
^{3}$, the components of the velocities $u$ and $v$ will be the denoted by the same $\mathbb{R}^{3}$ vector, which we will choose to be $\mathbf{v}$.

\end{remark}

\medskip\paragraph{\bf Deterministic Kelvin circulation theorem.}

Formula \eqref{Mom-6} is the \textit{Reynolds Transport Theorem} (RTT) for a
momentum density. When set equal to an assumed force density, the RTT produces
Newton's 2nd Law for fluids in equation \eqref{IN2ndLaw-1form}. Further
applying equation \eqref{IN2ndLaw-1form} to the time derivative of the Kelvin
circulation integral $I(t)=\oint_{c(t)} {v_{j}}(x,t) \,dx^{j} $ around a
material loop $c(t)$ moving with Eulerian velocity $u(x,t)$, leads to
\cite{HoMaRa1998}
\begin{align}%
\begin{split}
\frac{dI(t)}{dt}  &  = \frac{d}{dt}\oint_{c(t)} \Big( {v_{j}}(x,t) \,dx^{j}
\Big) = \frac{d}{dt}\oint_{c(0)} \eta^{*}\Big( {v_{j}}(x,t) \,dx^{j} \Big)\\
&  = \int_{c(0)} \eta_{t}^{*}\Big( (\partial_{t} + \mathcal{L} _{u}%
)\big( {v_{j}}(x,t) \,dx^{j} \big)\Big)\\
&  = \oint_{c(0)} \eta_{t}^{*} \Big(\big( \partial_{t}{v_{j}}(x,t) +
u^{k}\partial_{x^{k}}v_{j} + v_{k} \partial_{x^{j}}u^{k} \big)\,dx^{j}\Big)\\
&  = \oint_{c(t)} (\partial_{t} + \mathcal{L} _{u})\Big( {v_{j}}(x,t) \,dx^{j}
\Big)\\
&  = \oint_{c(t)} \rho^{-1}F_{i} \,dx^{i} \,.
\end{split}
\label{RTT-Kel}%
\end{align}
Perhaps not surprisingly, the Lie derivative appears again, and the
line-element stretching term in the deterministic time derivative of the
Kelvin circulation integral in the third line of \eqref{RTT-Kel} corresponds
to the transformation of the coordinate basis vectors in the RTT formula
\eqref{Mom-circ}. Moreover, the last line of \eqref{RTT-Kel} follows directly
from the Newton 2nd Law for fluids in equation \eqref{IN2ndLaw-1form}.

\medskip\paragraph{\bf The deterministic Euler fluid motion equations.}

The simplest case comprises the deterministic Euler fluid motion equations for
incompressible, constant-density flow in Euclidean coordinates on
$\mathbb{R}^{3}$,
\begin{align}
\partial_{t}{u_{i}}(x,t) + u^{k}\partial_{k}u_{i} + u_{k} \partial_{i}u^{k} =
- \partial_{i}p \,, \quad\hbox{with}\quad\partial_{j}u^{j} = 0 \,,
\label{Euler-mot-eqn}%
\end{align}
for which the two velocities are the same and the only force is the gradient
of pressure, $p$.

Upon writing the Euler motion equation \eqref{Euler-mot-eqn} as a 1-form
relation in vector notation,
\begin{align}
(\partial_{t} + \mathcal{L} _{u})(\mathbf{u}\cdot d\mathbf{x}) = - dp\,,
\label{Euler-mot-eqn-1form}%
\end{align}
one easily finds the dynamical equation for the vorticity, $\boldsymbol{\omega
}=\mathrm{curl}\,\mathbf{u}$, by taking the exterior differential of
\eqref{Euler-mot-eqn-1form}, since $\boldsymbol{\omega}\cdot d\boldsymbol{S}=
d(\mathbf{u}\cdot d\mathbf{x})$ and the differential $d$ commutes with the Lie
derivative $\mathcal{L} _{u}$. Namely,
\begin{align}
\partial_{t} \boldsymbol{\omega} + (\mathbf{u}\cdot\nabla) \boldsymbol{\omega}
- (\boldsymbol{\omega}\cdot\nabla) \mathbf{u} = 0 \,.
\label{Euler-mot-eqn-calc}%
\end{align}
In terms of vector fields, this vorticity equation may be expressed
equivalently as
\begin{align}
\partial_{t}\omega+ \big[u\,,\,\omega\big] = 0\,, \label{Euler-mot-eqn-VFs}%
\end{align}
where $[u,\omega]$ is the commutator of vector fields.

%In this example of the circulation integrand in $I(t)$, evaluating formula \eqref{RTT-Kel} at time $t=0$ gives the standard definition of Lie derivative of a 1-form ${v_j}(x) \,dx^j$ by a time-independent vector field $u=u^j(x)\partial_{x^j}$, namely,
%\begin{equation}
%\frac{d}{dt}\Big|_{t=0}\eta_t^*\Big( {v_j}(x) \,dx^j \Big)
%= \mathcal{L}_u(x) \Big( {v_j}(x) \,dx^j \Big)
%= \big(u^k\partial_{x^k}v_j + v_k \partial_{x^j}u^k \big)\,dx^j
%\,.
%\label{Lie-deriv-formula-scalar}
%\end{equation}

%\begin{framed}
%\end{framed}

\subsection{Stochastic Reynolds Transport Theorem (SRTT) for Fluid Momentum}

For the stochastic counterpart of the previous calculation we replace $u =
\dot{\eta}_{t}\,\eta_{t}^{-1}$ written above in equation \eqref{Eul-vel} with
the \textit{Stratonovich} stochastic vector field
\begin{align}
\mathsf{dy}_{t}^{k} = u^{k}(x,t)dt + \sum_{i} \xi^{k}_{i}(x) \circ dB^{i}_{t}
= d\eta_{t}\,\eta_{t}^{-1} \,, \label{StochVF}%
\end{align}
where the $B_{t}^{i}$ with $i\in\mathbb{N}$ are scalar independent Brownian
motions. This vector field corresponds to the Stratonovich stochastic process
\begin{align}
\mathsf{d}\eta_{t}^{k}(X) = u^{k}(\eta_{t}(X),t)dt + \sum_{i} \xi^{k}_{i}%
(\eta_{t}(X)) \circ dB^{i}_{t} \,, \label{StochProc}%
\end{align}
where $\eta_{t}$ is a temporally stochastic curve on the diffeomorphisms. This means
that the time dependence of $\eta_{t}$ is rough, in that time derivatives do
not exist. However, being a diffeo, its spatial dependence is still smooth.

Consequently, upon following the corresponding steps for the deterministic
case leading to equation \eqref{Mom-5}, the Stratonovich stochastic version of
the deterministic RTT in equation \eqref{Mom-5} becomes
\begin{align}
\mathsf{d}M_{i}(t)  &  = \int_{\Omega(t)} \rho(x,t)\Big( \mathsf{d}{v_{j}%
}(x,t) + \mathsf{dy}_{t}^{k}\partial_{x^{k}}v_{j} + v_{k} \partial_{x^{j}%
}\mathsf{dy}_{t}^{k} \Big) {e^{j}}_{i}(x) \,d^{3}x\,. \label{StochRTT}%
\end{align}
We compare \eqref{StochRTT} with the Lie-derivative relation, cf. equation
\eqref{Mom-6},
\begin{align}
\Big( \mathsf{d}{v_{j}}(x,t) + \mathsf{dy}_{t}^{k}\partial_{x^{k}}v_{j} +
v_{k} \partial_{x^{j}}\mathsf{dy}_{t}^{k} \Big)\,dx^{j} = \big(\mathsf{dy} +
\mathcal{L} _{\mathsf{dy}_{t}})\big({v_{j}}(x,t) \,dx^{j}\big) \,.
\label{Mom-7}%
\end{align}
Here, $\mathcal{L} _{\mathsf{dy}_{t}} $ denotes Lie derivative along the
Stratonovich stochastic vector field $\mathsf{dy}_{t} =\mathsf{dy}_{t}%
^{j}(x,t)\partial_{x^{j}}$ with vector components $\mathsf{dy}_{t}^{j}(x,t)$
introduced in \eqref{StochVF}. We also introduce the stochastic versions of
the auxiliary equations \eqref{Lie-deriv-formula-scalar} for a scalar function
$\theta$ and \eqref{continuity-eqn} for a density $\rho\,d^{3}x$, and we
compare these two formulas with their equivalent stochastic Lie-derivative
relations,
\begin{align}
\mathsf{d}\big(\eta_{t}^{*}\theta(x,t)\big)  &  = \eta_{t}^{*}\big( \mathsf{d}
\theta(x,t) + \mathcal{L} _{\mathsf{dy}_{t}} \theta(x,t)
\big) ,\label{Stoch-Lie-deriv-formula}\\
\mathsf{d}\big(\eta_{t}^{*}(\rho(x,t)d^{3}x)\big)  &  = \eta_{t}%
^{*}\Big(\big(\mathsf{d} \rho+ \partial_{x^{j}}\big(\rho\, \mathsf{dy}_{t}%
^{j}(x,t)\big) \big) d^{3}x\Big) = \eta_{t}^{*}\big( ( \mathsf{d} +
\mathcal{L} _{\mathsf{dy}_{t}})(\rho\,d^{3}x) \big) = 0 \,.
\label{Stoch-continuity-eqn}%
\end{align}

\medskip\paragraph{\bf Stochastic Newton's 2nd Law for fluids.}

The stochastic Newton's 2nd Law for fluids will take the form
\begin{align}
\mathsf{d}M_{i}(t)= \int_{\Omega(t)} \rho^{-1}F_{i} \,{e^{j}}_{i}(x)
\,\rho\,d^{3}x \,, \label{IN2ndLaw-Stoch}%
\end{align}
for an assumed force density $F_{i} \,{e^{j}}_{i}(x) \,d^{3}x$ in a coordinate
system with basis vectors ${e^{j}}_{i}(x)$. Because of the stochastic RTT in
\eqref{StochRTT}, the expression \eqref{Mom-7} in a 1-form basis and the mass
conservation law for the stochastic flow in \eqref{Stoch-continuity-eqn}, one
may write the stochastic Newton's 2nd Law for fluids in equation
\eqref{IN2ndLaw-Stoch} as a 1-form relation,
\begin{align}
\Big( \mathsf{d}{v_{j}}(x,t) + \mathsf{dy}_{t}^{k}\partial_{x^{k}}v_{j} +
v_{k} \partial_{x^{j}}\mathsf{dy}_{t}^{k} \Big)\,dx^{j} = \big(\mathsf{d} +
\mathcal{L} _{\mathsf{dy}_{t}})\big({v_{j}}(x,t) \,dx^{j}\big) = \rho
^{-1}F_{j} \,dx^{j} \,dt\,. \label{IN2ndLaw-Stoch-1form}%
\end{align}

\medskip\paragraph{\bf Stochastic Kelvin circulation theorem.}

The stochastic Newton's 2nd Law for fluids in the 1-form basis in
\eqref{IN2ndLaw-Stoch-1form} introduces the line-element stretching term
previously seen in the stochastic Kelvin circulation theorem in
\cite{Holm2015}.

\begin{proof}
Inserting relations \eqref{StochRTT} and \eqref{Mom-7} for the stochastic RTT
into the Kelvin circulation integral $I(t)=\oint_{c(t)} {v_{j}}(x,t) \,dx^{j}
$ around a material loop $c(t)$ moving with stochastic Eulerian vector field
$\mathsf{dy}_{t}$ in \eqref{StochVF}, leads to the following, cf.
\cite{Holm2015},
\begin{align}%
\begin{split}
\mathsf{d}I(t)  &  = \mathsf{d}\oint_{c(t)} \Big( {v_{j}}(x,t) \,dx^{j}
\Big) = \mathsf{d}\oint_{c(0)} \eta^{*}\Big( {v_{j}}(x,t) \,dx^{j} \Big)\\
&  = \int_{c(0)} \eta_{t}^{*}\Big( (\mathsf{d} + \mathcal{L} _{\mathsf{dy}%
_{t}})\big( {v_{j}}(x,t) \,dx^{j} \big)\Big)\\
&  = \oint_{c(0)} \eta_{t}^{*} \Big(\big( \mathsf{d} {v_{j}}(x,t) +
\mathsf{dy}_{t}^{k}\partial_{x^{k}}v_{j} + v_{k} \partial_{x^{j}}%
\mathsf{dy}_{t}^{k} \big)\,dx^{j}\Big)\\
&  = \int_{c(t)} ( \mathsf{d} + \mathcal{L} _{\mathsf{dy}_{t}})\Big( {v_{j}%
}(x,t) \,dx^{j} \Big) \,.
\end{split}
\label{RTT-Kel-stoch}%
\end{align}
Substituting the stochastic Newton's 2nd Law for fluids in the 1-form basis in
\eqref{IN2ndLaw-Stoch-1form} into the last formula in \eqref{RTT-Kel-stoch}
yields the stochastic Kelvin circulation theorem in the form of
\cite{Holm2015}, namely,
\begin{align}
\mathsf{d}I(t)  &  = \mathsf{d}\oint_{c(t)} {v_{j}}(x,t) \,dx^{j} =
\oint_{c(t)} \rho^{-1}F_{j} \,dx^{j} \,dt\,. \label{Kel-stoch}%
\end{align}

\end{proof}

\begin{remark}
\textrm{As we have seen, the development of stochastic fluid dynamics models
revolves around the choice of the forces appearing in Newton's 2nd Law
\eqref{IN2ndLaw-Stoch} and Kelvin's circulation theorem \eqref{Kel-stoch}. For
examples in stochastic turbulence modelling using a variety of choices of
these forces, see \cite{Me2014,Re2017}, whose approaches are the closest to
the present work that we have been able to identify in the literature. }
\end{remark}

\medskip\paragraph{\bf The stochastic Euler fluid motion equations in three dimensions.}

The simplest 3D case comprises the stochastic Euler fluid motion equations for
incompressible, constant-density flow in Euclidean coordinates on
$\mathbb{R}^{3}$ which was introduced and studied in \cite{Holm2015}. These
equations are given in \eqref{RTT-Kel-stoch} by
\begin{align}
\mathsf{d}{v_{i}}(x,t) + \mathsf{d}y_{t}^{k}\partial_{k}v_{i} + v_{k}
\partial_{i}\mathsf{d}y_{t}^{k} = - \partial_{i}p\,dt \,, \quad
\hbox{with}\quad\partial_{i}(\mathsf{d}y_{t}^{i}) = 0 \,,
\label{Euler-mot-eqn-stoch}%
\end{align}
in which the stochastic transport velocity ($\mathsf{dy}_{t}$) corresponds to
the vector field in \eqref{StochVF}, the only force is the gradient of
pressure, $p$, and the density $\rho$ is taken to be constant.

The transported momentum per unit mass with components $v_{j}$, with
$j=1,2,3,$ appears in the circulation integrand in \eqref{Kel-stoch} as
$v_{j}\,dx^{j} = \mathbf{v}\cdot d\mathbf{x}$. The 3D stochastic Euler motion
equation \eqref{Euler-mot-eqn-stoch} may be written equivalently by using
\eqref{RTT-Kel-stoch} as a 1-form relation
\begin{align}
(\mathsf{d} + \mathcal{L} _{\mathsf{dy}_{t}})(\mathbf{v}\cdot d\mathbf{x}) =
-\,dp\,dt\,, \label{Euler-mot-eqn-1form-stoch}%
\end{align}
where we recall that $(\mathsf{d})$ denotes the stochastic evolution operator,
while $(d)$ denotes the spatial differential. We may derive the stochastic
equation for the vorticity 2-form, defined as
\[
\boldsymbol{\omega}\cdot d\boldsymbol{S} := d(\mathbf{v}\cdot d\mathbf{x}) =
(v_{j,k} - v_{k,j})\,dx^{k}\wedge dx^{j} =: \omega_{jk}\,dx^{k}\wedge dx^{j}
=: \mathrm{curl}\,\mathbf{v} \cdot d\boldsymbol{S} \,,
\]
with $dx^{j}\wedge dx^{k}= - \, dx^{k}\wedge dx^{j}$, by taking the exterior
differential $(d)$ of \eqref{Euler-mot-eqn-1form-stoch} and then invoking the
two properties that (i) the spatial differential $d$ commutes with the Lie
derivative $\mathcal{L} _{\mathsf{dy}_{t}}$ of a differential form and (ii)
$d^{2}=0$, to find%

\begin{align}
0 = (\mathsf{d} + \mathcal{L} _{\mathsf{dy}_{t}}) (\boldsymbol{\omega} \cdot
d\boldsymbol{S}) = \Big(\mathsf{d} \boldsymbol{\omega} - \mathrm{curl}
\,(\mathsf{dy}_{t}\times\boldsymbol{\omega})\Big) \cdot d\boldsymbol{S} \,.
\label{Euler-vort-eqn-2form-stoch}%
\end{align}
In Cartesian coordinates, all of these quantities may treated
as divergence free vectors in $\mathbb{R}^{3}$, that is, $\nabla\cdot\mathbf{v}=0=\nabla\cdot \mathsf{dy}_{t}$. Consequently, equation \eqref{Euler-vort-eqn-2form-stoch} recovers the vector SPDE form of the 3D
stochastic Euler fluid vorticity equation \eqref{eq Euler Strat-intro},
\begin{align}
\mathsf{d} \boldsymbol{\omega} + (\boldsymbol{\mathsf{dy}}_{t}\cdot\nabla)
\boldsymbol{\omega} - (\boldsymbol{\omega}\cdot\nabla) \boldsymbol{\mathsf{dy}%
}_{t} = 0 \,. \label{Euler-vort-eqn-calc-stoch}%
\end{align}
In terms of volume preserving vector fields in $\mathbb{R}^{3}$, this vorticity equation may be expressed equivalently as
\begin{align}
\mathsf{d} \omega+ \big[\mathsf{dy}_{t}\,,\,\omega\big] = 0\,,
\label{Euler-vort-eqn-VFs-stoch}%
\end{align}
where $[\mathsf{dy}_{t}\,,\,\omega]$ is the commutator of vector fields,
$\mathsf{dy}_{t}:=\boldsymbol{\mathsf{dy}}_{t}\cdot\nabla$ and $\omega
:=\boldsymbol{\omega}\cdot\nabla$. Equation \eqref{Euler-vort-eqn-VFs-stoch} for the vector field $\omega$ implies
\begin{align}
\mathsf{d} (\eta_t^*\omega) 
= \eta_t^*(\mathsf{d}\omega + \mathcal{L} _{\mathsf{dy}_{t}}\omega) = 0\,,
\label{Euler-vort-eqn-PB-stoch}%
\end{align}
where $\mathcal{L} _{\mathsf{dy}_{t}}\omega=[\mathsf{dy}_{t}\,,\,\omega]$. In vector components, this implies the pullback relation
\begin{align}
\eta_t^*\big(\omega^j(x,t) \big)
=
\eta_t^*\left(\frac{\partial x^{j}}{\partial X^{A}}\right)
{\omega_0^{A}}(X)\,,
\quad\hbox{or}\quad
\omega^j(\eta(X,t),t) 
=
\left(\frac{\partial \eta^{j}(X,t)}{\partial X^{A}}\right)
{\omega_0^{A}}(X)\,,
\label{Cauchy-vort-soln-stoch}%
\end{align}
where $\omega_0^{A}(X)$ is the $A$-th Cartesian component of the initial vorticity, as a function of the Lagrangian spatial coordinates $X$ of the reference configuration at time $t=0$, and $\eta_t^*$ is the pullback by the stochastic process in \eqref{StochProc-intro}. 
Equation \eqref{Cauchy-vort-soln-stoch} is the stochastic generalization of Cauchy's 1827 solution for the vorticity of the deterministic Euler vorticity equation, in terms of the Jacobian of the Lagrange-to-Euler map. See \cite{FrVi2014} for a historical review of the role of Cauchy's relation in deterministic hydrodynamics. 

Foundational results for other SPDEs for hydrodynamics related to \eqref{Euler-vort-eqn-calc-stoch} can be found in \cite{Fl2011, FlaGat, FlMaNe2014} and references therein.

\end{document}